\documentclass[aps,pra,twocolumn,superscriptaddress]{revtex4-2}

\usepackage{amsmath,amsfonts}
\usepackage[caption=false,font=normalsize,labelfont=sf,textfont=sf]{subfig}

\usepackage[utf8]{inputenc}  
\usepackage[T1]{fontenc}    
\usepackage{hyperref}       
\usepackage{url}          
\usepackage{booktabs}        
\usepackage{amsfonts}        
\usepackage{nicefrac}        
\usepackage{microtype}       
\usepackage{bm}
\usepackage{braket}
\usepackage{graphicx}
\usepackage{amsthm}
\usepackage{amssymb}
\usepackage{amsmath}
\usepackage{bbm}
\usepackage{lipsum,mwe}

\usepackage[ruled,linesnumbered]{algorithm2e}

\usepackage{appendix}

\usepackage{adjustbox}
\usepackage[normalem]{ulem}
\newtheorem{theorem}{Theorem}
\newtheorem{lemma}{Lemma}

\newtheorem{assu}{Assumption}
\newtheorem{prop}{Proposition}

\newtheorem{prob}{Problem}

\newtheorem*{theorem-non}{Theorem}
\usepackage{mathtools}
\usepackage{bbm}
\usepackage[explicit]{titlesec}
\usepackage{nccmath}
\usepackage{xcolor}
\usepackage{makecell}
\usepackage{colortbl}

\usepackage{tkz-tab}
\usepackage{subcaption}
\usepackage{enumitem}
\usepackage{multirow}
\usepackage{caption}
\DeclareMathOperator{\C}{C}
\DeclareMathOperator{\Q}{Q}
\DeclareMathOperator{\RQ}{RQ}

\DeclareMathOperator{\Tr}{Tr}

\DeclareMathOperator{\Fnorm}{F}

\DeclareMathOperator{\haar}{Haar}
\DeclareMathOperator{\swap}{SWAP}

\DeclareMathOperator{\Rot}{Rot}
\DeclareMathOperator{\R}{R}
\DeclareMathOperator{\Gene}{Gene}
\DeclareMathOperator{\ERM}{ERM}

\allowdisplaybreaks[4]

\begin{document}
	
	\title{Separable Power of Classical and Quantum Learning Protocols Through the Lens of No-Free-Lunch Theorem}

    \author{Xinbiao~Wang}
    \affiliation{Institute of Artificial Intelligence, School of Computer Science, Wuhan University, Wuhan 430072, China}

    \author{Yuxuan~Du}
    \email{duyuxuan123@gmail.com}
    \affiliation{School of Computer Science and Engineering, Nanyang Technological University, Singapore 639798, Singapore}

    \author{Kecheng~Liu}
    \affiliation{Center on Frontiers of Computing Studies, Peking University, Beijing 100871, China}

    \author{Yong~Luo}
    \affiliation{Institute of Artificial Intelligence, School of Computer Science, Wuhan University, Wuhan 430072, China}

    \author{Bo~Du}
    \affiliation{Institute of Artificial Intelligence, School of Computer Science, Wuhan University, Wuhan 430072, China}

    \author{Dacheng~Tao}
    \affiliation{School of Computer Science and Engineering, Nanyang Technological University, Singapore 639798, Singapore}
	
	
\begin{abstract}  
    The No-Free-Lunch (NFL) theorem, which quantifies problem- and data-independent generalization errors regardless of the optimization process, provides a foundational framework for comprehending diverse learning protocols' potential. Despite its significance, the establishment of the NFL theorem for quantum machine learning models remains largely unexplored, thereby overlooking broader insights into the fundamental relationship between quantum and classical learning protocols. To address this gap, we categorize a diverse array of quantum learning algorithms into three learning protocols designed for learning quantum dynamics under a specified observable and establish their NFL theorem. The exploited protocols, namely Classical Learning Protocols (CLC-LPs), Restricted Quantum Learning Protocols (ReQu-LPs), and Quantum Learning Protocols (Qu-LPs), offer varying levels of access to quantum resources. Our derived NFL theorems demonstrate quadratic reductions in sample complexity across CLC-LPs, ReQu-LPs, and Qu-LPs, contingent upon the orthogonality of quantum states and the diagonality of observables. We attribute this performance discrepancy to the unique capacity of quantum-related learning protocols to indirectly utilize information concerning the global phases of non-orthogonal quantum states, a distinctive physical feature inherent in quantum mechanics. Our findings not only deepen our understanding of quantum learning protocols' capabilities but also provide practical insights for the development of advanced quantum learning algorithms.
\end{abstract}  
	
	\maketitle
	
	\section{Introduction}\label{sec:introduction}
    In the realm of artificial intelligence and machine learning, the utilization of vast datasets has become a cornerstone for advancing deep learning algorithms \cite{hastie2009elements, halevy2009unreasonable,lecun2015deep, sun2017revisiting, zhang2021understanding}. Recent breakthroughs in natural language processing and computer vision serve as compelling evidence of this trend \cite{brown2020language,ouyang2022training,bai2022training,touvron2023llama,zhao2023survey}. Despite the numerous merits of large-volume data, the crucial determinant of deep learning success lies in the ability to effectively extract intricate patterns and knowledge from the information-rich landscapes of collected data \cite{bishop2006pattern,krizhevsky2012imagenet,bengio2013representation,schmidhuber2015deep, wang2020deep}. To this end, various advanced learning strategies have been proposed to maximize information extraction. Concrete instances include contrastive learning for identifying the intrinsic invariance in image data \cite{purushwalkam2020demystifying,dangovski2021equivariant,li2023augmentation}, generative adversarial networks for extracting information by generating synthetic data samples that are remarkably similar to real data \cite{goodfellow2020generative,creswell2018generative}, and transfer learning for recognizing patterns over different tasks \cite{torrey2010transfer,zhuang2020comprehensive,zhu2023transfer}. However, deep learning models have encountered the issues of the scaling law \cite{hestness2017deep,kaplan2020scaling,henighan2020scaling,rosenfeld2021scaling}, which indicates that the limited computational power of classical computers results in diminishing returns in performance improvement as the size of the training dataset grows. Therefore, alternative approaches that enable the efficient processing of large datasets and facilitate the development of advanced AI models are highly demanded.

    Quantum machine learning (QML), which utilizes quantum computers to construct learning models, is heralded as a potential game-changer for artificial intelligence \cite{schuld2015introduction,biamonte2017quantum,ciliberto2018quantum,tian2023recent,zeguendry2023quantum}.  A primary objective within QML is to effectively address problems that are intractable for classical algorithms, thereby delineating the boundaries between quantum and classical learning models in terms of learning performance. In this regard, the formulation of rigorous theoretical frameworks to deeply understand the capabilities of QML models is crucial for advancing this field.  
    Existing literature has seen the emergence of sophisticated quantum learning algorithms with provable advantages for specific datasets and problems.
    For instance, quantum kernel methods have been shown to significantly lower the generalization error compared to their classical counterparts on synthetic datasets, as evidenced by studies of Ref.~\cite{huang2021power, liu2021rigorous, wang2021towards}.
     Furthermore, in the domain of quantum states learning, models exploiting quantum memory have demonstrated an exponential reduction in query complexity compared to that without employing quantum memory for learning some specific properties of unknown states \cite{huang2021information, chen2022exponential}. Despite the advancement, a broader question of a universal separation between classical and quantum learning approaches—ones that do not rely on particular datasets or problems—remains less understood. To further unlock the potential of QML and harness its universal power, it is imperative to extend our exploration beyond narrowly defined scenarios and delve into the fundamental distinctions and capabilities that quantum learning algorithms offer over classical ones.
 
    \begin{figure*}[htbp]
		\centering
		\includegraphics[width=0.98\textwidth]{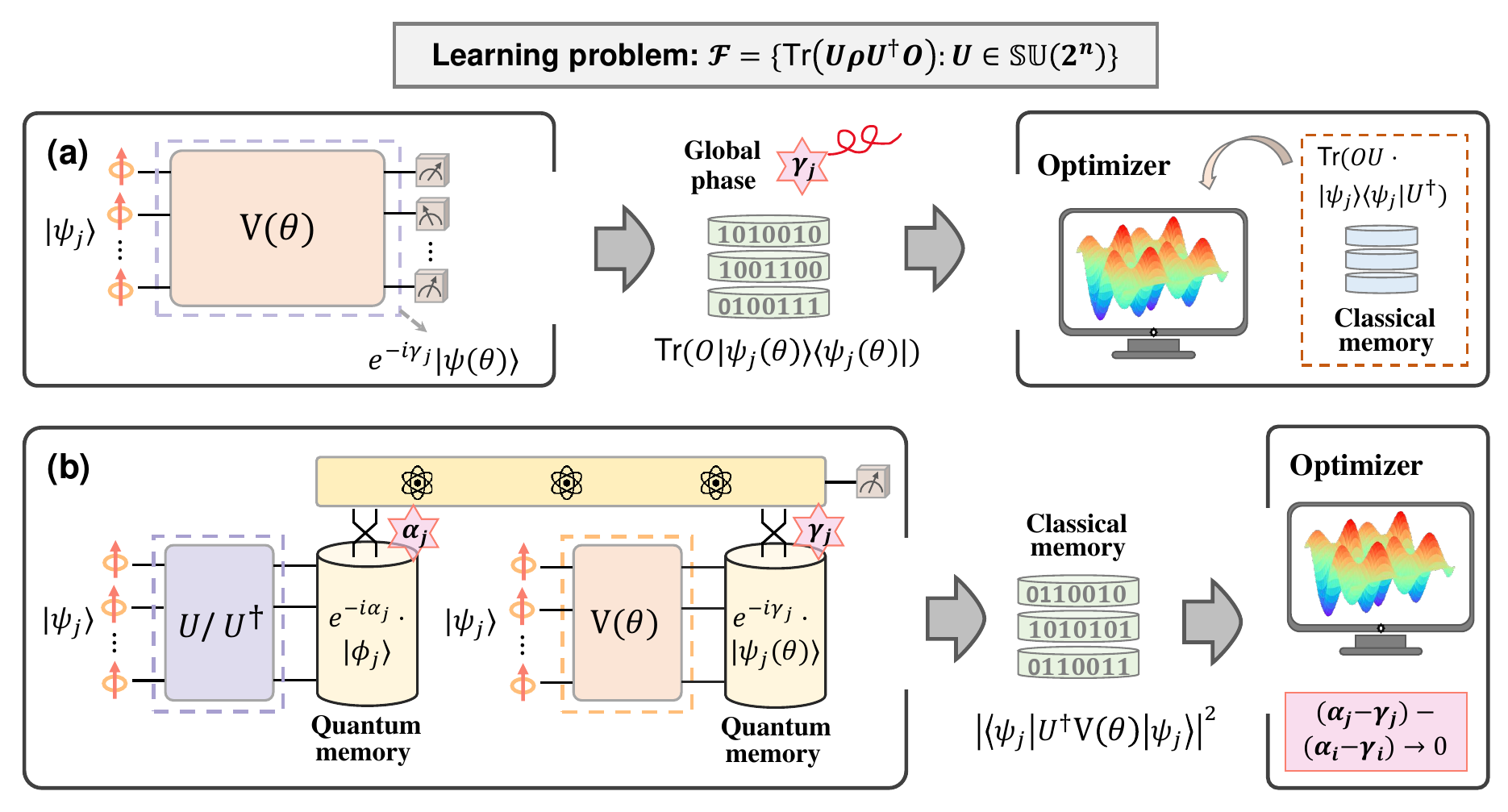}
		\caption{\small{\textbf{An overview of classical and (restricted) quantum learning protocols. 
				} Both the learning protocols shown in the upper panel (a) and the lower panel (b) aim to learn the target function $\Tr(U\rho U^{\dagger}O)$ for a given fixed observable $O$. \textbf{(a)} The classical learning protocols employ a tunable unitary $V(\bm{\theta})$ to evolve the input states $\ket{\bm{\psi}_j}$ to the output states $e^{-i\bm{\gamma}_j}\ket{\bm{\psi}_j(\bm{\theta})}:=V(\bm{\theta})\ket{\bm{\psi}_j}$ with $\bm{\gamma}_j$ being the global phase of states $V(\bm{\theta})\ket{\bm{\psi}_j}$, where $\bm{\theta}$ refers to tunable variables which could be discrete circuit structure or continuous parameters. A pre-defined observable $O$ is employed to measure the output states, leading to the information loss of the global phase $\bm{\gamma}_j$. An optimizer then is exploited to update the variables $\bm{\theta}$ according to the disparity of the measurement output and the target output $\Tr(OU\ket{\bm{\psi}_j}\bra{\bm{\psi}_j}U^{\dagger})$. \textbf{(b)} ReQu-LPs (or Qu-LPs) employ quantum operations to process the output states of the target unitary $U\ket{\bm{\psi}_j}$ (or $U^{\dagger}\ket{\bm{\psi}_j}$) and the tunable unitary $V(\bm{\theta})\ket{\bm{\psi}_j}$ in the same quantum system, where the output quantum states are stored in quantum memory. This allows for learning the phase information such that the phase difference $\bm{\alpha}_j-\bm{\gamma}_j$ over various training states converges to the same value as the optimization proceeds. }}
		\label{fig:scheme}
	\end{figure*}

    A pivotal metric for assessing the universal capabilities of a model is through the lens of the No Free Lunch (NFL) theorem, which is a fundamental concept in artificial intelligence that characterizes the capabilities of learning models \cite{wolpert1997no,ho2002simple, wolf2018mathematical}. One interpretation of the theorem is that the size of the training data determines the ultimate performance of models across different data types, regardless of the optimization procedure used. In particular, the NFL theorem considers the average learning performance across different training datasets and learning problems under the assumption of perfect training. This provides a viable solution for analyzing the data- and problem-independent power of quantum learning protocols. In this regard, establishing the NFL theorems for classical and quantum learning protocols on the same learning task not only contributes to understanding the power of the learning protocol itself, but also provides a fair manner for comparing the learning performance of various learning protocols on an equal footing. While initial endeavors have sought to formulate quantum NFL theorems specifically for compiling a quantum unitary \cite{poland2020no,sharma2022reformulation, wang2023transition}, they focused on the impact of entangled data on learning performance without considering the differences in various learning protocols (see Section~\ref{subsubsec:related_work_NFL} for details).
    
    \begin{table*}[ht]
		\centering
        \renewcommand{\arraystretch}{2.4}
		\caption{\small{\textbf{The achieved results of NFL theorems for various learning protocols.   
				} The notation $d=2^n$ refers to the dimension of an $n$-qubit unitary. $N$ is the size of the training dataset. $O$ and $\sum_{k\ne j}^d O_{kj}^2$ refers to the fixed observable and the square sum of the non-diagonal elements of $O$, respectively.}}
		\scalebox{1.10}{
			\begin{tabular}{c|c|c|c|c}
				\toprule
				\multicolumn{2}{c|}{Learning  Protocols}  & Classical & Restricted Quantum & Quantum \\
				\midrule
				\multirow{2}*{\makecell[l]{~ \\ Training \\States}} & Non-Orthogonal & \makecell[l]{$(d^2-N)/d^5 \times$ \\  $(d\Tr(O^2)-\Tr(O)^2)$} & \makecell[l]{$(d^2-N^2)/d^5 \times$ \\  $(d\Tr(O^2)-\Tr(O)^2)$} & \makecell[l]{$(d-N)/d^4 \times (d\Tr(O^2)-\Tr(O)^2)$ \\ $- N^2 \sum_{k\ne j}^d O_{kj}^2/d^4$} \\
                \cline{2-5}  
                & Orthogonal & \makecell[l]{$(d^2-N)/d^5\times$ \\  $(d\Tr(O^2)-\Tr(O)^2)$} & \makecell[l]{$(d^2-N)/d^5\times$ \\  $ (d\Tr(O^2)-\Tr(O)^2)$} & $(d-N)/d^4\times (d\Tr(O^2)-\Tr(O)^2)$ \\
				\bottomrule
			\end{tabular}
		}
		\label{table:NFL_bound}
	\end{table*}
 
	To address the identified knowledge gap of the universal power of QML models, we delve into exploring the task of learning an $n$-qubit quantum unitary $U$ under a known fixed observable $O$, which is the foundation of many practical tasks such as quantum state classification \cite{abohashima2020classification,li2022recent,du2023problem,qian2022dilemma}, quantum sensing \cite{degen2017quantum,alderete2022inference,coles2021pushing}, and quantum simulation \cite{georgescu2014quantum,endo2020variational,jones2019variational}. Many algorithms have been proposed to accomplish such tasks, such as variational quantum algorithms (VQAs) \cite{cerezo2021variational_VQA}, quantum process tomography with tensor networks \cite{torlai2023quantum}, and quantum circuit representation based on local inversions \cite{huang2024learning}.
    These algorithms are adaptable to different tasks by employing various learning protocols that depend on the level of access to quantum resources. For example, in quantum sensing tasks, direct access to the quantum system $U$ is possible, allowing for a straightforward learning protocol where learning models can directly interact with $U$. In contrast, in classification tasks, the learning model does not directly interact with $U$; instead, it utilizes classical data that encapsulates information about the target system $U$ for the learning process. Given the practical constraints on quantum resource usage, we categorize the mainstream algorithms into three distinct learning protocols: (1) \underline{C}Lassi\underline{C}al \underline{L}earning \underline{P}rotocols (CLC-LPs), which utilize measurement outputs from the target unitary $U$ for learning, as shown in Fig.~\ref{fig:scheme}(a); (2) \underline{Re}stricted \underline{Qu}antum \underline{L}earning \underline{P}rotocols (ReQu-LPs), which permit limited access to the target unitary during learning process; (3) \underline{Qu}antum \underline{L}earning \underline{P}rotocols (Qu-LPs), which allow access to the inverse of target unitary during learning process, as illustrated in Fig.~\ref{fig:scheme}(b).

    By establishing NFL theorems for these three different learning protocols separately, we rigorously show that there is a significant separation between the CLC-LPs and ReQu-LPs, as well as between the ReQu-LPs and Qu-LPs. We identify that the achieved separation depends on the orthogonality of training states and the properties of the observable $O$.
    
    \noindent \textit{Separation between CLC-LPs and the two quantum-related learning protocols:}
    The ReQu-LPs and Qu-LPs achieve a quadratic reduction from $4^n$ to $2^n$ in terms of the sample complexity compared to CLC-LPs, when the training quantum states are independent and non-orthogonal. The sample complexity $4^n$ for CLC-LPs matches the optimal query complexity for quantum state tomography with non-adaptive single copy measurements when not considering the sampling error of finite measurements \cite{lowe2022lower}. 
    
    \noindent \textit{Separation between ReQu-LPs and Qu-LPs:}
    The learning performance of ReQu-LPs reduces to the same level as that of CLC-LPs, when the training quantum states are orthogonal. In contrast, the learning performance of Qu-LPs remains the same for any type of training state when the observable is a diagonal matrix, achieving a quadratic reduction of sample complexity over the ReQu-LP with orthogonal training states. Additionally, such separation between ReQu-LPs and Qu-LPs shrinks with the increasing magnitude of the non-diagonal elements of the observable.
    
    \noindent \textit{The effect of inter-state relative phases on the separation.} We identify that the separation in learning performance between CLC-LPs and the two quantum learning protocols stems from the capability of ReQu-LPs and Qu-LPs to capture information about the inter-state relative phase, i.e., the difference between the global phases of output states of the target and learned unitary operators. As shown in Fig.~\ref{fig:scheme}, CLC-LPs utilize the measurement outputs as responses, wherein the measurement operation results in the loss of global phase information of the output states. In contrast, ReQu-LPs and Qu-LPs facilitate coherent quantum operations between the target and the optimized quantum systems, thereby effectively exploiting phase information. Additionally, it is important to note that capturing information about the inter-state relative phase is only possible when the training states are non-orthogonal.

    Overall, our main contributions can be  summarized from theoretical and empirical aspects as follows:
    \begin{itemize}
        \item \textbf{Theoretical contributions:} We first establish the quantum NFL theorem for the CLC-LPs, ReQu-LPs and Qu-LPs within the task of learning quantum dynamics under a fixed observable, as elucidated in Theorem~\ref{thm:NFL_Ob_maintext}, Theorem~\ref{thm:NFL_U_maintext}, and Theorem~\ref{thm:NFL_U_dagger_maintext}, respectively. These theorems unveil the average lower bounds of prediction error for various learning protocols, as summarized in Table~\ref{table:NFL_bound}, enabling us to comprehend the potential quantum advantage with quantum learning protocols. In particular, the achieved NFL theorems convey the following implications.
        \begin{enumerate}
            \item There is significant separation in terms of sample complexity between the CLC-LPs and the two quantum-related learning protocols (i.e., ReQu-LPs and Qu-LPs) when using non-orthogonal training states, as well as between ReQu-LPs and Qu-LPs for diagonal observables.
            \item We identify that the separation between the CLC-LPs and the two quantum-related learning protocols (i.e., ReQu-LPs and Qu-LPs) originates from the unique capability of the latter to effectively capture information about the inter-state relative phase over non-orthogonal states. To the best of our knowledge, this is the first study that uncovers the effect of relative phase on learning performance, an aspect overlooked in existing research.
        \end{enumerate} 
        \item \textbf{Numerical contributions:} We conduct numerical simulations for learning quantum unitaries under various observable to verify the established NFL theorems. All numerical results echo our theoretical results.
    \end{itemize}

    The remaining sections of this paper are structured as follows. In section~\ref{sec:preliminaries}, we introduce the preliminaries about quantum computing and variational quantum algorithms, and briefly review the related literature, including the NFL theorem for quantum machine learning as well as the established separation between classical and quantum machine learning in other learning tasks. Subsequently, we define various learning protocols for the task of learning a unitary under given observables in Section~\ref{sec:problem_setup}. The NFL theorems for each learning protocol are established in Section~\ref{sec:NFL} with the numerical simulations for verification conducted in Section~\ref{sec:numerics}. Section~\ref{sec:conclusion} concludes this paper with a discussion. All proofs are relegated to the supplementary material.

	\section{Preliminaries}\label{sec:preliminaries}
	
	In this section, we first review the essential foundations of quantum computing and variational quantum algorithms in Subsection~\ref{subsec:prep-Qc} and Subsection~\ref{subsec:VQA}, respectively. Subsequently, we introduce the related work of the NFL theorems in the context of QML and the established separation in classical and quantum learning protocols in Subsection~\ref{subsec:related_work}. The frequently used notations and abbreviations throughout the whole manuscript are summarized in Table~\ref{table:notation}.
	
	\begin{table}[ht]
		\centering
		\captionsetup{justification=centering}
		\caption{Notations and abbreviations}
		\scalebox{1}{
			\begin{tabular}{ll}
				\toprule
				Notation  & Description \\
				\midrule
				$n$ & Number of qubits \\
				$N$ & Size of training datasets \\
				$\mathcal{H}_d$ & $d$-dimensional Hilbert space \\
				$\ket{\bm{\phi}}$ & Quantum  state \\
				$\rho$ & Density matrix \\
                $\bm{\alpha},\bm{\beta},\bm{\gamma}$ & phase of quantum states \\
				$H^*$ & Complex conjugate of matrix $H$ \\
				$H^{\dagger}$ & Hermitian conjugate of matrix $H, H^{\dagger}=\left(H^{\mathrm{T}}\right)^*$ \\
				$\mathbb{I}$ & Identity matrix \\
				$U,V$ & Unitary matrix \\
				$[N]$ & Discrete set $\{1,2,\cdots,N\}$ \\
                $\haar$ & Haar distribution of quantum states (unitaries) \\
                CLC-LP & CLassiCal Learning Protocol \\
                ReQu-LP & Restricted Quantum Learning Protocol \\
                Qu-LP & Quantum Learning Protocol \\
				\bottomrule
			\end{tabular}
		}
		\label{table:notation}
	\end{table}
	
	\subsection{Quantum computing}
	\label{subsec:prep-Qc}
    \textbf{Quantum bit (qubit) and quantum states}. A quantum bit or qubit is the basic unit of quantum information. Unlike a classical bit that can be in a state of $0$ or $1$, a qubit state refers to a two-dimensional vector defined as  $\ket{\bm{\phi}}=a_0\ket{0}+a_1\ket{1}\in\mathbb{C}^2$ under Dirac notation, where $\ket{0}=[1,0]^{\top}$ and $\ket{1}=[0,1]^{\top}$ specify two unit bases, and the coefficients $a_0, a_1 \in \mathbb{C}$ satisfy $|a_0|^2+|a_1|^2=1$. Similarly, an $n$-qubit state is denoted by $ \ket{\bm{\psi}}=\sum_{i=1}^{2^n}a_i\ket{\bm{e}_i} \in \mathbb{C}^{2^n}$, where $\ket{\bm{e}_i}\in\mathbb{R}^{2^n}$ is the unit vector whose $i$-th entry is 1 and other entries are 0, and $\sum_{i=0}^{2^n-1}|a_i|^2=1$ with $a_i\in\mathbb{C}$. 

    \smallskip

    \noindent \textbf{Global phase, relative phase, and inter-state relative phase.} An $n$-qubit quantum state $\ket{\bm{\psi}}=\sum_{j=1}^{2^n}a_j\ket{\bm{e}_j} $ can be written as $\ket{\bm{\psi}}=e^{i\bm{\gamma}}\sum_{j=1}^{2^n}c_je^{i\bm{\gamma_j}}\ket{\bm{e_j}}$ with $c_j\in \mathbb{R}$ being a real number, $\bm{\gamma}_1=0$, and $a_j=e^{i\bm{\gamma}+\bm{\gamma}_j}c_j$. In this manner, $\bm{\gamma}$ and $\bm{\gamma}_j$ refer to the global phase and relative phase, respectively. In this study, we introduce another relative phase between states, dubbed the inter-state relative phase. Given two states $\ket{\bm{\psi}}$ and $\ket{\bm{\phi}}$ satisfying $\ket{\bm{\psi}}=e^{i\bm{\alpha}}\ket{\bm{\phi}}$, $\bm{\alpha}$ refer to the inter-state relative phase. In the end, the global phase is a special inter-state relative phase with the relative phase of $\ket{\bm{\phi}}$ satisfying $\bm{\gamma}_1=0$.

    \smallskip
    
    \noindent \textbf{Quantum gates and quantum circuits}. Quantum gates are the basic operations applied to qubits in a quantum computer. An $n$-qubit gate can be characterized by a unitary operator obeying $UU^{\dagger}=\mathbb{I}_{2^n}$. In practical quantum computing, only single-qubit and two-qubit gates are available. Typical single-qubit gates include the Pauli gates, which can be written as Pauli matrices:
		\begin{equation}
			X = \left[ \begin{array}{ccc}
				0 & 1 \\
				1 & 0 \\
			\end{array}
			\right],~  
			Y = \left[ \begin{array}{ccc}
				0 & -i \\
				i & 0 \\
			\end{array}
			\right],~
			Z = \left[ \begin{array}{ccc}
				1 & 0 \\
				0 & -1 \\
			\end{array}
			\right]. \label{eq:pauli}
		\end{equation}
		The more general quantum gates are their corresponding rotation gates $R_X(\bm{\theta})=e^{-i\frac{\bm{\theta}}{2}X}, R_Y(\bm{\theta})=e^{-i\frac{\bm{\theta}}{2}Y}$, and $R_Z(\bm{\theta})=e^{-i\frac{\bm{\theta}}{2}Z}$ with a tunable parameter $\bm{\theta}$.
    A quantum circuit is a sequence of quantum gates applied to qubits to perform a computation, aiming to achieve arbitrary large quantum gate operations.

    \smallskip
    
    \noindent \textbf{Quantum measurements}. Quantum measurement is the process of observing a quantum state, described by a Hermitian operator $H\in \mathcal{C}^{2^n \times 2^n}$ for an $n$-qubit system. Specifically, applying the measurement $H$ to the state $\ket{\bm{\psi}}$ allows for obtaining information about the state in terms of the expectation value 
	$\Tr(H \ket{\bm{\psi}}\bra{\bm{\psi}})$. For states $\ket{\bm{\psi}}=e^{i\bm{\gamma}}\ket{\bm{\phi}}$, the measurement outputs are independent with the inter-state relative phase $\bm{\gamma}$ as $\Tr(H \ket{\bm{\psi}}\bra{\bm{\psi}})=\Tr(H \ket{\bm{\phi}}\bra{\bm{\phi}})$.

    \smallskip
    
    \noindent \textbf{Haar measures, Haar states, and Haar unitaries}. The Haar measure is a fundamental concept in mathematics, characterizing the uniform distribution on a topological group, more specifically a unitary group $\mathbb{U}(2^n)$ \cite{nielsen2010quantum}. The Haar measure could be defined by its invariance under the group's action. Namely, let $\mu$ be the Haar measure, then $\mu(U) = \mu(UV)= \mu(VU)$ holds for any $U,V\in \mathbb{U}(2^n)$.  A Haar unitary refers to a unitary matrix that is selected according to the Haar measure from the group $\mathbb{U}(2^n)$. Random states obtained from the induced measure of Haar unitary are called Haar-random states. More precisely, any states obtained by acting the Haar unitary on an initial quantum state refer to the Haar-random states. In this regard, a Haar-random state in a Hilbert space is simply a state that is chosen at random from all possible states within that space with equal probability. For example, single-qubit Haar states are uniformly distributed on the Bloch sphere.

	\subsection{Variational quantum algorithms}\label{subsec:VQA}
	Variational quantum algorithms (VQAs) \cite{cerezo2021variational_VQA} provide a general framework for implementing the task of unitary learning with various learning protocols. As shown in Fig.~\ref{fig:VQA}, a VQA can be divided into four elemental blocks: the preparation of the initial state $\ket{\bm{\psi}}$, the implementation of a parameterized quantum circuit (PQC) $V(\bm{\theta})$, the measurement $H$, and the construction of the cost function $\mathcal{L}(\bm{\theta})$.

    \smallskip
    
    \noindent \textbf{Initial state preparation}. The input data in VQAs could be quantum states or classical vectors $\bm{x}$. In the latter case, the classical data must first be embedded into quantum states to enable processing on quantum computers. Typical encoding manners include amplitude embedding and gate embedding. Amplitude embedding maps the $2^n$-dimension vector $\bm{x}=(\bm{x}_1, \cdots, \bm{x}_{2^n})$ into the amplitude of an $n$-qubit quantum state $\ket{\bm{x}}=\sum_i \bm{x}_i\ket{\bm{e}_i}/\|\bm{x}\|_2 $ with $\ket{\bm{e}_i}$ being the computational basis. For gate embedding, the entries of classical vectors are encoded into the $n$-qubit parameterized gates $W(\bm{x})$ acting on the basis state $\ket{\bm{0}}^{\otimes n}$, leading to the quantum state representation $\ket{\bm{x}}=W(\bm{x})\ket{\bm{0}}^{\otimes n}$. In the subsequent discussion, we refer to the input data as quantum states $\ket{\bm{\psi}}$ assuming that all classical inputs have been encoded into quantum states.

    \smallskip
    
    \noindent \textbf{Implementations of PQCs}. The PQCs $V(\bm{\theta})$ consist of a sequence of parameterized gates with flexible structures, including hardware efficient ansatz, Hamiltonian variational ansatz, and tensor-network based ansatz \cite{kandala2017hardware,benedetti2019parameterized,du2021learnability,cong2019quantum,du2022efficient}. The PQCs can be directly applied to the prepared initial state or integrated into the process of initial state preparation, aiming to obtain the feature state $\ket{\bm{\psi}(\bm{\theta})}$ which takes the form of $V(\bm{\theta})\ket{\bm{\psi}}$ for the former case and $\prod_{k}V_k(\bm{\theta})W_k(\bm{x})\ket{\bm{0}}^{\otimes n}$ for the latter case \cite{perez2020data}.

    \smallskip
    
    \noindent \textbf{Cost function}.  The specific form of the cost function $\mathcal{L}(\cdot)$ varies with different learning tasks and learning algorithms. For example, in the classification task with training dataset $\mathcal{D}=\{\ket{\bm{\psi}_j}, \bm{y}_j\}$, the cost function refers to
    \begin{align}\label{eq:loss-QNN}
	\mathcal{L}_{\mathcal{D}}(\bm{\theta}) =    \frac{1}{N}\sum_{j=1}^N  (\Tr(V(\bm{\theta})\ket{\bm{\psi}_j}\bra{\bm{\psi}_j}V(\bm{\theta})^{\dagger}H) - \bm{y}_j)^2.
	\end{align}
    where $\bm{y}_j$ is the discrete label of $\ket{\bm{\psi}_j}$ and $H$ refers to a predefined projective measurement.

    \smallskip

    \noindent \textbf{Optimization of VQAs}. The optimization of VQA follows an iterative manner, where a classical optimizer utilizes the outputs from the quantum circuit to continuously update the trainable parameters $\bm{\theta}$ of the PQCs. This process aims to minimize the predefined objective function $\mathcal{L}_{\mathcal{D}}(\bm{\theta})$ in Eqn.~(\ref{eq:loss-QNN}). At the $t$-th iteration, the updating rule is 
    \begin{equation}
        \bm{\theta}^{(t+1)} = \bm{\theta}^{(t)}-\eta\frac{\partial \mathcal{L}_{\mathcal{D}}(\bm{\theta}^{(t)})}{\partial \bm{\theta}}, 
    \end{equation}    
    where $\eta$ is the learning rate. The updating continuously proceeds unless the loss $\mathcal{L}_{\mathcal{D}}$ is converged or the iteration $t$ exceeds the threshold. Finally, the optimized parameterized quantum circuits are employed to implement prediction on the unseen data. As the loss function $\mathcal{L}_{\mathcal{D}}$ measures how well the learned model performs in predicting the accurate response, it plays a crucial role in determining the model performance. 

    \begin{figure}  
    		\centering
    		\includegraphics[width=0.48\textwidth]{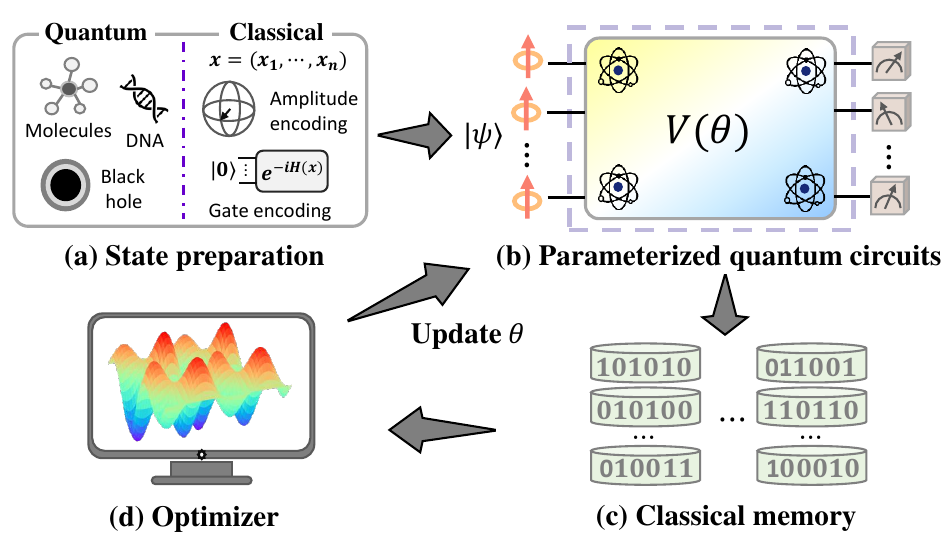}
    		\caption{\small{\textbf{Paradigm of variational quantum algorithms.} (a) Preparing input states by obtaining from quantum experiments or encoding classical data. (b) The input state is evolved by a parameterized quantum circuit $V(\bm{\theta})$ and then is measured to obtain the classical data stored in (c) classical memory. (d) The classical optimizer is employed to update the parameters $\bm{\theta}$ according to the measurement outputs.}}
    		\label{fig:VQA}
    	\end{figure}
     
    \smallskip
    
    \noindent \textbf{Evaluation metrics.} The typical metrics in evaluating the learning models include generalization error, sample complexity, or query complexity \cite{mohri2018foundations}. Generalization error $\R_{\Gene}$ refers to the difference between the prediction error $\R=\mathbb{E}_{\mathcal{D}}\mathcal{L}_{\mathcal{D}}(\bm{\theta})$ and training error $\R_{\ERM}=\mathcal{L}_{\mathcal{D}}(\bm{\theta})$, i.e., 
    \begin{equation}\label{eqn:def-gene-err}
    	\R_{\Gene}=\R-\R_{\ERM}.
    \end{equation}
   Moreover, generalization error is exactly the prediction error when the training error is zero. Sample complexity refers to the size of training datasets, or equivalently, the number of distinct quantum training states, which is generally considered in the ideal setting without considering the measurement error.
    
    Query complexity, a distinctive metric exclusive to QML, refers to the total number of queries of the explored quantum system. Remarkably, both query complexity and sample complexity have practical meanings from various considerations \cite{banchi2023statistical}. In this work, we only consider the relation between sample complexity and generalization error.

	\subsection{Related work}\label{subsec:related_work}
	In this subsection, we first review relevant literature pertaining to the  NFL theorem in the context of quantum machine learning in Subsubsection~\ref{subsubsec:related_work_NFL}. Subsequently, we discuss the differences in the achieved separation between classical and quantum learning protocols with that demonstrated in previous studies in Subsubsection~\ref{subsubsec:related_work_separa}.

	\subsubsection{Quantum NFL theorem}\label{subsubsec:related_work_NFL}
	The No-Free-Lunch (NFL) theorem is a renowned result in learning theory, emphasizing the constraints on the ability to learn a function solely from a training dataset. Its significance in classical learning theory is extensively acknowledged \cite{wolpert1996existence, wolpert1996lack, wolpert1997no, ho2002simple,wolf2018mathematical,adam2019no}. Ref.\cite{poland2020no} pioneered the NFL theorem's application to quantum machine learning, where both inputs and outputs are quantum states. Subsequently, Ref.\cite{sharma2022reformulation} refined the quantum NFL theorem for quantum-assisted learning protocols where the bipartite entangled states are prepared as the training states of learning models by introducing a reference quantum system. They highlighted that the use of such entangled data could offset the exponential training data size demands when learning unitaries with non-entangled data. Ref.~\cite{zhao2023learning} provides a unified information-theoretic reformulation of the quantum no-free-lunch theorem from the lens of representation space of quantum states, indicating the same power of mixed states \cite{yu2023optimal} as entangled states. Furthering this discourse, Ref.~\cite{wang2023transition} shifted the focus towards a pragmatic learning scenario—learning a target operator through finite measurement outcomes on quantum states, as opposed to learning directly from the quantum output states themselves. Their findings underlined that the degree of entanglement in quantum states presents a dual effect on the generalization performance of learning models contingent on the number of measurements. 
 
	However, despite the promising results shown by entangled data in reducing the size of training datasets, the preparation of entangled data is highly resource-intensive and tends to introduce additional quantum noise, particularly in the early stages of quantum computing. The problem settings explored in this work lean towards a more realistic approach, focusing on learning the quantum unitaries using unentangled data.

	\subsubsection{Separation between classical and quantum learning models} \label{subsubsec:related_work_separa}
	As the aim of quantum machine learning is to harness the peculiarities of quantum mechanics to improve upon the capabilities of classical learning models, the separation in learning performance between classical and quantum learning models has been explored in various learning tasks under various metrics, including classification \cite{havlivcek2019supervised,qian2022dilemma} and quantum system learning \cite{anshu2024survey}.

    For the classification tasks, Ref.~\cite{huang2021power} has demonstrated that for carefully sculpted synthetic data with quantum computers, quantum kernels could achieve a significant reduction in terms of the generalization error over the best classical kernel methods. Moreover, Ref.~\cite{liu2021rigorous} has shown the power of quantum kernels on classical data of discrete logarithm problems which could not be efficiently addressed by classical algorithms in polynomial time.

    Quantum system learning aims to fully characterize the quantum states or quantum channels by measuring the outcomes of the quantum systems.  A plethora of studies delved into exploring the separation in terms of query complexity between the learning protocols with and without employing quantum resources. In particular, Ref.~\cite{haah2016sample} obtains the optimal query complexity for quantum state learning by employing quantum memory to store multiple copies of a quantum state and performing entangled measurement on it, achieving a quadratic reduction in query complexity compared to the learning protocols without quantum memory. The same spirits are employed in the task of learning the properties of quantum states with classical shadow \cite{huang2020predicting}. Ref.~\cite{huang2021information} and Ref.~\cite{chen2022exponential} have demonstrated that in the average scenario, there is no separation between the learning protocols with and without quantum memory for such tasks, while there is an exponential separation in the worst case of specific quantum states. However, the utilization of quantum memory is highly resource-intensive in quantum resources, which are scarce in the early stages of quantum computing. In this regard, alternative algorithms with minimal quantum memory are proposed to achieve the same exponential separation by employing the conjugate states \cite{king2024exponential}. 

    Despite the advancement of quantum learning protocols, quantum advantages are established for specific data or problems, leaving the universal power of quantum learning protocols unknown. 

    \section{Problem setup for quantum dynamics learning}\label{sec:problem_setup}
    Quantum dynamics learning involves converting a black-box quantum evolution operator into a digital form that can be analyzed on the quantum or classical computer. As discussed in Section~\ref{sec:introduction}, quantum dynamics learning is essential for understanding the behavior and evolution of quantum systems, enabling precise control and manipulation of quantum states for various applications. By studying the evolution of quantum states, researchers can develop efficient algorithms such as quantum Fourier transformation \cite{nielsen2010quantum} and variational quantum eigensolver \cite{tilly2022variational}, optimize quantum circuits for quantum error correction \cite{chiaverini2004realization}, and tackle complex problems in fields like quantum information processing \cite{beckey2022variational, cerezo2020variational,larose2019variational}, quantum chemistry and material sciences \cite{bauer2016hybrid, google2020hartree, o2019calculating, peruzzo2014variational}, and particle physics \cite{avkhadiev2020accelerating, kokail2019self}. In other words, learning quantum dynamics is often a necessary step before executing algorithms on real quantum devices. Several companies have developed their own commercial quantum learning platforms to facilitate this process \cite{cross2017open, smith2016practical}.
    
    Here we consider the specific form of quantum dynamics learning---learning a quantum unitary $U$ under a given observable $O$, where $U$ could be the unknown dynamics of an experimental quantum system. Formally, the target function is given by 
    \begin{align}\label{eq:learning_model}
	f_{U}(\bm{\psi}) = \Tr(OU\ket{\bm{\psi}}\bra{\bm{\psi}}U^{\dagger}),
	\end{align}
    where $O$ is an arbitrary bounded Hermitian operator with $\|O\|\le \infty$. Let $\mathcal{D}=\{(\ket{\bm{\psi}_j},\bm{y}_j)\}_{j=1}^N$ be the training dataset of $N$ examples, where $\bm{y}_j$ could be quantum data or classical data, and is not necessarily represented in the form of Eqn.~(\ref{eq:learning_model}). In particular, the type of response $\bm{y}_j$ varies with different learning protocols as detailed below. The aim of quantum dynamics learning is to optimize a tunable unitary $V(\bm{\theta})$ to predict the output for an unseen input $\bm{\ket{\psi}}$ in as $h(\bm{\psi}) = \Tr(OV(\bm{\theta})\ket{\bm{\psi}}\bra{\bm{\psi}}V(\bm{\theta})^{\dagger})$, where $\bm{\theta}$ refers to tunable variables which could be discrete circuit structure or continuous parameters.

    \subsection{Risk function and No-Free-Lunch theorem} 
    The risk function (or prediction error) is a crucial measure in statistical learning theory to quantify how well the learned hypothesis function $h$ performs in predicting the concept $f$. As discussed in Subsection~\ref{subsec:VQA}, the risk function in the context of learning the target function of Eqn.~(\ref{eq:learning_model}) refers to 
    \begin{equation}\label{eq:risk}
        R_{U}(V_{{\mathcal{D}}}) = \int \mathrm{d} {\bm{\psi}} \left( f_U(\bm{\psi})-h_{{\mathcal{D}}}(\bm{\psi})  \right) ^2,
    \end{equation}
    where $h_{\mathcal{D}}(\bm{\psi}_j)=\Tr(OV_{\mathcal{D}}\ket{\bm{\psi}_j}\bra{\bm{\psi}_j}V_{\mathcal{D}}^{\dagger})$ with $V_{\mathcal{D}}$ referring to the trained unitary, and the integration is taken over the Haar states $\ket{\bm{\psi}}$. In the following, the term risk function and prediction error will be used interchangeably to elucidate the theoretical and empirical results.

    In traditional machine learning, the no-free-lunch (NFL) theorem aims at deriving the lower bound of the average risk function across all possible training datasets and target concepts, under the assumption of perfect training—that is, the learned hypothesis achieves zero training error for a given loss function \cite{wolpert1997no, wolf2023mathematical}. In the context of learning $n$-qubit quantum dynamics, the possible concept refers to the unitary operator within the unitary group, i.e., $U\in \mathbb{U}(d)$ with $d=2^n$ being the dimension of $n$-qubit quantum system. Moreover, the uniform distribution of the target concept corresponds to the Haar distribution of unitaries as introduced in Subsection~\ref{subsec:prep-Qc}. The type of training dataset $\mathcal{D}$ depends on the employed learning protocol. In this regard, we formulate the problem considered by the NFL theorem for quantum dynamics learning as follows.
    \begin{prob}\label{prob:NFL}
        Let $\mathcal{D}=\{(\ket{\bm{\psi}_j}, \bm{y}_j)|\ket{\bm{\psi}_j}\in \mathcal{H}_{\mathcal{X}},\bm{y}_j\in \mathcal{H}_{\mathcal{Y}} \}_{j=1}^N$ be the training dataset and $\mathcal{P}:\mathcal{H}_{\mathcal{X}}\to \mathcal{H}_{\mathcal{Y}}$ be the related learning protocol satisfying the perfect training assumption $\mathcal{P}(\ket{\bm{\psi}_j})=\bm{y}_j$ for any $j\in[N]$. Let $h_{\mathcal{D}}(\bm{\psi})=\Tr(OV_{\mathcal{D}}\ket{\bm{\psi}}\bra{\bm{\psi}}V_{\mathcal{D}}^{\dagger})$ be the learned hypothesis and $O$ an arbitrary observable. The NFL theorem for quantum dynamics learning considers the lower bound of the average risk function
        \begin{align}\label{eq:risk_function}
	\mathbb{E}_{U} \mathbb{E}_{\mathcal{D}}R_{U}(V_{{\mathcal{D}}}) = \mathbb{E}_{U} \mathbb{E}_{\mathcal{D}} \int \mathrm{d} {\psi} \left( f_U(\bm{\psi})-h_{{\mathcal{D}}}(\bm{\psi})  \right) ^2,
	\end{align}
        where the expectation is taken over the Haar unitary and all possible training datasets in some specific types.
    \end{prob}
   
    Notably, the average risk function depends on the hypothesis $h_{{\mathcal{D}}}(\cdot)$ learned from specific learning protocols. Under the framework of NFL theorems considering the average prediction error over various learning algorithms under the perfect training assumption regardless of the optimization process, various learning algorithms only differ in the varying restrictions of access to quantum resources. In this regard, we encapsulate different learning algorithms into three learning protocols according to the level of access to quantum resources, namely classical learning protocols, restricted quantum learning protocols, and quantum learning protocols, where their implementations are elucidated in the subsequent subsections.

 \noindent  \textbf{Remark.}  As discussed in Section~\ref{sec:introduction}the way of implementing $h(\bm{\psi})$ is versatile, where numerous advanced learning algorithms can be used to learn the target unitaries $U$, such as VQAs \cite{cerezo2020variational}, quantum process tomography with tensor networks \cite{torlai2023quantum}, and quantum circuit representation based on local inversions \cite{huang2024learning}. Besides, the optimization approaches of the tunable unitary $V(\bm{\theta})$ vary with different learning algorithms. The implementation of the tunable $V(\bm{\theta})$ is arbitrary, as long as it caters to the constraints on the available quantum resources in different learning protocols. For convenience in our discussion, we assume that the tunable unitary $V(\bm{\theta})$ is implemented on parameterized quantum circuits, which could be applied to various learning protocols. This assumption is mild as we only consider the optimized unitary to achieve perfect training regardless of its implementation.

    \begin{figure}  
		\centering
		\includegraphics[width=0.48\textwidth]{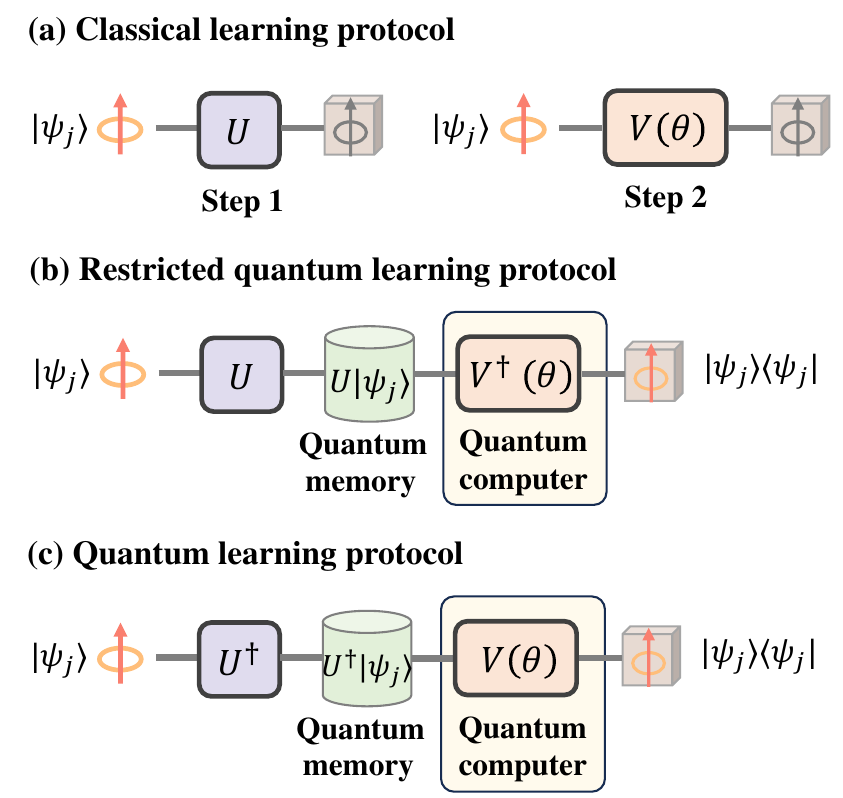}
		\caption{\small{\textbf{Scheme of various learning protocols.} The output states $U\ket{\bm{\psi}_j}$ and $U^{\dagger}\ket{\bm{\psi}_j}$ are stored in quantum memory. The tunable unitary $V(\bm{\theta})$ is implemented on quantum computers in ReQu-LPs and Qu-LPs. However, in the case of CLC-LPs, it can also be implemented on classical computers using tensor network techniques \cite{orus2019tensor} or deep neural networks \cite{gao2017efficient1}.}}
		\label{fig:learning_protocols}
	\end{figure}
 
    \subsection{Classical learning protocols}\label{subsec:CLP}
    The CLassiCal Learning Protocols (CLC-LPs) employ the states $\ket{\bm{\psi}_j}$ as inputs and the measurement output of the fixed observable $O$ on the evolved states $U\ket{\bm{\psi}_j}$ as the response, as shown in Fig.~\ref{fig:learning_protocols}(a). The input $\ket{\bm{\psi}_j}$ could be quantum states obtained in experiments or classical states of bits, referring to the computational basis $\ket{\bm{e}_j}$. Denote the fixed observable $O$ defined in Eqn.~(\ref{eq:learning_model}) as the spectral decomposition $O=\sum_{q=1}^{r} \bm{o}_q\ket{\bm{o}_q}\bra{\bm{o}_q}$, where  $\ket{\bm{o}_q}$ are the eigenvectors related to the eigenvalues $\bm{o}_q$ of observable $O$. As the expectation value of the observable $O$ could be obtained by measuring the output states with the projective measurement $\{\ket{\bm{o}_1}\bra{\bm{o}_1}, \cdots, \ket{\bm{o}_r}\bra{\bm{o}_r}\}$, the training dataset of $N$ examples takes the form 
    \begin{equation}\label{eq:data_c}
    \mathcal{D}_{\C}=\{(\ket{\bm{\psi}_j}, \bm{a}_j)~|~\ket{\bm{\psi}_j} \in \mathcal{H}_{d},\bm{a}_j=(\bm{a}_{j1}, \cdots, \bm{a}_{jr}) \}_{j=1}^N,
    \end{equation}
    where $\bm{a}_{jq}=\Tr(U^{\dagger}\ket{\bm{o}_q}\bra{\bm{o}_q}U\ket{\bm{\psi}_j}\bra{\bm{\psi}_j}U)$ refers to the expectation value of projective measurement $\ket{\bm{o}_q}\bra{\bm{o}_q}$. In this regard, the loss function is given by
    \begin{align}\label{eq:inco_loss_eigen}
        \mathcal{L}_C(\bm{\theta}) = & \frac{1}{N}\sum_{q=1}^r \sum_{j=1}^N  \bigg(\Tr(\ket{\bm{o}_q}\bra{\bm{o}_q}V(\bm{\theta})\ket{\bm{\psi}_j}\bra{\bm{\psi}_j}V(\bm{\theta})^{\dagger}- 
        \nonumber \\
        & \Tr(\ket{\bm{o}_q}\bra{\bm{o}_q}U\ket{\bm{\psi}_j}\bra{\bm{\psi}_j}U^{\dagger})\bigg)^2,
    \end{align}
    where the tunable unitary $V(\bm{\theta})$ is trained to learn the output of each projective measurement $\ket{\bm{o}_q}\bra{\bm{o}_q}$. 
    This leads to the form of perfect training assumption
    \begin{align}       \Tr(U^{\dagger}\ket{\bm{o}_q}\bra{\bm{o}_q}U\ket{\bm{\psi}_j}\bra{\bm{\psi}_j})&=\Tr(V_{\C}^{\dagger}\ket{\bm{o}_q}\bra{\bm{o}_q}V_{\C}\ket{\bm{\psi}_j}\bra{\bm{\psi}_j}), \label{eq:pt_inco}
    \end{align}
    where $V_{\C}$ is the well-trained unitary and $j\in[N], q\in [r]$. 
    Notably, the trained unitary $V_{\C}$ could be employed to predict the output of a class of observables with the form of $O_j=\sum_{q=1}^r \bm{o}_{jq}\ket{\bm{o}_q}\bra{\bm{o}_q}$ for any $\bm{o}_{jq} \in \mathbb{R}$ by post-processing the classical output of the projective measurement $\ket{\bm{o}_q}\bra{\bm{o}_q}$.

\noindent    \textbf{Remark}. (i) This learning protocol resembles the classical shadow \cite{huang2020predicting} where the outputs of random projective measurement are utilized to efficiently predict the expectation value of exponential diverse observables on a given quantum state, namely $\Tr(U^{\dagger}OU\rho)$ for $O\in \{O_j\}_{j=1}^{2^n}$. However, we focus on the sample complexity of learning a random quantum unitary, while classical shadow considers the query complexity for predicting the properties of specific states. (ii) This learning protocol is also dubbed the incoherent learning protocol in Ref.~\cite{jerbi2023power}. We suggest the name `classical learning protocols' because when the target unitary is generated by Clifford circuits, this learning protocol could be implemented completely on the classical computer without using any quantum resources by training on classical bitstrings, as Gottesman-Knill theorem \cite{gottesman1998heisenberg} shows that the Clifford circuits can be efficiently simulated with classical computers.

    \subsection{Restricted quantum learning protocols}\label{subsec:RQLP}
    Restricted Quantum Learning Protocols (ReQu-LPs) allow coherent operations to act on the target quantum system $U$, meaning that the target unitary $U$ and the tunable unitary $V(\bm{\theta})$ could be implemented on the same quantum devices. Specifically, the output states $U\ket{\bm{\psi}_j}$, derived from the given input state $\ket{\bm{\psi}_j}$, can be preserved in quantum memory. This storage facilitates subsequent direct operations using a parameterized quantum circuit, as depicted in Fig.~\ref{fig:learning_protocols}(b). In this end, the training data refers to the pair of input-output states $(\ket{\bm{\psi}_j}, U\ket{\bm{\psi}_j})$ stored in quantum memory. This differs from the CLC-LPs in the responses where the latter employs the measurement outcomes on the output state $U\ket{\bm{\psi}_j}$ for training.
    Given the training dataset with $N$ examples, i.e.,
    \begin{equation}\label{eq:data_RQ}
        \mathcal{D}_{\RQ}=\{(\ket{\bm{\psi}_j}, U\ket{\bm{\psi}_j} )~|~\ket{\bm{\psi}_j} \in \mathcal{H}_{d}\}_{j=1}^N,
    \end{equation}
    the loss function takes the trace norm of the difference between the evolved states under the target unitary $U$ and the tunable unitary $V(\bm{\theta})$, which is given by
	\begin{align}\label{eq:trace_norm_loss}
	\mathcal{L}_{\RQ}(\bm{\theta}) = \frac{1}{4N}\sum_{j=1}^N \big\| & U\ket{\bm{\psi}_j}\bra{\bm{\psi}_j}U^{\dagger}- 
	\nonumber \\ 
	& V(\bm{\theta})\ket{\bm{\psi}_j}\bra{\bm{\psi}_j}V(\bm{\theta})^{\dagger} \big\|_1^2.
	\end{align}
	By using the relation between the trace norm and fidelity \cite{nielsen2010quantum}, this loss function can be reformulated as the expression in terms of the average fidelity
	\begin{equation}\label{eq:coherent_loss_setup}
	\mathcal{L}_{\RQ}(\bm{\theta})=1-\frac{1}{N}\sum_{j=1}^N \bra{\bm{\psi}_j}V(\bm{\theta})^{\dagger}U\ket{\bm{\psi}_j}^2.
	\end{equation}
For ReQu-LPs, optimizing the loss function needs to calculate the inner product between the output states $U\ket{\bm{\psi}_j}$ and $V(\bm{\theta})\ket{\bm{\psi}_j}$, which could be efficiently computed employing the Loschmidt
	echo or swap test circuit as shown in Fig.~\ref{fig:scheme}(b). Moreover, given the training dataset $\mathcal{D}_{\RQ}$ in Eqn.~(\ref{eq:data_RQ}), the assumption of perfect training refers to $|\bra{\bm{\psi}_j}U^{\dagger}V_{\RQ}\ket{\bm{\psi}_j}|^2=1$ for any $j\in[N]$, or equivalently
    \begin{equation}
        U\ket{\bm{\psi}_j}=e^{i\bm{\alpha}_j}V_{\RQ}\ket{\bm{\psi}_j}, 
        \label{eq:pt_co}
    \end{equation}
    where $V_{\RQ}$ refers to the well-trained unitary on the training dataset $\mathcal{D}_{\RQ}$, $\bm{\alpha}_j$ refers to the inter-state relative phase between the output states and varies with different input states $\ket{\bm{\psi}_j}$. Particularly, the phase $\bm{\alpha}_j$ could be arbitrary and is uniformly distributed over any period of the function $e^{i\bm{\alpha}_j}$ in the context of NFL theorems considering the average case of all related ingredients.

    \subsection{Quantum learning protocols with access to $U^{\dagger}$}\label{subsec:QLP} 
    While ReQu-LPs are restricted to access to the vanilla target quantum system $U$, Quantum Learning Protocols (Qu-LPs) allow access to the transformation of the target unitary $U$. 
    In particular, employing the transformation of quantum systems such as conjugate or inverse has recently received significant attention as a potential way for achieving provable quantum advantage \cite{salmon2023provable, van2023quantum}.
    Here, we consider the case in which the access to the inverse of the target unitary $U^{\dagger}$ is allowed. Effective algorithms have been proposed to implement this transformation \cite{chen2024quantum}. In this case, the training dataset of size $N$ refers to 
    \begin{equation}\label{eq:data_Q}
        \mathcal{D}_{\Q}=\{(\ket{\bm{\psi}_j}, U^{\dagger}\ket{\bm{\psi}_j}): \ket{\bm{\psi}_j}\in \mathcal{H}_d\}_{j=1}^N.
    \end{equation}
    Moreover, the inverse of the tunable unitary is employed in the inference stage, which can be efficiently implemented by reversing the order of the gates and replacing each with its inverse. In particular, we denote the tunable unitary as $V(\bm{\theta})^{\dagger}$ for training and $V(\bm{\theta})$ for inference, as shown in Fig.~\ref{fig:learning_protocols}(c). In this regard, the loss function refers to
    \begin{equation}\label{eq:loss_co_dagger}
        \mathcal{L}_{\Q}(\bm{\theta})=1-\frac{1}{N}\sum_{j=1}^N \left|\bra{\bm{\psi}_j}V(\bm{\theta})U^{\dagger}\ket{\bm{\psi}_j}\right|^2,
    \end{equation}
    and the assumption of perfect training is given by 
    \begin{align}
        U^{\dagger}\ket{\bm{\psi}_j}&=e^{i\bm{\beta}_j}V_{\Q}^{\dagger}\ket{\bm{\psi}_j}, \label{eq:pt_co_dagger}
    \end{align}
    where $V_{\Q}$ refers to the well-trained unitary on $\mathcal{D}_{\Q}$, $\bm{\beta}_j$ refers to the inter-state relative phase between the output states $U^{\dagger}\ket{\bm{\psi}_j}$ and $V_{\Q}^{\dagger}\ket{\bm{\psi}_j}$. Similar to $\bm{\alpha}_j$ defined in Eqn.~(\ref{eq:pt_co}), the phase $\bm{\beta}_j$ is uniformly distributed over any period of the function $e^{i\bm{\beta}_j}$.
    
    Formally, the loss function defined in Eqn.~(\ref{eq:loss_co_dagger}) differs from that specified in Eqn.~(\ref{eq:coherent_loss_setup}) for ReQu-LPs solely by the order of multiplication of $U^{\dagger}$ and $V(\bm{\theta})$. This alteration significantly impacts the learning performance between ReQu-LPs and Qu-LPs for certain input states, which will be illustrated in Section~\ref{subsec:sep_theoretic}.

    \section{NFL theorem for quantum dynamics learning}\label{sec:NFL}
    In this section, we first introduce how to reduce the problem regarding NFL theorems of deriving a lower bound as elucidated in Problem~\ref{prob:NFL} to a simplified problem of deriving an upper bound in Subsection~\ref{subsec:central_lemma}. Then, we  establish the NFL theorems for CLC-LPs, ReQu-LPs, and ReQu-Qu-LPs in Section~\ref{subsec:CNFL}, Section~\ref{subsec:RQNFL}, and Section~\ref{subsec:QNFL}, respectively. Finally, we discuss the separation in terms of sample complexity between various learning protocols and identify the quantum resource used in quantum-related learning protocols causing such separation in Section~\ref{subsec:sep_theoretic}.

    \subsection{NFL theorem for QML}\label{subsec:central_lemma}
	Recall that the NFL theorem considers the lower bound of the average risk function as detailed in Problem~\ref{prob:NFL}. It is well known that deriving lower bounds is more challenging than deriving upper bounds, as there are well-established techniques for deriving upper bounds, such as using the Vapnik-Chervonenkis (VC) dimension \cite{vapnik1982necessary}, Rademacher complexity \cite{bartlett2002rademacher}, or covering numbers \cite{tikhomirov1993varepsilon}. In contrast, deriving lower bounds often requires more bespoke arguments tailored to the specific problem or model at hand. Fortunately, in the unitary learning problem, the problem of lower bounding the averaged risk function in Eqn.~(\ref{eq:risk_function}) could be transformed into the problem of upper bounding a trace function, as encapsulated in the following lemma with the proof deferred to Appendix~C.
	\begin{lemma}\label{lem:upper_to_lower}
		The averaged risk function defined in Eqn.~(\ref{eq:risk_function}) over the Haar input states yields
		\begin{align}\label{eq:risk_trans}
		R_U(V_{\mathcal{D}}) = \frac{1}{d(d+1)}\left[2\Tr(O^2)-2\Tr(U^{\dagger}OUV_{\mathcal{D}}^{\dagger}OV_{\mathcal{D}}) \right],
		\end{align}
        where $d=2^n$ refers to the dimension of the $n$-qubit quantum system, $V_{\mathcal{D}}$ refers to the learned unitary with training dataset $V_{\mathcal{D}}$.
	\end{lemma} 
	Lemma~\ref{lem:upper_to_lower} gives the following three implications. First, deriving the upper bound of the average of trace function $\Tr(U^{\dagger}OUV_{\mathcal{D}}^{\dagger}OV_{\mathcal{D}})$ suffices to obtain the lower bound of the average risk function $R_U(V_{\mathcal{D}})$. Such an upper-bounding problem could be effectively derived by combining it with the assumption of perfect training under various learning protocols, which will be elucidated in the following subsections. Second, there is a fixed $d$-dependent factor $1/d(d+1)$ in the prediction error. This is because we consider the learning performance in the average case where the Haar integration with respect to the input states leads to the factor $1/d(d+1)$ (Refer to Appendix~C for details). 
	Third, when the observable $O$ is proportional to the identity, the average risk function vanishes for any training dataset and target unitary. In this regard, we consider the observable satisfying $O\ne c\mathbb{I}_d$ for any complex number $c\in \mathbb{C}$.
 
	\subsection{NFL theorem for classical learning protocols}\label{subsec:CNFL}
    We first consider the NFL theorem in the CLassiCal Learning Protocol (CLC-LP) introduced in Section~\ref{subsec:CLP}. 
    Under this learning protocol with perfect training assumption in Eqn.~(\ref{eq:pt_inco}), the following NFL theorem gives the lower bound of the average risk function defined in Eqn.~(\ref{eq:risk_function}).

    \begin{theorem}
		[NFL for CLC-LPs] 
        Let $O$ be an arbitrary observable. For the learning task defined in Eqn.~(\ref{eq:learning_model}), when the CLC-LPs achieve the perfect training of Eqn.~(\ref{eq:pt_inco}) for the loss function defined in Eqn.~(\ref{eq:inco_loss_eigen}), the average risk function $R_U(V_{\C})$ defined in Eqn.~(\ref{eq:risk_function}) over all $n$-qubit unitaries $U$ and training datasets $\mathcal{D}_{\C}$ yields
		\begin{align}
		\mathbb{E}_{U,\mathcal{D}_{\C}} R_U(V_{\C}) \ge \Omega\left(\frac{(d^{2}-N)(d\Tr(O^2)-\Tr(O)^2)}{d^{5}}\right).
		\end{align}
		where $V_C$ refers to the learned unitary on the training dataset $\mathcal{D}_{\C}$ of size $N$, $d=2^n$ is the dimension of $n$-qubit quantum system, $\Omega$ means that the achieved lower bound omits the constant factor. 
		\label{thm:NFL_Ob_maintext}
	\end{theorem}
    The proof of Theorem~\ref{thm:NFL_Ob_maintext} is deferred to Appendix~D. In what follows, we combine the results established in the literature revolving around quantum state learning to comprehend the connection of quantum dynamics learning and quantum state learning, as well as the tightness of the achieved results of Theorem~\ref{thm:NFL_Ob_maintext}. Quantum state learning, also known as quantum state tomography, is a crucial task in quantum computing that aims to reconstruct quantum states with a specified level of precision by analyzing measurement outcomes from these states. For instance, when the observable $O$ is semi-definitely positive, the operator $U^{\dagger}OU/\Tr(O)$ could be treated as a mixed quantum state. It is straight to verify the property of quantum states $\Tr(U^{\dagger}OU/\Tr(O))=1$. In this regard, the task of learning quantum dynamics under the given observable $O$ is equivalent to the task of learning quantum states $U^{\dagger}OU/\Tr(O)$ with the measurement outputs $\bm{a}_j$ defined in Eqn.~(\ref{eq:data_c}) by regarding the input states $\ket{\bm{\psi}_j}\bra{\bm{\psi}_j}$ as the measurement operators. The optimal and tight query complexity for learning the quantum state with the finite number of non-adaptive measurements has been demonstrated to be $d^2/\varepsilon^2$, where $\varepsilon$ is the reconstruction precision \cite{lowe2022lower}.
    
    The result of Theorem~\ref{thm:NFL_Ob_maintext} indicates that achieving the zero generalization error requires the training dataset of at least $d^2=4^n$ examples, which scales exponentially with the number of qubits $n$. This lower bound matches the $d$-dependent term in the optimal lower bound $d^2/\varepsilon^2$ of learning the quantum state $U^{\dagger}OU/\Tr(O)$ discussed above.  The factor $1/\varepsilon^2$ is omitted in our discussion because our focus is on the learning protocols in the noiseless setting, where statistical measurement errors that typically cause this term are not considered.

    \subsection{NFL theorem for restricted quantum learning protocols}\label{subsec:RQNFL}
    For the Restricted Quantum Learning Protocols (ReQu-LPs) introduced in Section~\ref{subsec:RQLP},
    the NFL theorem is encapsulated in the following theorem, whose proof is deferred to Appendix~E.
    \begin{theorem}
		[NFL theorem for ReQu-LPs] 
        Let $O$ be an arbitrary observable. For the learning task defined in Eqn.~(\ref{eq:learning_model}), when the ReQu-LPs achieve the perfect training of Eqn.~(\ref{eq:pt_co}) with respect to the loss function defined in Eqn.~(\ref{eq:coherent_loss_setup}), the average risk function $R_U(V_{\RQ})$ defined in Eqn.~(\ref{eq:risk_function}) over all $n$-qubit unitaries $U$ and training datasets $\mathcal{D}_{\RQ}$ yields
		\begin{align}\label{eq:NFL_U}
		& \mathbb{E}_{U,\mathcal{D}_{\RQ}} R_U(V_{\RQ}) \ge \Omega \Bigg(\frac{(d^{2}-N^2)}{d^{5}} \cdot \mathbbm{1}(\bm{\alpha}\varpropto \mathbf{1}_N) 
        \nonumber \\
        & + \frac{(d^{2}-N)}{d^{5}} \cdot\mathbbm{1}(\bm{\alpha}\not\varpropto \mathbf{1}_N)\bigg) \times (d\Tr(O^2)-\Tr(O)^2),
		\end{align}
		where the equality holds for linearly independent states, $V_{\RQ}$ refers to the learned unitary on the training dataset $\mathcal{D}_{\RQ}$ of size $N$, $d=2^n$ is the dimension of the $n$-qubit quantum system, $\bm{\alpha}=(\bm{\alpha}_1, \cdots, \bm{\alpha}_N)$ and $\mathbf{1}_N$ refers to the all-ones vector of $N$-dimension, $\mathbbm{1}(\bm{\alpha}\varpropto \mathbf{1}_N)$ means that there exists a constant $c$ such that the phase vector $\bm{\alpha}=c\mathbf{1}_N$, i.e., $\bm{\alpha}_k=\bm{\alpha}_j$ for any $j,k\in [N]$. 
		\label{thm:NFL_U_maintext}
	\end{theorem}

    The achieved results in Theorem~\ref{thm:NFL_U_maintext} indicate the critical importance of capturing the inter-state relative phase over various input states, i.e., $\bm{\alpha}_k=\bm{\alpha}_j$ for all $j,k\in [N]$, in determining the generalization error. In particular, when the inter-state relative phases $\bm{\alpha}_j$ vary among input states, a training dataset of size $d^2$ is necessary to achieve zero prediction error, mirroring the same requirement in CLC-LPs. In contrast, if the inter-state relative phases over various input states are identical— a condition we term 'phase alignment' — then only $d$ examples are sufficient to attain zero prediction error. This represents a significant reduction in sample complexity from $d^2$ to $d$, offering an advantage of Re-LPs over CLC-LPs.

    We now demystify when ReQu-LPs could enable the condition of phase alignment. Particularly, for any training data $(\ket{\bm{\psi}_i}, U\ket{\bm{\psi}_i})$ and $(\ket{\bm{\psi}_j}, U\ket{\bm{\psi}_j})$ with $i\ne j$, we have
	\begin{align}\label{eq:phase_equality_text}
	\braket{\bm{\psi}_k | \bm{\psi}_j} = &	\bra{\bm{\psi}_k}U^{\dagger}U\ket{\bm{\psi}_j}
	\nonumber \\
	= & e^{i(\bm{\alpha}_j-\bm{\alpha}_k)}  \bra{\bm{\psi}_k}V_{\RQ}^{\dagger} V_{\RQ}\ket{\bm{\psi}_j} 
	\nonumber \\
	= & e^{i(\bm{\alpha}_j-\bm{\alpha}_k)} \braket{\bm{\psi}_k | \bm{\psi}_j}
	\end{align}
	where the second equality exploits the assumption of perfect training $U\ket{\bm{\psi}_j}= e^{i\bm{\alpha}_j}V_{\RQ} \ket{\bm{\psi}_j}$. Eqn.~(\ref{eq:phase_equality_text}) implies that $\bm{\alpha}_{k}=\bm{\alpha}_j$ holds for all $k,j\in [N]$ if $\braket{\bm{\psi}_k|\bm{\psi}_j}\ne 0$ or equivalently the input states are not orthogonal. This indicates that non-orthogonal training states warrant the condition enabling phase alignment, which leads to improved generalization performance. Moreover, the non-orthogonal states can be obtained by randomly sampling Haar states in $\mathbb{S}\mathbb{U}(2^n)$. Remarkably, while arbitrary two Haar states have been shown nearly orthogonal in the Hilbert space of a large quantum system, namely the trace distance between them is exponentially small with the number of qubits $n$, the probability of arbitrary two Haar states in $\mathbb{S}\mathbb{U}(2^n)$ being strictly orthogonal is zero.
    
    \subsection{NFL theorem for quantum learning protocols}\label{subsec:QNFL}
    We now elucidate the NFL theorem for Quantum Learning Protocols (Qu-LPs) with access to $U^{\dagger}$.  
    Under the learning protocol detailed in Section~\ref{subsec:QLP}, we prove the following quantum NFL theorem, where
    the formal statement and proof are deferred to Appendix~F.

    \begin{theorem}
		[NFL for Qu-LPs with access to $U^{\dagger}$] 
        Let $O$ be an arbitrary observable. For the learning task defined in Eqn.~(\ref{eq:learning_model}), when the access to the inverse of target unitary $U^{\dagger}$ is allowed and the perfect training of Eqn.~(\ref{eq:pt_co}) is achieved, the average risk function $R_U(V_{\Q})$ defined in Eqn.~(\ref{eq:risk_function}) over all $n$-qubit unitaries $U$ and training sets $\mathcal{D}_{\Q}$ yields
		\begin{align}
		& \mathbb{E}_{U,\mathcal{D}_{\Q}} R_U(V_{\Q}) \ge \Omega \bigg(\frac{(d-N)(d\Tr(O^2)-\Tr(O)^2)}{d^{4}}   
        \nonumber \\
        & + \frac{N^2}{d^4} \sum_{k\ne j}^d O_{kj}^2 \cdot \mathbbm{1}(\bm{\beta} \not\varpropto \mathbf{1}_N) \bigg) \label{eq:NFL_U_dagger}
		\end{align}
		where the equality holds for linearly independent states, $V_{\Q}$ refers to the learned unitary on the training dataset $\mathcal{D}_{\Q}$ of size $N$, $d=2^n$ is the dimension of $n$-qubit quantum system, $\bm{\beta}=(\bm{\beta}_1, \cdots, \bm{\beta}_N)$ and $\mathbf{1}_N$ refers to the all-ones vector of $N$-dimension, $\mathbbm{1}(\bm{\beta}\varpropto \mathbf{1}_N)$ means that there exists a constant $c$ such that the phase vector $\bm{\beta}=c\mathbf{1}_N$, i.e., $\bm{\beta}_k=\bm{\beta}_j$ for any $j,k\in [N]$. 
		\label{thm:NFL_U_dagger_maintext}
	\end{theorem}
    The results established in Theorem~\ref{thm:NFL_U_dagger_maintext} reveal that the learning performance of Qu-LPs depends on the diagonality of observable $O$ and the condition of phase alignment simultaneously. Specifically, when the observable $O$ is diagonal, the sample complexity for achieving zero prediction error always yields the scaling of $d=2^n$, regardless of whether the condition of phase alignment $\mathbbm{1}(\bm{\beta}\varpropto \mathbf{1}_N)$ is met. Moreover, for the non-diagonal observable, the average risk increases with the magnitude of the square sum of the non-diagonal elements of the observable $O$, particularly when the condition of phase alignment $\mathbbm{1}(\bm{\beta}\varpropto \mathbf{1}_N)$ is not satisfied. Consequently, a larger training dataset is required for the non-diagonal observable to achieve zero prediction error.

    \noindent \textbf{Remark.} The criteria for achieving phase alignment in Qu-LPs are identical to those in ReQu-LPs with access to $U$, specifically requiring the input states to be non-orthogonal.

    \subsection{Separation in various learning protocols}\label{subsec:sep_theoretic}
    In this subsection, we combine the established NFL theorems for various learning protocols in Theorems~\ref{thm:NFL_Ob_maintext}-\ref{thm:NFL_U_dagger_maintext} to comprehend the separable ability between various learning protocols.

    \smallskip
    
    \noindent \textbf{Separation between CLC-LPs and quantum-related learning protocols (ReQu-LPs and Qu-LPs).} When the training states are linearly independent and non-orthogonal such that the phase alignment $\mathbbm{1}(\bm{\alpha}\varpropto \mathbf{1}_N)$ or $\mathbbm{1}(\bm{\beta}\varpropto \mathbf{1}_N)$ is satisfied, the sample complexity required to achieving the zero prediction error scales with $d^2$ for CLC-LPs (Theorem~\ref{thm:NFL_Ob_maintext}) and $d$ for ReQu-LPs (Theorem~\ref{thm:NFL_U_maintext}) and Qu-LPs (Theorem~\ref{thm:NFL_U_dagger_maintext}). This suggests a quadratic reduction from CLC-LPs to ReQu-LPs as well as Qu-LPs. On the other hand, when the training states are orthogonal, the achieved lower bound for the ReQu-LPs is the same as that achieved in CLC-LPs, indicating the vanishing separation between the two learning protocols. 

    \smallskip

    \noindent \textbf{Separation between ReQu-LPs and Qu-LPs.}
    When the training states are orthogonal, compared to the achieved results for ReQu-LPs in Theorem~\ref{thm:NFL_U_maintext}, Qu-LPs with access to the inverse of target unitary $U^{\dagger}$ could still reduce the sample complexity from $d^2$ to $d$ for the diagonal observable $O$, regardless of the orthogonality of training states (or phase alignment). This is because the term $\sum_{k\ne j}^d O_{kj}^2$ related to the indicator of phase non-alignment $\mathbbm{1}(\bm{\beta}\not \varpropto \mathbf{1}_N)$ in Eqn.~(\ref{eq:NFL_U_dagger}) is always zero for the diagonal observable $O$. This observation provides another piece of evidence for achieving quantum advantages of employing the transformation of the target quantum system, as demonstrated by recent studies from the perspective of query complexity \cite{salmon2023provable, king2024exponential}. 
    Moreover, when the observable is non-diagonal, the separation in terms of sample complexity between the quantum learning protocols with access to $U$ (ReQu-LPs) and $U^{\dagger}$ (Qu-LPs) will diminish as the magnitude of the non-diagonal elements of observable $O$ increases.

    \smallskip

     \begin{figure*}[htbp]
		\centering
		\includegraphics[width=0.98\textwidth]{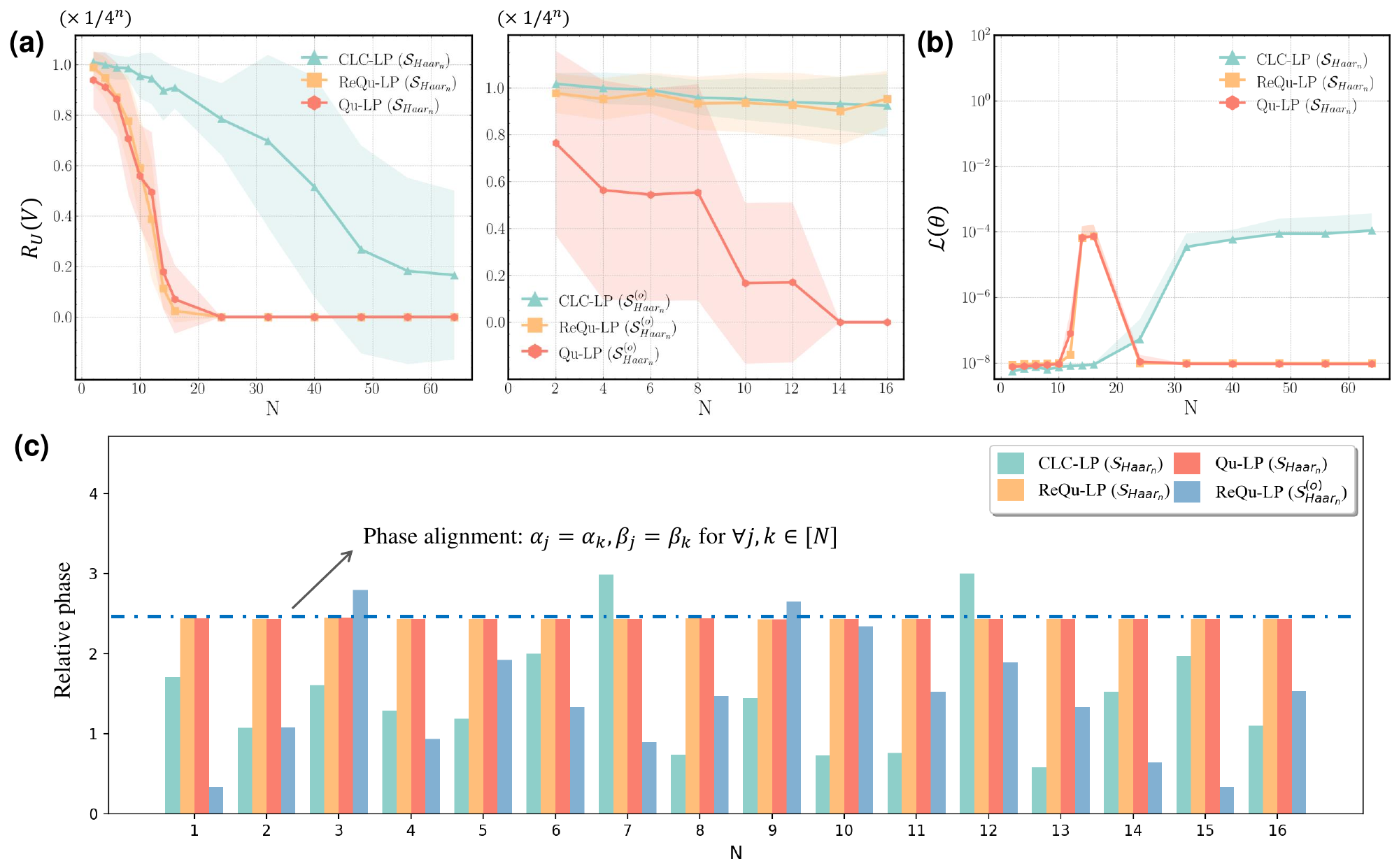}
		\caption{\small{\textbf{Numerical results for quantum dynamics learning.  
				} \textbf{(a)} The averaged prediction error with varying training dataset size $N$ under various learning protocols and different types of input states when the number of qubits $n=4$. \textbf{(b)} The averaged training error with varying training dataset size $N$ under various LPs when $n=4$. \textbf{(c)} The magnitude of the inter-state relative phase defined in Eqn.~(\ref{eq:pt_co}) and Eqn.~(\ref{eq:pt_co_dagger}) for various LPs.
                The labels `CLC-LP', `ReQu-LP', `Qu-LP' refer to the classical learning protocol, restricted quantum learning protocol, and quantum learning protocol with access to $U^{\dagger}$. $R_U(V)$ and $\mathcal{L(\bm{\theta})}$ refers to the prediction error and training error. The label `$(\times 1/4^n)$’ means that the plotted prediction error is normalized by a multiplier factor $1/4^n$.
     The label `$\mathcal{S}_{\haar_n}^{(o)}$' and `$\mathcal{S}_{\haar_n}$' refer to the input states being the orthogonal $n$-qubit Haar states and the $n$-qubit Haar states.}}
		\label{fig:Ob_basis_0}
	\end{figure*}
 
    \noindent \textbf{The impact of phase alignment on the separation.} As discussed in Section~\ref{subsec:RQNFL}, there is an exact correspondence between the learning performance of ReQu-LPs and the phenomenon of phase alignment. Specifically, when the phase alignment $\mathbbm{1}(\bm{\alpha}\not \varpropto \mathbf{1}_N)$ does not happen, the learning performance of ReQu-LPs is reduced to the same as that of CLC-LPs. Furthermore, the occurrence of phase alignment is precluded for CLC-LPs, as they inherently cannot capture the inter-state relative phase between training states due to the measurement-first paradigm inherent in their learning process, wherein the information of global phase of output states is lost during the measurement process represented by $\Tr(U^{\dagger}OU\ket{\bm{\psi}_j}\bra{\bm{\psi}_j})$.
    
    \begin{figure}  
		\centering
		\includegraphics[width=0.48\textwidth]{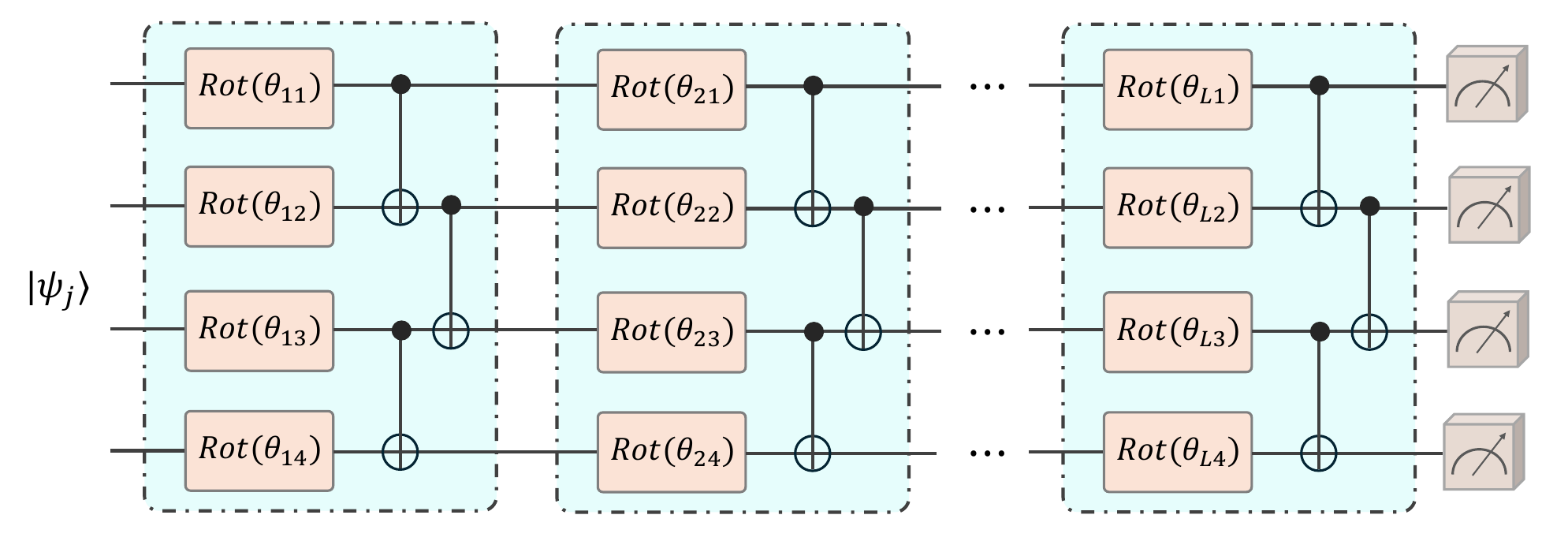}
		\caption{\small{\textbf{Structure of hardware efficient ansatz (HEA)}. $\Rot(\bm{\theta}_{j1})$ represents the general rotation gate with three rotation angles $\bm{\theta}_x, \bm{\theta}_y, \bm{\theta}_z$. }}
		\label{fig:CircStruc}
	\end{figure}
	
	\section{Experiments} \label{sec:numerics}
	We conduct numerical simulations to verify the separation in learning performance between various learning protocols discussed above, showing the significance of capturing the phase information during the training process. All numerical simulations are implemented by Pennylane frameworks \cite{bergholm2018pennylane}. 

    \begin{figure*}[htbp]
		\centering
		\includegraphics[width=0.98\textwidth]{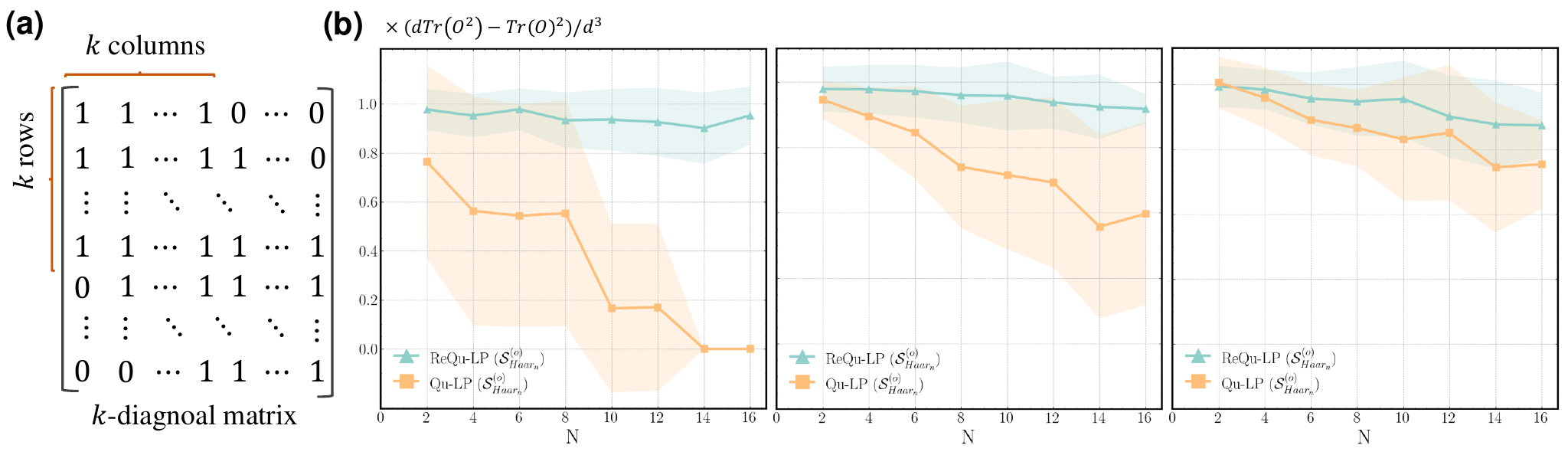}
		\caption{\small{\textbf{Numerical results for non-diagonal matrix.  
				} \textbf{(a)} The $k$-diagonal matrix with $k$ referring to the diagonal offset. \textbf{(b)} The averaged prediction error with varying training dataset size $N$ under ReQu-LPs and Qu-LPs with access to $U^{\dagger}$ when the number of qubits $n=4$. The panels from left to right correspond respectively to the $k$-diagonal observables of rank $5$, with $k$ values of $1$, $2$, and $4$. The notations are identical to those in Fig.~\ref{fig:Ob_basis_0}.}}
		\label{fig:non_diag_ob}
	\end{figure*}
 
	\subsection{Datasets}
    The target unitary is an $n$-qubit Haar random unitary $U$ with $n\in \{4,5\}$. The training datasets corresponding to the CLC-LPs, ReQu-LPs, and Qu-LPs take the form of Eqn.~(\ref{eq:data_c}), Eqn.~(\ref{eq:data_RQ}), and Eqn.~(\ref{eq:data_Q}). Two types of states $\ket{\bm{\psi}_j}$ sampling from different distributions are considered. The first is the class of orthogonal Haar random $n$-qubit states $\mathcal{S}_{\haar_n}^{(o)}=\{(\ket{\bm{\psi}_1}, \cdots, \ket{\bm{\psi}_N}): \braket{\bm{\psi}_i|\bm{\psi}_j}=\sigma_{ij}\}$. In this case, one can first sample a Haar unitary and random uniformly choose $N$ column vectors of this unitary as the input states. This leads to the training dataset of maximal size less than $2^n$.  The second is the class of  $n$-qubit Haar random states $\mathcal{S}_{\haar_n}$. The training dataset size takes $N \in \{1,2,\cdots,16\}$ for $\mathcal{S}_{\haar_n}^{(o)}$ and $N\in \{1,2, \cdots, 64\}$ for $\mathcal{S}_{\haar_n}$. As discussed in Section~\ref{subsec:RQNFL}, the probability of the Haar random states being strictly orthogonal is almost zero. 
 
    \subsection{Model construction and evaluation metric}
    \textit{Circuit optimization:} We take the same hardware efficient ansatz (HEA) $V(\bm{\theta})$ as the parameterized quantum unitary towards achieving the perfect training in various learning protocols, as shown in Fig.~\ref{fig:CircStruc}. The HEA consists of single rotation gates and CNOT gates. The number of layers is set as $L=30$ for $n=4$ and $L=42$ for $n=5$. The ADAM optimizer is employed to update the parameters in the HEA $V(\bm{\theta})$ with learning rate $\eta=0.01$. The maximal iteration step is set as $T=1000$ for $n=4$ and $T=500$ for $n=5$.
    
    \textit{Evaluation metric:} To show the separation between the classical learning protocols and quantum learning protocols, we
	record the averaged prediction error and averaged training error defined in Eqn.~(\ref{eq:risk_function}) by learning four different unitaries $U$ with 10 random training datasets. Moreover, the inter-state relative phases between the output states of the target unitary and trained parameterized unitary are recorded to show the signification of employing the information of phase in learning protocols.

    \subsection{Experimental results} 
    \noindent \textbf{Separation between the CLC-LPs, ReQu-LPs, and Qu-LPs.} The simulation results for the observable $O$ being projective measurement $(\ket{\bm{0}}\bra{\bm{0}})^{\otimes n}$ are displayed in Fig.~(\ref{fig:Ob_basis_0}) (a)-(c). In particular, Fig.~\ref{fig:Ob_basis_0} (a) indicates that the Qu-LPs could achieve a significant reduction in the sample complexity for achieving the zero prediction error from $N=64$ to $N=16$, compared to CLC-LPs when the input states are sampled from Haar distribution. 
    Moreover, when the input states are orthogonal, the prediction performance of the ReQu-LPs with access to $U$ is the same as the CLC-LPs, while the Qu-LPs with access to the inverse of the target unitary $U^{\dagger}$ retains the same prediction performance as the case of using Haar input states, achieving quantum advantages over the other two learning protocols. In particular, when the training dataset size is $N=16$, the Qu-LPs with access to $U^{\dagger}$ reach zero prediction error, and the CLC-LPs and ReQu-LPs only have a small reduction of the prediction error than that of $N=1$ and get a prediction error higher than $0.8$. The same phenomenon also happens in the case of $n=5$, as shown in Table~\ref{tab:5-qubit}. These numerical results echo with the theoretical results established in Section~\ref{sec:NFL} for the projective measurement.  

    \smallskip

    \begin{table*}  
        \centering
        \renewcommand\arraystretch{1.2}
        \caption{Averaged prediction error normalized with a multiplier factor $1/4^n$ when the number of qubits $n=5$.}
        \begin{tabular}{l|cccc|cccc}
        \toprule[1pt]
             \multirow{2}*{\makecell[l]{Learning \\ Protocols}} & \multicolumn{4}{c}{$\mathcal{S}_{\haar_n}$} & \multicolumn{4}{c}{$\mathcal{S}_{\haar_n}^{(o)}$}\\
             \cline{2-5} \cline{6-9}
             & $N=4$ & $N=32$ & $N=64$ & $N=128$ & $N=4$ & $N=16$ & $N=24$ & $N=32$ \\
              \midrule
             Classical & 0.635 & 0.633 & 0.589 & \cellcolor{black!20}0 & 0.619 & 0.635 & 0.624 & 0.621 \\ 
             Quantum-$U$ & 0.632 & \cellcolor{black!20}0.236 & 0.093 & 0.061 & 0.632 & 0.624 & 0.590 & 0.609  \\ 
             Quantum-$U^{\dagger}$  & \cellcolor{black!20}0.628& 0.243 & \cellcolor{black!20}0.075 & 0.079 & \cellcolor{black!20}0.615 & \cellcolor{black!20}0.605 & \cellcolor{black!20}0.266 & \cellcolor{black!20}0.043  \\ 
        \bottomrule[1pt]
        \end{tabular}
        \label{tab:5-qubit}
    \end{table*}

    \begin{table*}  
        \centering
        \renewcommand\arraystretch{1.2}
        \caption{Averaged prediction error normalized with a multiplier factor $1/4^n$ with a various number of layers of PQC when the number of qubits $n=4$. The minimum prediction error under various hyper-parameter settings is highlighted in gray color.}
        \begin{tabular}{l|cccc|cccc|cccc}
        \toprule[1pt]
             \multirow{2}*{\makecell[l]{Learning \\ Protocols}} & \multicolumn{4}{c}{$L=5$} & \multicolumn{4}{c}{$L=10$} & \multicolumn{4}{c}{$L=20$} \\
             \cline{2-5} \cline{6-9} \cline{10-13}
             & $N=2$ & $N=16$ & $N=32$ & $N=64$  & $N=2$ & $N=16$ & $N=32$ & $N=64$ & $N=2$ & $N=16$ & $N=32$ & $N=64$ \\
              \midrule
             Classical & \cellcolor{black!20}0.907 & 0.835 & \cellcolor{black!20}0.465 & \cellcolor{black!20}0 & \cellcolor{black!20}0.955 & 0.933 & 0.658 & \cellcolor{black!20}0.148 & 1.027 & 0.886 & 0.514 & 0.134 \\ 
             Quantum-$U$ & 0.968 & 0.843 & 0.854 & 0.733 & 0.995 & \cellcolor{black!20}0.512 & 0.347 & 0.263 & 1.041 & \cellcolor{black!20}0.124 & \cellcolor{black!20}0.007 & \cellcolor{black!20}0.004 \\ 
             Quantum-$U^{\dagger}$  & 0.951 & \cellcolor{black!20}0.756 & 0.676 & 0.776 & 0.956 & 0.554 & \cellcolor{black!20}0.325 & 0.319 & \cellcolor{black!20}0.991 & 0.215 &  \cellcolor{black!20}0.007 & \cellcolor{black!20}0.004 \\ 
        \bottomrule[1pt]
        \end{tabular}
        \label{tab:varyin_layer}
    \end{table*}
    
    \noindent \textbf{The significance of phase alignment.} Fig.~\ref{fig:Ob_basis_0}(c) shows the consistency of the phase alignment and the learning performance shown in Fig.~\ref{fig:Ob_basis_0}(a), which highlights the significance of learning phase information in incurring the separation in sample complexity between CLC-LPs and ReQu-LPs. In particular, when the ReQu-LPs perform better with a smaller prediction error than the CLC-LPs for the Haar input states, the inter-state relative phase $\bm{\alpha}_j$ and $\bm{\beta}_j$ over various input states defined in Eqn.~(\ref{eq:pt_co}) and  Eqn.~(\ref{eq:pt_co_dagger}) are all equal, while the CLC-LPs fails to do this. On the other hand, when the ReQu-LPs can not enable the inter-state relative phase over orthogonal states to be equal, the average prediction error of such ReQu-LPs is reduced to the same trend as that of the CLC-LPs.

    \smallskip
    
    \noindent \textbf{The results for non-diagonal observables.} We conduct numerical simulations on the non-diagonal observable to show the impact of the magnitude of non-diagonal elements of the fixed observable $O$ on the prediction performance of Qu-LPs with access to $U^{\dagger}$. In particular, we consider the observable $O$ of rank $r$, where the first $r$-dimensional sub-matrix is $k$-diagonal matrix of elements $1$ and the other elements in $O$ are zero, as shown in Fig.~\ref{fig:non_diag_ob}(a). The rank is set as $r=5$ and the diagonal offset is set as $k\in \{1, 2, 4\}$. The setting of other hyper-parameters is the same as the case of $n=4$. The numerical results are shown in Fig.~\ref{fig:non_diag_ob}(b). In particular, the separation in terms of the prediction error between the ReQu-LPs and Qu-LPs with access to $U^{\dagger}$ is narrowed as the diagonal offset $k$, or equivalently the magnitude of non-diagonal elements, increases from the case of $k=1$ to the case of $k=4$. 

    \smallskip
    
    \noindent \textbf{The impact of the number of layers of PQCs.} 
    The number of layers in PQCs significantly influences the training error across different learning protocols. This relationship has been underscored by the studies in the context of over-parameterization theory of quantum neural networks \cite{larocca2021theory, wang2022symmetric, liu2023analytic}, indicating that PQCs must possess a parameter count of $4^n$ to ensure global convergence---a prerequisite for achieving perfect training---particularly in tasks of unitary learning with quantum learning protocols. To highlight the importance of the perfect training, we conduct numerical simulation on various amounts of layers for the case of $n=4$, where the number of layers is set as $L=\{5,10,20\}$. The input states are sampled from the Haar distribution and the other hyper-parameters are set as the same as the case of $n=4$. The numerical results are shown in Table~\ref{tab:varyin_layer}. In particular, when the PQCs are shallow with $L=5$, the CLC-LPs perform a smaller prediction error than Qu-LPs. Moreover, the CLC-LP achieves zero prediction error when $N=64$, while ReQu-LPs and Qu-LPs still retain a large prediction error over $0.73$. This is because the CLC-LPs require fewer parameters to reach the over-parameterization regime than the two quantum-related learning protocols such that when $L=5$ the former achieves perfect training regime, while the latter two learning protocols get stuck in local minima. Moreover, when the PQCs are deep enough with $L=20$ such that all learning protocols achieve perfect training, the ReQu-LPs and Qu-LPs perform better than CLC-LPs as discussed above, where achieving a small prediction error less than $0.14$ requires the training dataset size of $N=16$ for CLC-LPs and of $N=64$ for ReQu-LPs and Qu-LPs. 
	
	\medskip 
 
	\section{Conclusion} \label{sec:conclusion}
    In this article, we rigorously explored the universal learning performance of three distinct learning protocols for learning a target unitary under a fixed observable—namely, classical learning protocols (CLC-LPs), restricted quantum learning protocols (ReQu-LPs), and quantum learning protocols (Qu-LPs)—each employing quantum resources in ascending order. By establishing no-free-lunch (NFL) theorems tailored to these protocols, our theoretical results reveal a quadratic separation in terms of sample complexity between CLC-LPs, ReQu-LPs, and Qu-LPs. Our analysis suggests that the separation between CLC-LPs and the two quantum-related learning protocols originates from their differing capabilities in capturing the inter-state relative phase—a feature that ReQu-LPs and Qu-LPs possess when training states are non-orthogonal, but which CLC-LPs lack regardless of the training states. The conducted numerical simulations accord with the theoretical result. The results achieved in this study not only deepen our understanding of the capabilities of quantum learning protocols from both problem- and data-dependent perspectives but also elucidate the inherent quantum mechanisms that enable these capabilities. This insight provides practical guidance for designing advanced quantum learning protocols.

	\clearpage
	\newpage
	\appendix
	\onecolumngrid 
	
	\tableofcontents
	
	\medskip
	
	\noindent \textit{\textbf{Roadmap:}} Appendix~\ref{app_sec:notation} introduces the notations used throughout the whole paper. Some useful properties about the calculation of integration with respect to Haar states and Haar unitaries are provided in Appendix~\ref{app_sec:prelimilaries}. The remaining portions are devoted to theoretical results and mathematical proofs. Appendix~\ref{app_sec:central_lemma} provides a central lemma that reduces the problem of deriving a lower bound considered in No-Free-Lunch (NFL) theorems to an upper bounding problem. The proof of NFL theorems for classical learning protocols, restricted quantum learning protocols, and quantum learning protocols are given in Appendix~\ref{app_sec:NFL_SC},~\ref{app_sec:NFL_SQ}, and~\ref{app_sec:NFL_SQ_dagger}. 
	
	\section{Notation}\label{app_sec:notation}
	We unify the notations throughout the whole paper. The number of qubits and training data size are denoted by $n$ and $N$, respectively. The corresponding dimension of Hilbert space to the $n$-qubit system is denoted by $d=2^n$. Let $\mathcal{H}_d$ denote a $d$-dimensional Hilbert space. The notation $[N]$ refers to the set $\{1, 2, \cdots, N\}$. We denote $\|\cdot\|_1$ as the trace norm, i.e. the sum of the absolute singular values, and $\|\cdot\|_{\Fnorm}$ as the Frobenius norm. The cardinality of a set is denoted by $|\cdot|$. We use the standard bra-ket notation for pure quantum states. The training dataset is denoted by $\mathcal{D}$ with proper subscripts to distinguish various learning protocols. We interchangeably use the notation $\mathbbm{1}(x=y)$ and $\delta_{xy}$ to refer to the indicator function, which takes the value of $1$ if $x=y$ and $0$ otherwise. We denote $\mathbb{I}_N$ and $\mathbf{1}_N$ as the $N$-dimensional identity matrix and the $N$-dimensional all-ones vector, i.e., $\mathbf{1}_N=(1,1,\cdots,1)$, respectively.
	
	\section{Haar integration}\label{app_sec:prelimilaries}
	In this section, we present a set of properties that enable the analytical computation of integrals of polynomial functions over the unitary group and the quantum state with respect to the unique normalized Haar measure. For a more comprehensive discussion on this subject, we refer the readers to Refs.~\cite{collins2006integration,puchala2011symbolic,cerezo2021cost,zhao2022hamiltonian}.

	\begin{prop}[Lemma~1 of Ref.~\cite{cerezo2021cost}]
    \label{prop:Tr(WW)}
		Let $\{ W_y\}_{y\in Y} \subset \mathbb{U}(d)$ form a unitary $t$-design with $t\ge1$, and let $A, B: \mathcal{H}_d\to\mathcal{H}_d$ be arbitrary linear operators. Then
		\begin{equation}
			\frac{1}{|Y|}\sum_{y\in Y}\Tr[W_{y}AW_{y}^{\dagger}B]=\int_{\haar}\mathrm{d}W\Tr[WAW^{\dagger}B]=\frac{\Tr[A]\Tr[B]}{d}.
		\end{equation}
	\end{prop}

    \begin{prop}[Lemma~2 of Ref.~\cite{cerezo2021cost}]
    \label{prop:Tr(WWWW)}
		Let $\{ W_y\}_{y\in Y} \subset \mathbb{U}(d)$ form a unitary $t$-design with $t\ge 2$, and let $A, B, C, D: \mathcal{H}_d\to\mathcal{H}_d$ be arbitrary linear operators. Then
		\begin{align}
			\frac{1}{|Y|}\sum_{y\in Y}\Tr[W_{y}AW_{y}^{\dagger}BW_{y}CW_{y}^{\dagger}D]=&\int_{\haar}\mathrm{d}W \Tr[WAW^{\dagger}BWCW^{\dagger}D] \nonumber \\
			=&\frac{1}{d^2-1}\left(\Tr[AC]\Tr[B]\Tr[D] + \Tr[A]\Tr[C]\Tr[BD] \right) \nonumber \\
			& - \frac{1}{d(d^2-1)}\left(\Tr[A]\Tr[B]\Tr[C]\Tr[D] + \Tr[AC]\Tr[BD] \right).
		\end{align}
	\end{prop}
 
	\begin{prop}[Lemma~3 of Ref.~\cite{cerezo2021cost}]
    \label{prop:Tr(WW)Tr(WW)}
		Let $\{ W_y\}_{y\in Y} \subset \mathbb{U}(d)$ form a unitary $t$-design with $t\ge 2$, and let $A, B, C, D: \mathcal{H}_d\to\mathcal{H}_d$ be arbitrary linear operators. Then
		\begin{align}
			\frac{1}{|Y|}\sum_{y\in Y}\Tr[W_{y}AW_{y}^{\dagger}B]\Tr[W_{y}CW_{y}^{\dagger}D]=&\int_{\haar}\mathrm{d}W \Tr[WAW^{\dagger}B]\Tr[WCW^{\dagger}D] \nonumber \\
			=&\frac{1}{d^2-1}\left(\Tr[A]\Tr[B]\Tr[C]\Tr[D] + \Tr[AC]\Tr[BD] \right) \nonumber \\
			& - \frac{1}{d(d^2-1)}\left(\Tr[AC]\Tr[B]\Tr[D] + \Tr[A]\Tr[C]\Tr[BD] \right).
		\end{align}
	\end{prop}
	
	\begin{prop}\label{prop:Tr(psipsi)}
		Let the quantum state $\ket{\bm{\phi}} \in \mathcal{H}_d$ follows the Haar distribution, and let $A: \mathcal{H}_d\to\mathcal{H}_d$ be arbitrary linear operators. Then
		\begin{equation}
			\int_{\haar}\ket{\bm{\phi}}\bra{\bm{\phi}} \mathrm{d} \bm{\phi} = \frac{\mathbb{I}_d}{d}, \quad  \mbox{and hence}  \quad \int_{\haar}\mathrm{d} \bm{\phi} \Tr\left(\ket{\bm{\phi}}\bra{\bm{\phi}}A \right)
			= \frac{\Tr(A)}{d}.
		\end{equation}
	\end{prop}
	
	\begin{prop}\label{prop:Tr(psipsi)Tr(psipsi)}
		Let the quantum state $\ket{\bm{\phi}} \in \mathcal{H}_d$ follows the Haar distribution, and let $A: \mathcal{H}_d\to\mathcal{H}_d$ be arbitrary linear operators. Then
		\begin{equation}
			\int_{\haar}\ket{\bm{\phi}}\bra{\bm{\phi}}^{\otimes 2} \mathrm{d} \bm{\phi} = \frac{\mathbb{I}_d^{\otimes 2}+\swap}{d(d+1)} \quad  \mbox{and hence}  \quad 
			\int_{\haar}\mathrm{d} \bm{\phi} \Tr\left(\ket{\bm{\phi}}\bra{\bm{\phi}}A \right)^2
			= \frac{\Tr(A)^2+\Tr(A^2)}{d(d+1)}.
		\end{equation}
	\end{prop}

	\section{Proof of Lemma 1}\label{app_sec:central_lemma}
	We first recall that the learning problem considered in this study aims to predict the output of quantum states evolved by an unknown unitary $U$ under a given observable $O$, i.e.,
    \begin{align}\label{eq:app_learning_model}
	f_{U}(\bm{\psi}) = \Tr(OU\ket{\bm{\psi}}\bra{\bm{\psi}}U^{\dagger}),
	\end{align} 
    where $\ket{\bm{\psi}}$ is the input state. Denote the hypothesis learned from the training dataset $\mathcal{D}$ by $h_{\mathcal{D}}(\bm{\psi})=\Tr(V_{\mathcal{D}}^{\dagger}OV_{\mathcal{D}}\ket{\bm{\psi}}\bra{\bm{\psi}})$ with $V_{\mathcal{D}}$ being the trained unitary. To quantify how well the hypothesis $h_{\mathcal{D}}$ performs in predicting the target concept $f_{U}$, we define the risk function as 
    \begin{equation}\label{eq:app_risk}
        R_{U}(V_{{\mathcal{D}}}) := \int \mathrm{d} \bm{\psi} \left( f_U(\bm{\psi})-h_{{\mathcal{D}}}(\bm{\psi})  \right) ^2 = \mathbb{E}_{\ket{\bm{\psi}} \sim \haar} \Tr\left(O\left(V_{\mathcal{D}}\ket{\bm{\psi}}\bra{\bm{\psi}}V_{\mathcal{D}}^{\dagger}-U\ket{\bm{\psi}}\bra{\bm{\psi}}U^{\dagger}\right)\right) ^2.
    \end{equation}

 The no-free-lunch (NFL) theorem considers the lower bound of the average risk function over all possible target unitary and training datasets, i.e., $\mathbb{E}_{U}\mathbb{E}_{\mathcal{D}}R_U(V_{\mathcal{D}})$, which is a hard problem to analyze in both quantum and classical learning theory. In this regard, we give the following lemma to reduce the problem of lower bounding the risk function to the problem of a manageable upper bound problem regarding the trace function $\Tr(U^{\dagger}OUV_{\mathcal{D}}^{\dagger}OV_{\mathcal{D}})$.
	\begin{lemma}\label{lem:app_upper_to_lower}
		The averaged risk function defined in Eqn.~(\ref{eq:app_risk}) over the Haar input state yields
		\begin{equation}\label{eq:app_risk_trans}
			R_U(V_{\mathcal{D}}) = \frac{1}{d(d+1)}\left[2\Tr(O^2)-2\Tr(U^{\dagger}OUV_{\mathcal{D}}^{\dagger}OV_{\mathcal{D}}) \right]
		\end{equation}
	\end{lemma} 
	
	\begin{proof}	
		[Proof of Lemma~\ref{lem:upper_to_lower}]
		Define $G:=U^{\dagger}OU-V_{\mathcal{D}}^{\dagger}OV_{\mathcal{D}}$. Then we have
		\begin{align}
			R_U(V_{\mathcal{D}}) 
			& =  \mathbb{E}_{\ket{\bm{\psi}} \sim \haar} \Tr\left[ G  \ket{\bm{\psi}}\bra{\bm{\psi}} \right]^2 \nonumber \\
			& = \Tr\left[ G^{\otimes 2} \mathbb{E}_{\ket{\bm{\psi}} \sim \haar} \ket{\bm{\psi}}\bra{\bm{\psi}}^{\otimes 2} \right] \nonumber \\
			& = \frac{1}{d(d+1)}\Tr\left[ G^{\otimes 2} (\mathbb{I}_d^{\otimes 2} + \swap)\right] \nonumber \\
			& = \frac{1}{d(d+1)}\left[\Tr(G)^2+\Tr(G^2)\right]\nonumber \\
			& = \frac{1}{d(d+1)}\left[2\Tr(O^2)-2\Tr(U^{\dagger}OUV_{\mathcal{D}}^{\dagger}OV_{\mathcal{D}}) \right], \label{eq:app_risk_simplify}
		\end{align}
		where the second equality employs $\Tr(A)^2=\Tr(A\otimes A)$, the third equality exploits Property~\ref{prop:Tr(psipsi)Tr(psipsi)}, and the $\swap$ is the swap operator on the space $\mathcal{H}_d^{\otimes 2}$.	
	\end{proof}

	\section{Proof of Theorem 1 (NFL theorem for classical learning protocols)}\label{app_sec:NFL_SC}
	We first recall the definition of classical learning protocols, which is specified by the training dataset and the assumption of perfect training. In particular, for a given the observable $O$ with the spectral decomposition $O=\sum_{q=1}^{r} \lambda_q\ket{\bm{o}_q}\bra{\bm{o}_q}$, the training pairs consist of the input states $\ket{\bm{\psi}_j}$ and the expectation of the projective measurement $\ket{\bm{o}_q}\bra{\bm{o}_q}$ on the output states $U\ket{\bm{\psi}_j}$. This leads to the training dataset as
\begin{equation}\label{eq:app_data_c}
    \mathcal{D}_{\C}=\{(\ket{\bm{\psi}_j}, \bm{a}_j)~|~\ket{\bm{\psi}_j} \in \mathcal{H}_{d},\bm{a}_j=(\bm{a}_{j1}, \cdots, \bm{a}_{jr}) \}_{j=1}^N,
    \end{equation}
    where $\bm{a}_{jq}=\Tr(\ket{\bm{o}_q}\bra{\bm{o}_q}U\ket{\bm{\psi}_j}\bra{\bm{\psi}_j}U^{\dagger})$ refers to the expectation value of the projective measurement $\ket{\bm{o}_q}\bra{\bm{o}_q}$. In this regard, the assumption of perfect training is given by 
    \begin{align}       \Tr(U^{\dagger}\ket{\bm{o}_q}\bra{\bm{o}_q}U\ket{\bm{\psi}_j}\bra{\bm{\psi}_j})&=\Tr(V_{\C}^{\dagger}\ket{\bm{o}_q}\bra{\bm{o}_q}V_{\C}\ket{\bm{\psi}_j}\bra{\bm{\psi}_j}), \mbox{~for~} \forall j\in[N], \forall  q\in [r],
    \label{eq:app_pt_inco}
    \end{align}
    where $V_{\C}$ is the trained unitary on the training dataset $\mathcal{D}_C$. We are now ready to give the proof of Theorem~\ref{thm:NFL_Ob_maintext}.
	
	\begin{theorem-non}[Formal statement of Theorem~\ref{thm:NFL_Ob_maintext}]
        Let $O=\sum_{q=1}^{r} \lambda_q\ket{\bm{o}_q}\bra{\bm{o}_q}$ be an arbitrary observable.
		For the learning model defined in Eqn.~(\ref{eq:app_learning_model}), when the classical learning protocols achieve the perfect training of Eqn.~(\ref{eq:app_pt_inco}), the average risk function over all possible $n$-qubit unitaries $U$ and training datasets $\mathcal{D}_{\C}$ defined in Eqn.~(\ref{eq:app_data_c}) yields
		\begin{align}
		\mathbb{E}_{U,\mathcal{D}_{\C}} R_U(V_{\C}) \ge \Omega\left(\frac{(d^{2}-N)(d\Tr(O^2)-\Tr(O)^2)}{d^{5}}\right).
		\end{align}
		where $d$ and $N$ refer to the dimension of an $n$-qubit quantum system with $d=2^n$ and the size of the training dataset, respectively. 
	\end{theorem-non}
	
	\begin{proof}
		[Proof of Theorem~\ref{thm:NFL_Ob_maintext}]  Following Lemma~\ref{lem:upper_to_lower}, a critical step for deriving the lower bound of the risk function is to obtain the upper bound of $\Tr(U^{\dagger}OUV_{\C}^{\dagger}OV_{\C})$, which can be written as
        \begin{equation}\label{eq:UOUVOV_eigen_decomp}
            \Tr(U^{\dagger}OUV_{\C}^{\dagger}OV_{\C})=\sum_{p=1}^{d}\sum_{q=1}^{d}\lambda_p \lambda_q\Tr(U^{\dagger}\ket{\bm{o}_p}\bra{\bm{o}_p}UV_{\C}^{\dagger}\ket{\bm{o}_q}\bra{\bm{o}_q}V_{\C}).
        \end{equation}
        In the following, we elaborate on how to obtain the upper bound of the term $\Tr(U^{\dagger}\ket{\bm{o}_p}\bra{\bm{o}_p}UV_{\C}^{\dagger}\ket{\bm{o}_q}\bra{\bm{o}_q}V_{\C})$ with the assumption of perfect training in Eqn.~(\ref{eq:app_pt_inco}) by separately considering the case of the input states being linearly independent and linearly dependent. 
        
		\smallskip

        \noindent \underline{\textit{The input states are linearly independent.}} To analyze the relation between the training dataset $\{\ket{\bm{\psi}_j} \}_{j=1}^N$ and the average risk $\mathbb{E}_{U,\mathcal{D}_{\C}} \Tr(U^{\dagger}\ket{\bm{o}_p}\bra{\bm{o}_p}UV_{\C}^{\dagger}\ket{\bm{o}_q}\bra{\bm{o}_q}V_{\C})$, we first consider how to express $U^{\dagger}\ket{\bm{o}_p}$ in linearly independent states $\{\ket{\bm{\psi}_j} \}_{j=1}^N$ with $N\le d$. Applying QR decomposition to the matrix $\bm{\Psi}_{1:N}=[\ket{\bm{\psi}_1}, \cdots, \ket{\bm{\psi}_N}]$ yields the corresponding orthogonal matrix $\widetilde{\bm{\Psi}}_{1:N}=[\ket{\widetilde{\bm{\psi}}_1}, \cdots, \ket{\widetilde{\bm{\psi}}_N}]$ such that $\ket{\widetilde{\bm{\psi}}_1}=\ket{\bm{\psi}_1}$ and $\bm{\Psi}_{1:N}=\widetilde{\bm{\Psi}}_{1:N}R^{(N)}$ with $R^{(N)}$ being the $N$-dimensional upper triangular matrix and the states set $\{\ket{\bm{\psi}_j} \}_{j=1}^N$ and $\{\ket{\widetilde{\bm{\psi}}_j} \}_{j=1}^N$ are equivalent. 
        Spanning the orthogonal training input states $\{\ket{\widetilde{\bm{\psi}}_j} \}_{j=1}^N$ to the complete orthogonal basis in the Hilbert space $\mathcal{H}_d$ yields $\widetilde{\bm{\Psi}}=[\ket{\widetilde{\bm{\psi}}_1}, \cdots, \ket{\widetilde{\bm{\psi}}_N}, \ket{\widetilde{\bm{\psi}}_{N+1}}, \cdots, \ket{\widetilde{\bm{\psi}}_{d}} ]$.
        Due to the equivalence between the basis $\{\ket{\bm{\psi}_j} \}_{j=1}^N$ and $\{\ket{\widetilde{\bm{\psi}}_j} \}_{j=1}^N$, another set of complete basis yields 
        \begin{equation}\label{eq:orth_basis}
        \bm{\Psi}=[\ket{\bm{\psi}_1}, \cdots, \ket{\bm{\psi}_N}, \ket{\widetilde{\bm{\psi}}_{N+1}}, \cdots, \ket{\widetilde{\bm{\psi}}_{d}} ] \mbox{~with~} \bm{\Psi}=\widetilde{\bm{\Psi}}R=\widetilde{\bm{\Psi}}\left( {\begin{array}{cc}
				R^{(N)} & 0 \\
				0 & \mathbb{I}_{d-N} \\
		\end{array} } \right)
        \end{equation}
        where $\ket{\widetilde{\bm{\psi}}_{j}}$ are mutually orthogonal for any $j\in[d]\backslash[N]$ and are orthogonal with $\ket{\bm{\psi}_j}$ for any $j\in [N]$, $R$ and $R^{(N)}$ refer to the upper triangular matrices in the QR decomposition of $\bm{\Psi}$ and $\bm{\Psi}_{1:N}$. In this regard, expressing the states $U^{\dagger}\ket{\bm{o}_p}$ in the complete basis $\bm{\Psi}$ yields
        \begin{equation}\label{eq:Uketo}
            U^{\dagger}\ket{\bm{o}_p} = \sum_{j=1}^N \bm{u}_j^{(p)} \ket{\bm{\psi}_j} + \sum_{j=N+1}^d \widetilde{\bm{u}}_j^{(p)} \ket{\widetilde{\bm{\psi}}_j} = \bm{\Psi} \cdot (\bm{u}^{(p)})^{\intercal},
        \end{equation}
        where $\bm{u}^{(p)}=(\bm{u}_1^{(p)},\cdots,\bm{u}_N^{(p)}, \widetilde{\bm{u}}_{N+1}^{(p)},\cdots, \widetilde{\bm{u}}_d^{(p)})$ refers to the related coefficients. Then the inner product between $U^{\dagger}\ket{\bm{o}_p}$ and $\ket{\widetilde{\bm{\psi}}_j}$ for $j\in [d] \backslash [N]$ yields
        \begin{align}
            \bra{\widetilde{\bm{\psi}}_j} U^{\dagger}\ket{\bm{o}_p} = \sum_{k=1}^N \bm{u}_k^{(p)} \braket{\widetilde{\bm{\psi}}_j|\bm{\psi}_k} + \sum_{k=N+1}^d \widetilde{\bm{u}}_k^{(p)} \braket{\widetilde{\bm{\psi}}_j|\widetilde{\bm{\psi}}_k}= \widetilde{\bm{u}}_j^{(p)}, \label{eq:tilde_u_j}
        \end{align}
        where the second equality follows the orthogonality between  $\ket{\widetilde{\bm{\psi}}_k}$ with the other states in $\bm{\Psi}$. Similar results also hold for $U^{\dagger}\ket{\bm{o}_q}$, $V_{\C}^{\dagger}\ket{\bm{o}_p}$, and $V_{\C}^{\dagger}\ket{\bm{o}_q}$ with defining 
        \begin{align}
            \widetilde{\bm{u}}_k^{(q)}:= \bra{\widetilde{\bm{\psi}}_k}U^{\dagger}\ket{\bm{o}_q}
            ~\mbox{and}~
            \widetilde{\bm{v}}_k^{(q)}:= \bra{\widetilde{\bm{\psi}}_k}V_{\C}^{\dagger}\ket{\bm{o}_q},~\mbox{for}~k\in [d]\backslash[N]. \label{eq:u_v_q}
        \end{align}
        The assumption of perfect training suggests that $|\bra{\bm{\psi}_j}U^{\dagger}\ket{\bm{o}_q}|^2=|\bra{\bm{\psi}_k}V_{\C}^{\dagger}\ket{\bm{o}_q}|^2$  or equivalently $\bra{\bm{\psi}_j}U^{\dagger}\ket{\bm{o}_q}=e^{i\bm{\gamma}_j^{(q)}}\bra{\bm{\psi}_k}V_{\C}^{\dagger}\ket{\bm{o}_q}$ for $j\in [N]$ with $\bm{\gamma}_j^{(q)}$ being the inter-state relative phase. Notably, the value $\bm{\gamma}_j^{(q)}$ could be an arbitrary real number and is assumed to be uniformly distributed over the value space, which is reasonable in the context of NFL theorems considering the average prediction error.
        
        We are now ready to analyze the average risk function. In conjunction with the above notations defined in Eqns.~(\ref{eq:Uketo})-~(\ref{eq:u_v_q}) and the risk function defined in Eqn.~(\ref{eq:app_risk}), we have
        \begin{align}
			& \mathbb{E}_{U,\mathcal{D}_{\C}} \Tr\left(U^{\dagger}\ket{\bm{o}_p}\bra{\bm{o}_p}UV_{\C}^{\dagger}\ket{\bm{o}_q}\bra{\bm{o}_q}V_{\C}\right)  
			\nonumber \\
            = & \mathbb{E}_{U,\mathcal{D}_{\C}} \Tr\left[\left(\sum_{j=1}^N \bm{u}_j^{(p)} \ket{\bm{\psi}_j} + \sum_{j=N+1}^d \widetilde{\bm{u}}_j^{(p)} \ket{\widetilde{\bm{\psi}}_j}\right)\left(\sum_{k=1}^N (\bm{u}_k^{(p)})^* \bra{\bm{\psi}_k} + \sum_{k=N+1}^d (\widetilde{\bm{u}}_k^{(p)})^* \bra{\widetilde{\bm{\psi}}_k}\right)V_{\C}^{\dagger}\ket{\bm{o}_q}\bra{\bm{o}_q}V_{\C}\right]
            \nonumber \\
            = & \mathbb{E}_{U,\mathcal{D}_{\C}} \Bigg[\underbrace{\Tr\left(\sum_{j=1}^N \sum_{k=1}^N \bm{u}_j^{(p)}(\bm{u}_k^{(p)})^* \ket{\bm{\psi}_j}\bra{\bm{\psi}_k}U^{\dagger}\ket{\bm{o}_q}\bra{\bm{o}_q}U e^{i(\bm{\gamma}_k^{(q)}-\bm{\gamma}_j^{(q)})} \right)}_{T_1} 
            \nonumber \\
            &  + \underbrace{\Tr\left(\sum_{j=1}^N \sum_{k=N+1}^d  \bm{u}_j^{(p)}(\widetilde{\bm{u}}_k^{(p)})^* \ket{\bm{\psi}_j}\bra{\widetilde{\bm{\psi}}_k}V_{\C}^{\dagger}\ket{\bm{o}_q}\bra{\bm{o}_q}U e^{i\bm{\gamma}_j^{(q)}} \right)}_{T_2}
            \nonumber \\
            &  + \underbrace{\Tr\left(\sum_{j=N+1}^d \sum_{k=1}^N  \widetilde{\bm{u}}_j^{(p)}(\bm{u}_k^{(p)})^* \ket{\widetilde{\bm{\psi}}_j}\bra{\bm{\psi}_k}U^{\dagger}\ket{\bm{o}_q}\bra{\bm{o}_q}V_{\C} e^{-i\bm{\gamma}_k^{(q)}} \right)}_{T_3} + \underbrace{\Tr\left(\sum_{j=N+1}^d \sum_{k=N+1}^d \widetilde{\bm{u}}_j^{(p)} (\widetilde{\bm{v}}_j^{(q)})^*(\widetilde{\bm{u}}_k^{(p)})^*\widetilde{\bm{v}}_k^{(q)} \right)}_{T_4} \Bigg],
            \label{eq:NFL_Ob_Linear_In_1}
		\end{align} 
        where the first equality employs Eqn.~(\ref{eq:Uketo}) which gives the linear representation of $U^{\dagger}\ket{\bm{o}_p}$ in the set of completely basis $\Psi$, the second equality follows the assumption of perfect training $\bra{\bm{\psi}_j}U^{\dagger}\ket{\bm{o}_q}=e^{i\bm{\gamma}_j^{(q)}}\bra{\bm{\psi}_k}V_{\C}^{\dagger}\ket{\bm{o}_q}$ for $j\in [N]$ and the representation of $\widetilde{\bm{v}}_{k}^{(q)}=\bra{\widetilde{\bm{\psi}}_k}V_{\C}^{\dagger}\ket{\bm{o}_q}$ and $\widetilde{\bm{u}}_{k}^{(p)}=\bra{\widetilde{\bm{\psi}}_k}U^{\dagger}\ket{\bm{o}_p}$ for $k\in [d]\backslash [N]$ given in Eqn.~(\ref{eq:u_v_q}). In the following, we will separately analyze the terms $T_1$, $T_2$, $T_3$, and $T_4$.

        \smallskip
        
        \textbf{Simplification of term $T_1$}. Denoting the diagonal matrix consisting of the $N$ inter-state relative phase $\bm{\gamma}_1^{(q)}, \cdots, \bm{\gamma}_N^{(q)}$ as $P_{\bm{\gamma}}^{(N)}$ and the $d$ dimensional diagonal matrix consisting of $(\bm{\gamma}_1^{(q)}, \cdots, \bm{\gamma}_N^{(q)}, 0, \cdots, 0)$ as $P_{\bm{\gamma},0}^{(N)}=\left( {\begin{array}{cc}
				P_{\bm{\gamma}}^{(N)} & 0 \\
				0  & 0  \\
		\end{array} } \right)$, we have 
        \begin{align}
            \mathbb{E}_{U,\mathcal{D}_{\C}}{T_1} \overset{1}{=} & \mathbb{E}_{U,\mathcal{D}_{\C},\bm{\gamma}} \Tr\left(\sum_{j=1}^N \sum_{k=1}^N \bm{u}_j^{(p)}(\bm{u}_k^{(p)})^* \ket{\bm{\psi}_j}\bra{\bm{\psi}_k}U^{\dagger}\ket{\bm{o}_q}\bra{\bm{o}_q}U e^{i(\bm{\gamma}_k^{(q)}-\bm{\gamma}_j^{(q)})} \right)
            \nonumber \\
    		\overset{2}{=} &  \mathbb{E}_{U,\mathcal{D}_{\C},\bm{\gamma}} \Tr\left( \bm{\Psi}  P_{\bm{\gamma},0}^{(N)}(\bm{u}^{(p)})^{\intercal} \bar{\bm{u}}^{(p)}(P_{\bm{\gamma},0}^{(N)})^{\dagger}\bm{\Psi}^{\dagger} U^{\dagger}\ket{\bm{o}_q}\bra{\bm{o}_q}U  \right)
            \nonumber \\
    		\overset{3}{=} &  \mathbb{E}_{U,\mathcal{D}_{\C},\bm{\gamma}} \Tr\left( \bm{\Psi}  P_{\bm{\gamma},0}^{(N)} \bm{\Psi}^{-1} \bm{\Psi}(\bm{u}^{(p)})^{\intercal} \bar{\bm{u}}^{(p)}\bm{\Psi}^{\dagger}(\bm{\Psi}^{\dagger})^{-1}(P_{\bm{\gamma},0}^{(N)})^{\dagger}\bm{\Psi}^{\dagger} U^{\dagger}\ket{\bm{o}_q}\bra{\bm{o}_q}U  \right)
            \nonumber \\
    		\overset{4}{=} &  \mathbb{E}_{U,\mathcal{D}_{\C},\bm{\gamma}} \Tr\left( \bm{\Psi}  P_{\bm{\gamma},0}^{(N)} \bm{\Psi}^{-1} U^{\dagger}\ket{\bm{o}_p}\bra{\bm{o}_p}U (\bm{\Psi}^{\dagger})^{-1}(P_{\bm{\gamma},0}^{(N)})^{\dagger}\bm{\Psi}^{\dagger} U^{\dagger}\ket{\bm{o}_q}\bra{\bm{o}_q}U  \right)
            \nonumber \\
    		\overset{5}{=} &  \mathbb{E}_{\mathcal{D}_{\C},\bm{\gamma}} \frac{1}{d(d^2-1)}\left[ (d-\delta_{pq})\Tr\left(\bm{\Psi}  P_{\bm{\gamma},0}^{(N)} \bm{\Psi}^{-1} (\bm{\Psi}^{\dagger})^{-1}(P_{\bm{\gamma},0}^{(N)})^{\dagger}\bm{\Psi}^{\dagger} \right) + (d\delta_{pq}-1) \Tr\left(P_{\bm{\gamma},0}^{(N)} \right)\Tr\left((P_{\bm{\gamma},0}^{(N)})^{\dagger}\right)\right]
            \nonumber \\
    		\overset{6}{=} &  \mathbb{E}_{\mathcal{D}_{\C},\bm{\gamma}} \frac{1}{d(d^2-1)}\left[ (d-\delta_{pq})\Tr\left(\widetilde{\bm{\Psi}}R  P_{\bm{\gamma},0}^{(N)} R^{-1}\widetilde{\bm{\Psi}}^{-1} \widetilde{\bm{\Psi}}(R^{\dagger})^{-1}(P_{\bm{\gamma},0}^{(N)})^{\dagger}R^{\dagger}\widetilde{\bm{\Psi}}^{\dagger} \right) + (d\delta_{pq}-1) \Tr\left(P_{\bm{\gamma}}^{(N)} \right)\Tr\left((P_{\bm{\gamma}}^{(N)})^{\dagger}\right)\right]
            \nonumber \\
    		\overset{7}{=} &  \mathbb{E}_{\mathcal{D}_{\C},\bm{\gamma}} \frac{1}{d(d^2-1)}\left[ (d-\delta_{pq})\Tr\left(R  P_{\bm{\gamma},0}^{(N)} R^{-1}(R^{\dagger})^{-1}(P_{\bm{\gamma},0}^{(N)})^{\dagger}R^{\dagger} \right) + (d\delta_{pq}-1) \sum_{j=1}^N \sum_{k=1}^N e^{i(\bm{\gamma}_k^{(q)}-\bm{\gamma}_j^{(q)})}\right]
            \nonumber \\
    		\overset{8}{=} &  \mathbb{E}_{\mathcal{D}_{\C},\bm{\gamma}} \frac{1}{d(d^2-1)}\left[ (d-\delta_{pq})\Tr\left(R^{(N)}  P_{\bm{\gamma}}^{(N)} (R^{(N)})^{-1}((R^{(N)})^{\dagger})^{-1}(P_{\bm{\gamma}}^{(N)})^{\dagger}(R^{(N)})^{\dagger} \right) + (d\delta_{pq}-1) N\right], 
            \label{eq:Ob_T1-0}
        \end{align}
        where the second equality uses the matrix representation of $\bm{\Psi}  P_{\bm{\gamma},0}^{(N)}(\bm{u}^{(p)})^{\intercal} = \sum_{j=1}^N e^{i\bm{\gamma}_j^{(q)}} \bm{u}_j^{(p)} \ket{\bm{\psi}_j} $, the third and fourth equality employs the term $\bm{\Psi}^{-1}\bm{\Psi}=\mathbb{I}_d$ and the representation $\bm{\Psi}  (\bm{u}^{(p)})^{\intercal}=U^{\dagger}\ket{\bm{o}_p}$ given in Eqn.~(\ref{eq:Uketo}), the fifth equality follows Property~\ref{prop:Tr(WWWW)} on the Haar unitary $U$, the sixth and seventh equality leverage the QR-decomposition of square matrix $\bm{\Psi}=\widetilde{\bm{\Psi}}R$ and the orthogonality  of $\widetilde{\bm{\Psi}}$ such that $\widetilde{\bm{\Psi}}^{\dagger}=\widetilde{\bm{\Psi}}^{-1}$, the eighth  equality uses the representation $R=\left( {\begin{array}{cc}
				R^{(N)} & 0 \\
				0  & \mathbb{I}_{d-N}  \\
		\end{array} } \right)$ and $P_{\bm{\gamma},0}^{(N)}=\left( {\begin{array}{cc}
				P_{\bm{\gamma}}^{(N)} & 0 \\
				0  & 0  \\
		\end{array} } \right)$, and follows that  $\mathbb{E}_{\bm{\gamma}_j,\bm{\gamma}_k}e^{i(\bm{\gamma}_k^{(q)}-\bm{\gamma}_j^{(q)})}=\mathbb{E}_{\bm{\gamma}_j,\bm{\gamma}_k}\cos(\bm{\gamma}_k^{(q)}-\bm{\gamma}_j^{(q)})-i\mathbb{E}_{\bm{\gamma}_j,\bm{\gamma}_k} \sin(\bm{\gamma}_k^{(q)}-\bm{\gamma}_j^{(q)})=0$ for any $\bm{\gamma}_k^{(q)}\ne \bm{\gamma}_j^{(q)}$, as $\bm{\gamma}_k^{(q)}$ and $\bm{\gamma}_j^{(q)}$ are independent random variables and the expectation of the periodic function cosine and sine over the uniform distribution of $\bm{\gamma}_k^{(q)}$ or $\bm{\gamma}_j^{(q)}$ within a period is zero.

        Note that the matrix $R^{(N)}$ in Eqn.~(\ref{eq:Ob_T1-0}) refers to the $N$-dimensional upper triangular matrix satisfying $\bm{\Psi}_{1:N}=\widetilde{\bm{\Psi}}_{1:N}R^{(N)}$ defined in Eqn.~(\ref{eq:orth_basis}) with $\bm{\Psi}_{1:N},\widetilde{\bm{\Psi}}$ referring to the matrices of size  $d\times N$ and $\widetilde{\bm{\Psi}}_{1:N}^{\dagger}\widetilde{\bm{\Psi}}_{1:N}=\mathbb{I}_N$. Hence, denoting $\Gamma:=(R^{(N)})^{\dagger}R^{(N)}=(R^{(N)})^{\dagger}\widetilde{\bm{\Psi}}_{1:N}^{\dagger}\widetilde{\bm{\Psi}}_{1:N}R^{(N)}=\bm{\Psi}_{1:N}^{\dagger}\bm{\Psi}_{1:N}$, we have
        \begin{align}
            \mathbb{E}_{U,\mathcal{D}_{\C}}{T_1} 
    		= &  \mathbb{E}_{\mathcal{D}_{\C},\bm{\gamma}} \frac{1}{d(d^2-1)}\left[ (d-\delta_{pq})\Tr\left(  P_{\bm{\gamma}}^{(N)} \Gamma^{-1}(P_{\bm{\gamma}}^{(N)})^{\dagger}\Gamma \right) + (d\delta_{pq}-1) N\right]  
            \nonumber \\
    		= &  \mathbb{E}_{\mathcal{D}_{\C},\bm{\gamma}} \frac{1}{d(d^2-1)}\left[ (d-\delta_{pq})\sum_{j=1}^N\sum_{k=1}^N e^{i(\bm{\gamma}_j^{(q)}-\bm{\gamma}_k^{(q)})} \Gamma_{jk} (\Gamma^{-1})_{kj}+ (d\delta_{pq}-1) N\right]
            \nonumber \\
    		= &  \mathbb{E}_{\mathcal{D}_{\C}} \frac{1}{d(d^2-1)}\left[ (d-\delta_{pq})\sum_{j=1}^N\Gamma_{jj} (\Gamma^{-1})_{jj}+ (d\delta_{pq}-1) N\right]
            \nonumber \\
    		= &  \mathbb{E}_{\mathcal{D}_{\C}} \frac{1}{d(d^2-1)}\left[ (d-\delta_{pq})\Tr(\Gamma^{-1})+ (d\delta_{pq}-1) N\right],
            \label{eq:Ob_T1-01}
        \end{align}
        where the second equality follows direct algebraic calculation, the third equality follows the same derivation in Eqn.~(\ref{eq:Ob_T1-0}) that $\mathbb{E}_{\bm{\gamma}_j,\bm{\gamma}_k}e^{i(\bm{\gamma}_j^{(q)}-\bm{\gamma}_k^{(q)})}=0$ for any $\bm{\gamma}_j^{(q)}\ne \bm{\gamma}_k^{(q)}$, the last equality follows that the $jk$-th entry of $\Gamma=\bm{\Psi}_{1:N}^{\dagger}\bm{\Psi}_{1:N}$ refers to $\Gamma_{jk}=\braket{\bm{\psi}_j|\bm{\psi}_k}$ and hence $\Gamma_{jj}=1$. Now the problem of simplifying the term $T_1$ reduces to derive the expectation of the trace of matrix inversion $\Gamma^{-1}$ with $\Gamma=\bm{\Psi}_{1:N}^{\dagger}\bm{\Psi}_{1:N}$.  When the input states $\ket{\bm{\psi}_1}, \cdots, \ket{\bm{\psi}_N}$ are orthogonal, the interested matrix $\Gamma$ is exactly the $N$-dimensional identity matrix $\mathbb{I}_N$. Hence, the expectation of term $T_1$ in Eqn.~(\ref{eq:Ob_T1-01}) yields 
        \begin{equation}
            \mathbb{E}_{U,\mathcal{D}_{\C}}{T_1} 
    		= \mathbb{E}_{\mathcal{D}_{\C}} \frac{1}{d(d^2-1)}\left[ (d-\delta_{pq})\Tr(\mathbb{I}_N)+ (d\delta_{pq}-1) N\right]
            = N \cdot \frac{1+\delta_{pq}}{d(d+1)}.
            \label{eq:Ob_T1-01_orth}
        \end{equation}        
        We now discuss the case of linearly independent but non-orthogonal input states. We first note that in the sense of expectation over independent identically distributed quantum states, the diagonal element of $\Gamma^{-1}$ is symmetric with $\mathbb{E}_{\mathcal{D}_C}(\Gamma^{-1})_{jj}=\mathbb{E}_{\mathcal{D}_C}(\Gamma^{-1})_{kk}$ for any $j,k\in [N]$. In this regard, we have $\mathbb{E}_{\mathcal{D}_C} \Tr(\Gamma^{-1})=N\mathbb{E}_{\mathcal{D}_C} (\Gamma^{-1})_{11}$. For the matrix representation of $\Gamma=(R^{(N)})^{\dagger}R^{(N)}$ with $R^{(N)}$ and $(R^{(N)})^{\dagger}$ being the upper and lower triangular matrix, the matrix inversion yields $\Gamma^{-1}=(R^{(N)})^{-1}((R^{(N)})^{\dagger})^{-1}$ where $(R^{(N)})^{-1}$ and $((R^{(N)})^{\dagger})^{-1}$ are still the upper and lower triangular matrix respectively. In this regard, we have $(\Gamma^{-1})_{11}=((R^{(N)})^{-1})_{11}^2$. Recalling a fundamental fact in linear algebra that for any given invertible matrix $A$, the representation of $ij$-entry in matrix inversion $A^{-1}$ refers to $(A^{-1})_{ij}=M_{ij}/\det (A)$ with $M_{ij}$ being the corresponding minor of the entry $A_{ij}$ and $\det (A)$ referring to the determinant of matrix $A$. Applying this representation to the upper triangular matrix $R^{(N)}$, we have
        \begin{equation}
            ((R^{(N)})^{-1})_{11}=\prod_{j=2}(R^{(N)})_{jj}/\prod_{j=1}(R^{(N)})_{jj}=1/(R^{(N)})_{11}=1,\label{eq:R_11}
        \end{equation}
        where the first equality follows that the minor $M_{11}$ for $((R^{(N)})^{-1})_{11}$ and the determinant $\det(R^{(N)})$ refer to $\prod_{j=2}(R^{(N)})_{jj}$ and $\prod_{j=1}(R^{(N)})_{jj}$ respectively, the last equality follows that $(R^{(N)})_{11}=\braket{\bm{\psi}_1|\widetilde{\bm{\psi}}_1}=1$ as $\ket{\widetilde{\bm{\psi}}_1}=\ket{\bm{\psi}_1}$ in the QR decomposition. Based on the above discussion, we have 
        \begin{equation}
            \mathbb{E}_{\mathcal{D}_C} \Tr(\Gamma^{-1})=N\mathbb{E}_{\mathcal{D}_C} (\Gamma^{-1})_{11} = N\mathbb{E}_{\mathcal{D}_C}((R^{(N)})^{-1})_{11}^2 = N,
            \label{eq:gamma_inverse}
        \end{equation}
        where the first equality employs the symmetry of diagonal elements in $\Gamma^{-1}$, and the third equality follows Eqn.~(\ref{eq:R_11}). In this regard, combining Eqn.~(\ref{eq:Ob_T1-01}) and Eqn.~(\ref{eq:gamma_inverse}) yields 
        \begin{align}
            \mathbb{E}_{U,\mathcal{D}_{\C}}{T_1} =
    		 \frac{1}{d(d^2-1)}\left[ (d-\delta_{pq})N+ (d\delta_{pq}-1) N\right]
            = N \cdot \frac{1+\delta_{pq}}{d(d+1)}.
            \label{eq:Ob_T1-01_non-orth}
        \end{align}
        This is the same as that derived in Eqn.~(\ref{eq:Ob_T1-01_orth}) for orthogonal states.

        \textbf{Simplification of term $T_2$ and term $T_3$}. Considering the integration of term $T_2$ with respect to the inter-state relative phase $\bm{\gamma}$, we have 
        \begin{align}
           \mathbb{E}_{U,\mathcal{D}_{\C},\bm{\gamma}} T_2 = & \mathbb{E}_{U,\mathcal{D}_{\C},\bm{\gamma}}\Tr\left(\sum_{j=1}^N \sum_{k=N+1}^d  \bm{u}_j^{(p)}(\widetilde{\bm{u}}_k^{(p)})^* \ket{\bm{\psi}_j}\bra{\widetilde{\bm{\psi}}_k}V_{\C}^{\dagger}\ket{\bm{o}_q}\bra{\bm{o}_q}U e^{i\bm{\gamma}_j^{(q)}} \right)
            \nonumber \\
            = & \mathbb{E}_{U,\mathcal{D}_{\C}}\Tr\left(\sum_{j=1}^N \sum_{k=N+1}^d  \bm{u}_j^{(p)}(\widetilde{\bm{u}}_k^{(p)})^* \ket{\bm{\psi}_j}\bra{\widetilde{\bm{\psi}}_k}V_{\C}^{\dagger}\ket{\bm{o}_q}\bra{\bm{o}_q}U \mathbb{E}_{\bm{\gamma}_j^{(q)}}(\cos(\bm{\gamma}_j^{(q)})-i\sin(\bm{\gamma}_j^{(q)})) \right)
            \nonumber \\
            = & 0,
            \label{eq:Ob_T2_bound}
        \end{align}
        where the last equality follows the same discussion as in Eqn.~(\ref{eq:Ob_T1-0}) that the expectation of periodic function cosine and sine within a period is zero, i.e., $\mathbb{E}_{\bm{\gamma}_j}\cos(\bm{\gamma}_j)=\mathbb{E}_{\bm{\gamma}_j}\sin(\bm{\gamma}_j)=0$. It is straight to check that similar results hold for the term $T_3$, i.e. $\mathbb{E}_{U,\mathcal{D}_{\C}} T_3=0$.

        \textbf{Simplification of term $T_4$}. The term $T_4$ is associated with the orthogonal states $\{\ket{\widetilde{\bm{\psi}}_{N+1}}, \cdots, \ket{\widetilde{\bm{\psi}}_{d}}\}$ spanned from the training input states. Denote $\widetilde{\bm{u}}^{(p)}=(\widetilde{\bm{u}}_1^{(p)}, \cdots, \widetilde{\bm{u}}_d^{(p)})$ with $\widetilde{\bm{u}}_j^{(p)}=\bra{\widetilde{\bm{\psi}}_j}U^{\dagger}\ket{\bm{o}_p}$, where $\widetilde{\bm{u}}^{(p)}$ forms a state vector as it is obtained by acting the unitary $U_{\widetilde{\Psi}}=[\ket{\widetilde{\bm{\psi}}_1}, \cdots, \ket{\widetilde{\bm{\psi}}_d}]$ on the state vector $U^{\dagger}\ket{\bm{o}_p}$, i.e., $\widetilde{\bm{u}}^{(p)}=U_{\widetilde{\Psi}}U^{\dagger}\ket{\bm{o}_p}$. The similar results hold for $\widetilde{\bm{v}}$ with $\widetilde{\bm{v}}^{(p)}=U_{\widetilde{\Psi}}V_{\C}\ket{\bm{o}_p}$. While the state vector $\widetilde{\bm{u}}^{(p)}$ follows the Haar distribution as $U$ is a Haar unitary operator, the distribution of vector $\widetilde{\bm{u}}_{1:N}^{(p)} = (\widetilde{\bm{u}}_1^{(p)}, \cdots, \widetilde{\bm{u}}_N^{(p)})^{\intercal}$ and $\widetilde{\bm{u}}_{N+1:d}^{(p)} = (\widetilde{\bm{u}}_{N+1}^{(p)}, \cdots, \widetilde{\bm{u}}_d^{(p)})^{\intercal}$ remains unknown, leading to a challenge in calculating the integration of $T_4$. To address this problem, we note that given the norm of the random vector $s_u^{(p)}:=\sum_{j=1}^N|\widetilde{\bm{u}}_j^{(p)}|^2$, the random vector $\widetilde{\bm{u}}_{N+1:d}^{(p)}/(1-s_u^{(p)})$ follows the Haar distribution in the Hilbert space $\mathcal{H}^{d-N}$. With this regard, employing the tower property of conditional expectation to the term $T_4$ with condition $A=\{U:\|\widetilde{\bm{u}}_{1:N}^{(p)}\|^2 = s_{u}^{(p)}, \|\widetilde{\bm{u}}_{1:N}^{(q)}\|^2 = s_{u}^{(q)}\}$ yields
		\begin{align}
            \mathbb{E}_{U,\mathcal{D}_{\C}}T_4
			= &  
			\mathbb{E}_{U} \Bigg[\mathbb{E}_{\widetilde{\bm{u}}_{N+1:d}, \widetilde{\bm{v}}_{N+1:d}} \Bigg(\sum_{j=N+1}^d \widetilde{\bm{u}}_j^{(p)} (\widetilde{\bm{v}}_j^{(q)})^*\sum_{k=N+1}^d (\widetilde{\bm{u}}_k^{(p)})^*\widetilde{\bm{v}}_k^{(q)} \Bigg|  A \Bigg) \Bigg]
			\nonumber \\
			= &  
			\mathbb{E}_{U} \Bigg[\mathbb{E}_{\widetilde{\bm{u}}_{N+1:d}^{(p)}, \widetilde{\bm{v}}_{N+1:d}^{(q)}} \Bigg( (\widetilde{\bm{v}}_{N+1:d}^{(q)})^{\dagger} \widetilde{\bm{u}}_{N+1:d}^{(p)} \cdot  (\widetilde{\bm{u}}_{N+1:d}^{(p)})^{\dagger} \widetilde{\bm{v}}_{N+1:d}^{(q)} \Bigg| A  \Bigg) \Bigg]
			\nonumber \\
			= & 
			\mathbb{E}_{U} \Bigg[\mathbb{E}_{\widetilde{\bm{u}}_{N+1:d}^{(p)}, \widetilde{\bm{v}}_{N+1:d}^{(q)}} \Bigg(\left(1-s_u^{(p)}\right)\left(1-s_u^{(q)}\right)
            \frac{ (\widetilde{\bm{v}}_{N+1:d}^{(q)})^{\dagger} \widetilde{\bm{u}}_{N+1:d}^{(p)} \cdot  (\widetilde{\bm{u}}_{N+1:d}^{(p)})^{\dagger} \widetilde{\bm{v}}_{N+1:d}^{(q)} }{ \|\widetilde{\bm{v}}_{N+1:d}^{(q)}\|^2 \|\widetilde{\bm{u}}_{N+1:d}^{(p)}\|^2} \Bigg| A \Bigg)\Bigg]
			\nonumber \\
			= & \mathbb{E}_{U} \left[ \left(1-s_{u}^{(p)}\right)\left(1-s_{u}^{(q)}\right)\mathbb{E}_{\widetilde{\bm{u}}_{N+1:d}^{(p)}}\mathbb{E}_{\widetilde{\bm{v}}_{N+1:d}^{(q)}} \Tr\left( \frac{\widetilde{\bm{v}}_{N+1:d}^{(q)} (\widetilde{\bm{v}}_{N+1:d}^{(q)})^{\dagger}}{\|\widetilde{\bm{v}}_{N+1:d}^{(q)}\|^2}  \cdot \frac{\widetilde{\bm{u}}_{N+1:d}^{(p)} (\widetilde{\bm{u}}_{N+1:d}^{(p)})^{\dagger}}{\|\widetilde{\bm{u}}_{N+1:d}^{(p)}\|^2} \right) \Bigg| A  \right]
            \nonumber \\
			= &  \mathbb{E}_{U} \left[ \frac{(1-s_{u}^{(p)})(1-s_{u}^{(q)})}{d-N}\mathbb{E}_{\widetilde{\bm{u}}_{N+1:d}^{(p)}} \Tr\left( \frac{\widetilde{\bm{u}}_{N+1:d}^{(p)} (\widetilde{\bm{u}}_{N+1:d}^{(p)})^{\dagger}}{\|\widetilde{\bm{u}}_{N+1:d}^{(p)}\|^2} \right) \Bigg| A  \right] \nonumber \\
			= &  \mathbb{E}_{U} \left[ \frac{(1-s_{u}^{(p)})(1-s_{u}^{(q)})}{d-N}  \Bigg| A  \right],
            \label{eq:Ob_T4-1}
		\end{align}
        where the first equality ultilizes the tower property of conditional expectation, i.e., $\mathbb{E}(X)=\mathbb{E}_Y\mathbb{E}_{X|Y}(X|Y)$, the second equality exploits the vector representation of $\widetilde{\bm{u}}_{N+1:d}^{(p)} = (\widetilde{\bm{u}}_{N+1}^{(p)}, \cdots, \widetilde{\bm{u}}_d^{(p)})^{\top}$ and $\widetilde{\bm{v}}_{N+1:d} = (\widetilde{\bm{v}}_{N+1}^{(q)}, \cdots, \widetilde{\bm{v}}_d^{(q)})^{\top}$, the third equality employs the conditional independence between $\widetilde{\bm{u}}_{N+1:d}^{(p)}$ and $\widetilde{\bm{v}}_{N+1:d}^{(q)}$ for given $\sum_{j=1}^N |\bm{u}_j^{(p)}|^2 = s_u^{(p)}$, the fifth equality and the sixth equality employs Property~\ref{prop:Tr(psipsi)} to calculate the integration over the Haar states $\widetilde{\bm{u}}_{N+1:d}/ \|\widetilde{\bm{u}}_{N+1:d}\|$ and $\widetilde{\bm{v}}_{N+1:d}/ \|\widetilde{\bm{v}}_{N+1:d}\|$ respectively.
        Now Eqn.~(\ref{eq:Ob_T4-1}) only evolves one random vector $s_{u}(p)$. Denoting $s_u^{(p)}=\sum_{j=1}^N |\widetilde{\bm{u}}_j^{(p)}|^2 = (\widetilde{\bm{u}}^{(p)})^{\dagger} E_N \widetilde{\bm{u}}^{(p)} $ with $E_N=\left( {\begin{array}{cc}
				\mathbb{I}_{N} & 0 \\
				0 & 0 \\
		\end{array} } \right)$,
        and $\mathbb{I}_N$ being the $N$-dimensional identity matrix, we have 
        \begin{align}
            \mathbb{E}_{U,\mathcal{D}_{\C}}T_4
			= &  \mathbb{E}_{U}
            \frac{1-\Tr\left( \widetilde{\bm{u}}^{(p)} (\widetilde{\bm{u}}^{(p)})^{\dagger} {E}_N \right)-\Tr\left( \widetilde{\bm{u}}^{(q)} (\widetilde{\bm{u}}^{(q)})^{\dagger} {E}_N \right) + \Tr\left(\widetilde{\bm{u}}^{(p)} (\widetilde{\bm{u}}^{(p)})^{\dagger}  {E}_N \right) \Tr \left( \widetilde{\bm{u}}^{(q)} (\widetilde{\bm{u}}^{(q)})^{\dagger}{E}_N\right) }{(d-N)} 
            \nonumber \\
			= &  \frac{1}{d-N}\left[1-\frac{2N}{d}+\frac{N^2+N\mathbbm{1}(p=q)}{d(d+1)} \right]
            \nonumber \\
			= &  \frac{d-N+1}{d(d+1)}
            - \frac{N}{d(d+1)(d-N)}\mathbbm{1}(p\ne q)
            \nonumber \\
			\le &  \frac{d-N+1}{d(d+1)}
            - \frac{N}{d(d^2-1)}\mathbbm{1}(p\ne q),
            \label{eq:Ob_T4_bound}
        \end{align}
        where the second equality employs Property~\ref{prop:Tr(WW)Tr(WW)} and Property~\ref{prop:Tr(psipsi)} with respect to the Haar state $\widetilde{\bm{u}}^{(p)}$, the first equality follows that $d-1\ge d-N$ for any $N\ge 1$.
        
        \smallskip
        
        \textbf{The lower bound of the average risk function:} With the simplified results of terms $T_1$ to $T_4$, we are now ready to derive the upper bound of $\mathbb{E}_{U}\mathbb{E}_{\mathcal{D}_{\C}}\Tr(U^{\dagger}OUV_{\C}^{\dagger}OV_{\C})$. In particular, in conjunction with Eqn.~(\ref{eq:NFL_Ob_Linear_In_1}), Eqn.~(\ref{eq:Ob_T1-01_non-orth}) (term $T_1$), Eqn.~(\ref{eq:Ob_T2_bound}) (term $T_2$), and Eqn.~(\ref{eq:Ob_T4_bound}) (term $T_4$), we have
        \begin{align}
            \mathbb{E}_{U}\mathbb{E}_{\mathcal{D}_{\C}}\Tr(U^{\dagger}OUV_{\C}^{\dagger}OV_{\C})=&\sum_{p=1}^{d}\sum_{q=1}^{d}\lambda_p \lambda_q\Tr(U^{\dagger}\ket{\bm{o}_p}\bra{\bm{o}_p}UV_{\C}^{\dagger}\ket{\bm{o}_q}\bra{\bm{o}_q}V_{\C})
            \nonumber \\
            \le &\sum_{p=1}^{d}\lambda_p^2 \left(\frac{1}{d}+\frac{N}{d(d+1)}\right) + \sum_{p\ne q}^d \lambda_p \lambda_q \left(\frac{1}{d}-\frac{N}{d(d^2-1)} \right)
            \nonumber \\
            = &\sum_{p=1}^{d}\lambda_p^2 \left(\frac{N}{d(d^2-1)}+\frac{N}{d(d+1)}\right) + \sum_{p=1}^d \sum_{q=1}^d \lambda_p \lambda_q \left(\frac{1}{d}-\frac{N}{d(d^2-1)} \right)
            \nonumber \\
            = & \frac{N\Tr(O^2)}{d^2-1} +  \frac{d^2-N-1}{d(d^2-1)} \Tr(O)^2,
            \label{eq:Ob_UOUVOV_bound}
        \end{align}
        where the last equality employs $\Tr(O^2)=\sum_{p=1}^d \lambda_p^2$ and $\Tr(O)^2=\sum_{p=1}^d\sum_{q=1}^d \lambda_p \lambda_q$.
		Finally, in conjunction with Eqn.~(\ref{eq:app_risk_simplify}) and Eqn.~(\ref{eq:Ob_UOUVOV_bound}), we have
		\begin{align}
			\mathbb{E}_U\mathbb{E}_{\mathcal{\mathcal{D}_{\C}}} R_{U}(h_{\mathcal{D}_{\C}}) = & 	\frac{2}{d(d+1)}\left[\Tr(O^2)-\mathbb{E}_U\mathbb{E}_{\mathcal{\mathcal{D}_{\C}}}\Tr(U^{\dagger}OUV_{\C}^{\dagger}OV_{\C}) \right]
			\nonumber \\
			\ge & \frac{2(d^2-N-1)(d\Tr(O^2)-\Tr(O)^2)}{d^2(d+1)(d^2-1)}.
            \label{eq:Ob_final_LD_bound}
		\end{align}
        
        \smallskip

        \noindent \underline{\textit{The input states are linearly dependent.}}
        We note that the cases of $\mathcal{D}_{\C}$ consisting of $N$ linearly dependent input states could be reduced to the case of $\widetilde{N}$ linearly independent states when $\widetilde{N}< d$. In this regard, the achieved bound in Eqn.~(\ref{eq:Ob_final_LD_bound}) also holds by replacing $N$ with $\widetilde{N}$. Hence, we only focus on the cases of $N>\widetilde{N}= d$ linearly dependent input states. In particular, there could exist multiple sets of linearly independent states $\{\ket{\bm{\psi}_{(1)}}, \cdots, \ket{\bm{\psi}_{(d)}}\} \subset \{\ket{\bm{\psi}_1}, \cdots, \ket{\bm{\psi}_N}\}$. We denote the corresponding set of indices as $\mathcal{I}=\{(1),\cdots,(d)\} \subset [N]$ and the phase-related matrix as $P_{\bm{\gamma}}^{\mathcal{I}}$ consisting of  $e^{i\bm{\gamma}_{(1)}^{(q)}}, \cdots, e^{i\bm{\gamma}_{(d)}^{(q)}}$. Hence $U^{\dagger}\ket{\bm{o}_p}$ in Eqn.~(\ref{eq:UOUVOV_eigen_decomp}) can be expressed
        \begin{equation}
            \label{eq:U_keto_d}
            U^{\dagger}\ket{\bm{o}_p}=\sum_{j\in \mathcal{I}}\bm{u}_j^{(p)} \ket{\bm{\psi}_j}=\bm{\Psi}_{\mathcal{I}} \cdot (\bm{u}^{(p)})^{\intercal},
        \end{equation}
        where $\bm{u}^{(p)}=(\bm{u}_{(1)}^{(p)}, \cdots, \bm{u}_{(d)}^{(p)})$ refers to the vector of coefficients and $\bm{\Psi}_{\mathcal{I}}=[\ket{\bm{\psi}_{(1)}}, \cdots, \ket{\bm{\psi}_{(d)}}]$.
        Moreover, following the same routine of simplifying the term $T_1$ in Eqn.~(\ref{eq:Ob_T1-0}) with employing the representation of $U^{\dagger}\ket{\bm{o}_p}$ in Eqn.~(\ref{eq:U_keto_d}), we have 
        \begin{align}
			& \mathbb{E}_{U,\mathcal{D}_{\C}} \Tr\left(U^{\dagger}\ket{\bm{o}_p}\bra{\bm{o}_p}UV_{\C}^{\dagger}\ket{\bm{o}_q}\bra{\bm{o}_q}V_{\C}\right)  
            \nonumber \\
            = & 
            \mathbb{E}_{\mathcal{D}_{\C},\bm{\gamma}} \frac{1}{d(d^2-1)}\left[ (d-\delta_{pq})\Tr\left(R  P_{\bm{\gamma}}^{\mathcal{I}} R^{-1}(R^{\dagger})^{-1}(P_{\bm{\gamma}}^{\mathcal{I}})^{\dagger}R^{\dagger} \right) + (d\delta_{pq}-1) \Tr((P_{\bm{\gamma}}^{\mathcal{I}})^{\dagger})\Tr(P_{\bm{\gamma}}^{\mathcal{I}})\right]
            \nonumber \\
            = & 
            \mathbb{E}_{\mathcal{D}_{\C},\bm{\gamma}} \frac{1}{d(d^2-1)}\left[ (d-\delta_{pq})\Tr(\Gamma_{\mathcal{I}}^{-1}) + (d\delta_{pq}-1) \Tr((P_{\bm{\gamma}}^{\mathcal{I}})^{\dagger})\Tr(P_{\bm{\gamma}}^{\mathcal{I}})\right]
            \nonumber \\
            = & 
            \mathbb{E}_{\mathcal{D}_{\C},\bm{\gamma}} \frac{1}{d(d^2-1)}\left[ (d-\delta_{pq})d + (d\delta_{pq}-1) \Tr((P_{\bm{\gamma}}^{\mathcal{I}})^{\dagger})\Tr(P_{\bm{\gamma}}^{\mathcal{I}}) \right]
            \label{eq:Ob_NFL_DP_exp_pq}
        \end{align}
        where $\delta_{pq}$ equals $1$ if $p=q$ otherwise $0$, the second equality follows the simplification in Eqn.~(\ref{eq:Ob_T1-01}), the third equality follows the derivation in Eqn.~(\ref{eq:gamma_inverse}) of $\mathbb{E}_{\mathcal{D}_C}\Tr(\Gamma_{\mathcal{I}}^{-1})=d$ with $\Gamma_{\mathcal{I}}=\bm{\Psi}_{\mathcal{I}}^{\dagger}\bm{\Psi}_{\mathcal{I}}$. In conjunction with Eqn.~(\ref{eq:UOUVOV_eigen_decomp}) and Eqn.~(\ref{eq:Ob_NFL_DP_exp_pq}), we have
        \begin{align}
            & \mathbb{E}_{U}\mathbb{E}_{\mathcal{D}_{\C}}\Tr(U^{\dagger}OUV_{\C}^{\dagger}OV_{\C})
            \nonumber \\ =&\mathbb{E}_{U,\mathcal{D}_{\C}} \sum_{p=1}^{d}\sum_{q=1}^{d}\lambda_p \lambda_q\Tr(U^{\dagger}\ket{\bm{o}_p}\bra{\bm{o}_p}UV_{\C}^{\dagger}\ket{\bm{o}_q}\bra{\bm{o}_q}V_{\C})
            \nonumber \\
            =&\mathbb{E}_{\bm{\gamma}} \left(\frac{d\Tr((P_{\bm{\gamma}}^{\mathcal{I}})^{\dagger})\Tr(P_{\bm{\gamma}}^{\mathcal{I}})-d}{d(d^2-1)}\cdot \left(\sum_{p=1}^{d}\lambda_p^2\right) +  \frac{d^2-\Tr((P_{\bm{\gamma}}^{\mathcal{I}})^{\dagger})\Tr(P_{\bm{\gamma}}^{\mathcal{I}})}{d(d^2-1)}\cdot \left(\sum_{p=1}^{d}\sum_{q=1}^{d}\lambda_p \lambda_q\right) \right)
            \nonumber \\
            =&\mathbb{E}_{\bm{\gamma}} \left(\frac{d\Tr(O^2)-\Tr(O)^2}{d(d^2-1)} \Tr((P_{\bm{\gamma}}^{\mathcal{I}})^{\dagger})\Tr(P_{\bm{\gamma}}^{\mathcal{I}})\right) +  \frac{d\Tr(O)^2-\Tr(O^2)}{d^2-1}
            \nonumber \\
            \le & \mathbb{E}_{\bm{\gamma}} \max_{\mathcal{I}}\left(\frac{d\Tr(O^2)-\Tr(O)^2}{d(d^2-1)} \Tr((P_{\bm{\gamma}}^{\mathcal{I}})^{\dagger})\Tr(P_{\bm{\gamma}}^{\mathcal{I}})\right) +  \frac{d\Tr(O)^2-\Tr(O^2)}{d^2-1}
            \nonumber \\
            = &  \left(\frac{d\Tr(O^2)-\Tr(O)^2}{d(d^2-1)} \mathbb{E}_{\bm{\gamma}} \max_{\mathcal{I}}\sum_{j\in \mathcal{I}} \sum_{k\in \mathcal{I}} e^{i(\bm{\gamma}_j^{(q)}-\bm{\gamma}_k^{(q)})}\right) +  \frac{d\Tr(O)^2-\Tr(O^2)}{d^2-1}
            \nonumber \\
            \le &  \left(\frac{d\Tr(O^2)-\Tr(O)^2}{d(d^2-1)} \min\left\{\mathbb{E}_{\bm{\gamma}}\sum_{j=1}^N \sum_{k=1}^N e^{i(\bm{\gamma}_j^{(q)}-\bm{\gamma}_k^{(q)})},d^2\right\} \right) +  \frac{d\Tr(O)^2-\Tr(O^2)}{d^2-1}
            \nonumber \\
            = & \frac{d(\min\{N,d^2\}-1)\Tr(O^2)+(d^2-\min\{N,d^2\})\Tr(O)^2}{d(d^2-1)},
        \end{align}
        where the second equality employs Eqn.~(\ref{eq:Ob_NFL_DP_exp_pq}), the third equality employs $\Tr(O^2)=\sum_{p=1}^d \lambda_p^2$ and $\Tr(O)^2=\sum_{p=1}^d\sum_{q=1}^d \lambda_p \lambda_q$, the first inequality follows that there exists various matrix $P_{\bm{\gamma}}^{\mathcal{I}}$ corresponding to different sets of linearly independent states $\{\ket{\bm{\psi}_j}\}_{j\in \mathcal{I}}$ and $d\Tr(O^2)-\Tr(O)^2\ge 0$ for any observable $O$, the second inequality follows that $\mathbb{E}_{\bm{\gamma}} \max_{\mathcal{I}}\sum_{j\in \mathcal{I}} \sum_{k\in \mathcal{I}} e^{i(\bm{\gamma}_j^{(q)}-\bm{\gamma}_k^{(q)})}\le \mathbb{E}_{\bm{\gamma}}\sum_{j=1}^N \sum_{k=1}^N e^{i(\bm{\gamma}_j^{(q)}-\bm{\gamma}_k^{(q)})}$ as $\mathbb{E}_{\bm{\gamma}_k^{(q)},\bm{\gamma}_j^{(q)}} e^{i(\bm{\gamma}_j^{(q)}-\bm{\gamma}_k^{(q)})}=\delta_{kj}\ge 0$ for any $k,j \in [N]$, and $\sum_{j\in \mathcal{I}} \sum_{k\in \mathcal{I}} e^{i(\bm{\gamma}_j^{(q)}-\bm{\gamma}_k^{(q)})} =\sum_{j\in \mathcal{I}} \sum_{k\in \mathcal{I}} \cos(\bm{\gamma}_j^{(q)}-\bm{\gamma}_k^{(q)})\le d^2$ as $\cos(\bm{\gamma}_j^{(q)}-\bm{\gamma}_k^{(q)})\le 1$ always  holds, the last equality follows that $\mathbb{E}_{\bm{\gamma}_j^{(q)}}\mathbb{E}_{\bm{\gamma}_k^{(q)}}e^{i(\bm{\gamma}_j^{(q)}-\bm{\gamma}_k^{(q)})}=0$ for any $k\ne j$.
        
        Finally, following the conventions of Eqn.~(\ref{eq:Ob_UOUVOV_bound}) and Eqn.~(\ref{eq:Ob_final_LD_bound}), we have
        \begin{align}
			\mathbb{E}_U\mathbb{E}_{\mathcal{\mathcal{D}_{\C}}} R_{U}(h_{\mathcal{D}_{\C}}) = & 	\frac{2}{d(d+1)}\left[\Tr(O^2)-\mathbb{E}_U\mathbb{E}_{\mathcal{\mathcal{D}_{\C}}}\Tr(U^{\dagger}OUV_{\C}^{\dagger}OV_{\C}) \right]
			\nonumber \\
			\ge & \frac{2(d^2-N)(d\Tr(O^2)-\Tr(O)^2)}{d^2(d+1)(d^2-1)} \cdot \mathbbm{1}(N\le d^2).
            \label{eq:Ob_N_Ge_d_LD_bound}
		\end{align}
  
		This completes the proof.
		
	\end{proof}

	\section{Proof of Theorem~2 (NFL theorem for restricted quantum learning protocols)}\label{app_sec:NFL_SQ}
	Recall that the training dataset for restricted quantum learning protocols refers to $\{\ket{\bm{\psi}_j}, U\ket{\bm{\psi}_j}\}_{j=1}^N$ and the assumption of perfect training is given by
	$|\bra{\bm{\psi}_j}U^{\dagger}V_{\RQ}\ket{\bm{\psi}_j} |=1$ (i.e., $U\ket{\bm{\psi}_j} = e^{i\bm{\alpha}_j}V_{\RQ}\ket{\bm{\psi}_j}$) where $V_{\RQ}$ is the optimized unitary, $\bm{\alpha}_j$ is the inter-state relative phase of the true response state $U\ket{\bm{\psi}_j}$ and the estimated response state $V_{\RQ}\ket{\bm{\psi}_j}$. The proof of Theorem~\ref{thm:NFL_U_maintext} employs the following lemma about the relation between $U$ and $V_{\RQ}$.
    \begin{lemma}[Adapted from the proof of Theorem~1 in Ref.~\cite{sharma2022reformulation}]
        \label{lem:U_dagger_V}
        Let $U$ be a Haar unitary, $\{\ket{\bm{\psi}_j}\}_{j=1}^N$ arbitrary independent states, and $V$ arbitrary unitary satisfying the assumption of perfect training $U\ket{\bm{\psi}_j} = e^{i\bm{\alpha}_j}V\ket{\bm{\psi}_j}$) with $\bm{\alpha}_j$ being the inter-state relative phase. Then the product of unitaries $U^{\dagger}V$ yields
        \begin{equation}\label{eq:two_stage_perf_tr_mat_pt2}
		W:=U^{\dagger}V=\left(\begin{array}{ccc|c}e^{i \bm{\alpha}_{1}} & \cdots & 0 & \\ \vdots & \ddots & & 0 \\ 0 & & e^{i \bm{\alpha}_{N}} & \\ \hline & 0 & & \mathrm{Y}\end{array}\right)= \left( {\begin{array}{cc}  P_{\bm{\alpha}} & 0 \\ 0 & Y \\
			\end{array} } \right) ,
	   \end{equation}
        where $Y\in \mathbb{C}^{d-N}$ is assumed to be a Haar unitary operator, and $P_{\bm{\alpha}}$ refers to the $N$-dimensional diagonal matrix of inter-state relative phases $e^{i\bm{\alpha}_j}$.
    \end{lemma}
	
	There are two points to note in this lemma. First, the inter-state relative phase $\bm{\alpha}_j$ could be arbitrary under the perfect training assumption $|\bra{\bm{\psi}_j}U^{\dagger}V_{\RQ}\ket{\bm{\psi}_j} |=1$. In this regard, it is reasonable to assume that $\bm{\alpha}_j$ is uniformly distributed over any period of the function $e^{i\bm{\alpha}_j}$ in the context of NFL theorems considering the average cases of all related ingredients. Second, the randomness of $Y$ originates from the randomness of $U$ and $V_{\RQ}$. With this regard, the random unitary operator $Y$ can be rewritten as the product of two independent Haar unitary operators, namely $Y=Y_U Y_V$ where $Y_U$ and $Y_V$ are correlated with $U$ and $V_{\RQ}$, respectively. The independence of $Y_U$ and $Y_V$ is allowed because there is no correlation between the target unitary $U$ and the learned hypothesis $V_{\RQ}$ on the subspace complementary to the space spanned by input states in $\mathcal{D}_{\RQ}$. 
	We formulate these arguments as the following assumption.
	\begin{assu}\label{assum:W_decomp}
		With the definition of $W$ in Eqn.~(\ref{eq:two_stage_perf_tr_mat_pt2}), the operator $W$ is assumed to be decomposed into the product of two independent random operators whose randomness originates from $U$ and $V_{\RQ}$ respectively, i.e.,
		\begin{equation}\label{eq:two_stage_perf_tr_mat_pt2_dcp}
			W:=U^{\dagger}V=W_U W_V=\left(\begin{array}{ccc|c}e^{i \bm{\alpha}_{1}} & \cdots & 0 & \\ \vdots & \ddots & & 0 \\ 0 & & e^{i \bm{\alpha}_{N}} & \\ \hline & 0 & & \mathrm{Y_U}\end{array}\right)\left(\begin{array}{ccc|c}1 & \cdots & 0 & \\ \vdots & \ddots & & 0 \\ 0 & & 1 & \\ \hline & 0 & & \mathrm{Y_V}\end{array}\right)=\left( {\begin{array}{cc}  P_{\bm{\alpha}}^{(N)} & 0 \\ 0 & Y_U \\
			\end{array} } \right) \left( {\begin{array}{cc}
					\mathbb{I}_N & 0 \\ 0 & Y_V \\
			\end{array} } \right),
		\end{equation}
		where $W_U= \left( {\begin{array}{cc}
				P_{\bm{\alpha}}^{(N)} & 0 \\
				0  & Y_U  \\
		\end{array} } \right)$ and $W_V= \left( {\begin{array}{cc}
				\mathbb{I}_N & 0 \\
				0  & Y_V  \\
		\end{array} } \right)$ with $\mathbb{I}_N$ referring to the $N$-dimensional identity matrix and $P_{\bm{\alpha}}^{(N)}$ being the $N$-dimensional diagonal matrix of phases $e^{i\bm{\alpha}_j}$, $Y_U$ and $Y_V$ are two independent Haar unitary operators, $\bm{\alpha}_j$ are uniformly distributed over any period of $e^{i\bm{\alpha}_j}$. 
	\end{assu}
	With this assumption, the NFL theorem for the restricted quantum learning protocols is encapsulated in the following theorem.
	\begin{theorem-non}[Formal statement of Theorem~\ref{thm:NFL_U_maintext}]
		Let $f_U(\bm{\psi})=\Tr(U^{\dagger}OU\ket{\bm{\psi}}\bra{\bm{\psi}})$ be the target concept, the observable $O$ be any Hermitian matrix and the training dataset $\mathcal{D}_{\RQ}=\{\ket{\bm{\psi}_i}, \ket{\bm{\phi}_i} \}_{i=1}^N$ with $\ket{\bm{\phi}_i}=U\ket{\bm{\psi}_i}$. Let $h_{V_{\RQ}}$ be the learned hypothesis with the corresponding unitary operator $V_{\RQ}$ satisfying the assumption of perfect training $U\ket{\bm{\psi}_j} = e^{i\bm{\alpha}_j}V_{\RQ}\ket{\bm{\psi}_j}$) where $\bm{\alpha}_j$ is the relative phase. Under  Assumption~\ref{assum:W_decomp}, the averaged risk function over all unitaries $U$ and training sets $\mathcal{D}_{\RQ}$ yields
		\begin{equation}
			\mathbb{E}_U \mathbb{E}_{\mathcal{D}_{\RQ}} R_{U}(V_{\RQ})  \ge   
			\frac{2(d^2-N^2-1)(d\Tr(O^2)-\Tr(O)^2)}{d^4(d+1)}\cdot  \mathbbm{1}(\bm{\alpha}\varpropto \mathbf{1}_N) + \frac{2(d^2-N-1)(d\Tr(O^2)-\Tr(O)^2)}{d^4(d+1)}\cdot  \mathbbm{1}(\bm{\alpha}\not\varpropto \mathbf{1}_N),\nonumber
		\end{equation} 
        where the lower bound is tight in terms of training data size $N$ when the input states are linearly independent,  $d$ refers to the dimension of $n$-qubit quantum system with $d=2^n$, $\bm{\alpha}=(\bm{\alpha}_1, \cdots, \bm{\alpha}_N)$ and $\mathbf{1}_N$ refers to the all-ones vector of dimension $N$, $\mathbbm{1}(\bm{\alpha}\varpropto \mathbf{1}_N)$ means that there exists a constant $c$ such that the phase vector $\bm{\alpha}=c\mathbf{1}_N$, i.e., $\bm{\alpha}_k=\bm{\alpha}_j$ for any $j,k\in [N]$. 
	\end{theorem-non}
	
	\begin{proof}[Proof of Theorem~\ref{thm:NFL_U_maintext}]
		As suggested by Lemma~\ref{lem:upper_to_lower}, to derive the average risk, we first analyze the average value of the term $\mathbb{E}_{U,V_{\RQ}}\Tr\left(U^{\dagger}OUV_{\RQ}^{\dagger}OV_{\RQ}\right)$ in Eqn.~(\ref{eq:app_risk_simplify}) by separately considering the case of the input states being orthogonal, non-orthogonal but linearly independent, and linearly dependent.

        \smallskip
        
        \noindent \underline{\textit{The input states are orthogonal.}} We first recall that the task in the context of NFL considers learning a Haar random unitary $U$. Without loss of generality, we consider the orthogonal input states $\ket{\bm{\psi}_j}$ to be the computational basis $\ket{\bm{e}_j}$, as there always exists a unitary transformation $U_T$ such that $U_T\ket{\bm{\psi}_j}=\ket{\bm{e}_j}$ and the Haar measure is invariant under any unitary transformation. We then span the orthogonal training input states $\{\ket{\bm{e}_j} \}_{j=1}^N$ to a complete orthogonal basis $\{\ket{\bm{e}_1}, \cdots, \ket{\bm{e}_N}, \ket{\bm{e}_{N+1}}, \cdots, \ket{\bm{e}_d} \}$ of the Hilbert space $\mathcal{H}_d$ with $\sum_{j=1}^N\ket{\bm{e}_j}\bra{\bm{e}_j}+\sum_{j=N+1}^d\ket{\bm{e}_j}\bra{\bm{e}_j}=\mathbb{I}_d$.
        
        Under the Assumption~\ref{assum:W_decomp}, we have
		\begin{align}
			& 	\mathbb{E}_{U,V_{\RQ}}\Tr\left(U^{\dagger}OUV_{\RQ}^{\dagger}OV_{\RQ}\right)
			\nonumber \\
			= & \mathbb{E}_{U,V_{\RQ}}\Tr\bigg(U^{\dagger}OU  	\sum_{j=1}^d\ket{\bm{e}_j}\bra{\bm{e}_j} V_{\RQ}^{\dagger}OV_{\RQ} \sum_{k=1}^d\ket{\bm{e}_k}\bra{\bm{e}_k}\bigg) 
			\nonumber \\
			= & \mathbb{E}_{U,V_{\RQ}}\Bigg(\Tr\bigg[U^{\dagger}OU   	\sum_{j=1}^N\ket{\bm{e}_j}\bra{\bm{e}_j} V_{\RQ}^{\dagger}OV_{\RQ} \sum_{k=1}^N\ket{\bm{e}_k}\bra{\bm{e}_k}\bigg] + \Tr\bigg[U^{\dagger}OU   \sum_{j=1}^N\ket{\bm{e}_j}\bra{\bm{e}_j} V_{\RQ}^{\dagger}OV_{\RQ} \sum_{k=N+1}^d\ket{\bm{e}_k}\bra{\bm{e}_k}\bigg] 
			\nonumber \\
			& + \Tr\bigg[U^{\dagger}OU   \sum_{j=N+1}^d\ket{\bm{e}_j}\bra{\bm{e}_j} 	V_{\RQ}^{\dagger}OV_{\RQ} \sum_{k=1}^N\ket{\bm{e}_k}\bra{\bm{e}_k}\bigg]
			+ \Tr\bigg[U^{\dagger}OU   \sum_{j=N+1}^d\ket{\bm{e}_j}\bra{\bm{e}_j} 	V_{\RQ}^{\dagger}OV_{\RQ} \sum_{k=N+1}^d\ket{\bm{e}_k}\bra{\bm{e}_k}\bigg]
			\Bigg)  
			\nonumber \\
			= & 	\mathbb{E}_{\bm{\alpha}}\mathbb{E}_{U,V_{\Q}}\Bigg(\underbrace{\sum_{j=1}^N\sum_{k=1}^N\Tr\bigg[U^{\dagger}OU   \ket{\bm{e}_j}\bra{\bm{e}_j} U^{\dagger}OU \ket{\bm{e}_k}\bra{\bm{e}_k}\bigg] e^{i(\bm{\alpha}_k-\bm{\alpha}_j)} }_{T_1} 
			\nonumber \\
			& + \underbrace{ \sum_{j=1}^N\sum_{k=N+1}^d\Tr\bigg[\ket{\bm{e}_j}\bra{\bm{e}_k}U^{\dagger}OU  \bigg]  	\Tr\bigg[W_V^{\dagger}W_U^{\dagger}U^{\dagger}OU W_U W_V \ket{\bm{e}_k}\bra{\bm{e}_j}\bigg]}_{T_2}
			\Bigg) 
            \nonumber \\
			& + \underbrace{\sum_{j=N+1}^d\sum_{k=1}^N\Tr\bigg[\ket{\bm{e}_j}\bra{\bm{e}_k}U^{\dagger}OU  \bigg]  	\Tr\bigg[W_V^{\dagger}W_U^{\dagger}U^{\dagger}OU W_U W_V \ket{\bm{e}_k}\bra{\bm{e}_j}\bigg]}_{T_3}
			\Bigg) 
            \nonumber \\
			& + \underbrace{ \sum_{j=N+1}^d \sum_{k=N+1}^d \Tr\bigg[\ket{\bm{e}_j}\bra{\bm{e}_k}U^{\dagger}OU  \bigg]  	\Tr\bigg[W_V^{\dagger}W_U^{\dagger}U^{\dagger}OU W_U W_V \ket{\bm{e}_k}\bra{\bm{e}_j}\bigg]}_{T_4}
			\Bigg) 
			\label{eq:two_stage_training}
		\end{align}
		where the first equality follows $\mathbb{I}_d=\sum_{j=1}^d\ket{\bm{e}_j}\bra{\bm{e}_j}$, and the last three terms $T_2$, $T_3$, and $T_4$ in the third equality employs Eqn.~(\ref{eq:two_stage_perf_tr_mat_pt2_dcp}), i.e., $V_{\RQ}=UW_UW_V$. The expectation in terms of  $\bm{\alpha}$ is assumed to be taken over the uniform distribution of $(\bm{\alpha}_1, \cdots, \bm{\alpha}_N)$, which is faithful in the context of NFL theorems considering the average prediction error. In the following, we separately analyze the terms $T_1$, $T_2$, $T_3$, and $T_4$.

        \smallskip

        \textbf{Simplification of term $T_1$.}
		The term $T_1$ in Eqn.~(\ref{eq:two_stage_training}) yields 
		\begin{align}
			T_1 & =\mathbb{E}_{\bm{\alpha}}\mathbb{E}_{U} \sum_{j=1}^N\sum_{k=1}^N \Tr\left(U^{\dagger}OU   	\ket{\bm{e}_j}\bra{\bm{e}_k}\right) \Tr\left(U^{\dagger}OU \ket{\bm{e}_k}\bra{\bm{e}_j}\right)  e^{i(\bm{\alpha}_k-\bm{\alpha}_j)} 
			\nonumber \\
			& =\mathbb{E}_{\bm{\alpha}}\sum_{j=1}^N\sum_{k=1}^N 	\left(\frac{1}{d^2-1}\left(\Tr(O)^2\Tr(\ket{\bm{e}_k}\bra{\bm{e}_j})^2+\Tr(O^2)\right) - \frac{1}{d(d^2-1)} \left(\Tr(O^2)\Tr(\ket{\bm{e}_k}\bra{\bm{e}_j})^2+ \Tr(O)^2\right) \right)e^{i(\bm{\alpha}_k-\bm{\alpha}_j)}
			\nonumber \\
			& = \mathbb{E}_{\bm{\alpha}} \left(\frac{\Tr(O)^2+\Tr(O^2)}{d(d+1)} \sum_{j=1}^N e^{i(\bm{\alpha}_j-\bm{\alpha}_j)}+\frac{d\Tr(O^2)-\Tr(O)^2}{d(d^2-1)}\sum_{j=1}^N\sum_{k\ne j}^N  e^{i(\bm{\alpha}_k-\bm{\alpha}_j)} \right)
			\nonumber \\
            & = \frac{\Tr(O^2)+\Tr(O)^2}{d(d+1)} \cdot N + \frac{d\Tr(O^2)-\Tr(O)^2}{d(d^2-1)}\mathbb{E}_{\bm{\alpha}} \left(\sum_{j=1}^N\sum_{k=j+1}^N  2\cos(\bm{\alpha}_k -\bm{\alpha}_j)\right)
			\nonumber \\
            & = \frac{N(\Tr(O^2)+\Tr(O)^2)}{d(d+1)} + \frac{d\Tr(O^2)-\Tr(O)^2}{d(d^2-1)} \cdot (N^2-N) \cdot  \mathbbm{1} (\bm{\alpha}\varpropto \mathbf{1}_N) 
			\label{eq:T_1}
		\end{align}
		where the second equality employs Property~\ref{prop:Tr(WW)Tr(WW)}, the third equality is obtained by spliting the summation $\sum_{j=1}^{N}\sum_{k=1}^{N} (\cdot)$ into the case $\sum_{j=k}^{N} (\cdot)$ and $\sum_{j \ne k}^{N} (\cdot)$ with $\Tr(\ket{\bm{e}_k}\bra{\bm{e}_j})=0$ for $k\ne j$, the fourth equality exploits $e^{i(\bm{\alpha}_k-\bm{\alpha}_j)}+e^{i(\bm{\alpha}_j-\bm{\alpha}_k)}=2\cos(\bm{\alpha}_k-\bm{\alpha}_j)$. The last equality is obtained by separately analyzing the cases of whether or not the distribution of the input states enable the condition $\bm{\alpha}_k = \bm{\alpha}_j$ for any $k,j\in [N]$, which is described by the  indicator function $\mathbbm{1}(\bm{\alpha} \varpropto \mathbf{1}_N)$ with $\bm{\alpha}=(\bm{\alpha}_1, \cdots, \bm{\alpha}_N)$ and $\mathbf{1}_N$ referring to all-ones vector of dimension $N$.  The case of  $\bm{\alpha}\varpropto \mathbf{1}_N$ follows direct calculation and the case of  $\bm{\alpha}\not \varpropto \mathbf{1}_N$ exploits the fact that the expectation of the periodic even function cosine over a period is zero. 

        \smallskip
        
        \textbf{Simplification of terms $T_2$, $T_3$, and $T_4$}. As the terms $T_2$, $T_3$, and $T_4$ involve the summation of the entries with the same form but different indices, we hence first focus on the simplification of the entry
        \begin{equation}
            t_{jk}:=\mathbb{E}_{\bm{\alpha}}\mathbb{E}_{U,V_{\Q}}\Tr\bigg[\ket{\bm{e}_j}\bra{\bm{e}_k}U^{\dagger}OU  \bigg]  	\Tr\bigg[W_V^{\dagger}W_U^{\dagger}U^{\dagger}OU W_U W_V \ket{\bm{e}_k}\bra{\bm{e}_j}\bigg],
        \end{equation}
        where $T_2$, $T_3$, and $T_4$ satisfy $T_2=\sum_{j=1}^N\sum_{k=N+1}^d t_{jk}$, $T_3=\sum_{j=N+1}^d\sum_{k=1}^N t_{jk}$, and $T_4=\sum_{j=N+1}^d\sum_{k=N+1}^d t_{jk}$.
        
 Denote $A=U^{\dagger}OU=\left( {\begin{array}{cc}
				A_{1} & A_{2} \\
				A_{3} & A_{4} \\
		\end{array} } \right)$, $B=W_U^{\dagger}U^{\dagger}OU W_U=\left( {\begin{array}{cc}
				B_{1} & B_{2} \\
				B_{3} & B_{4} \\
		\end{array} } \right) = \left( {\begin{array}{cc}
				(P_{\bm{\alpha}}^{(N)})^{\dagger}A_1P_{\bm{\alpha}}^{(N)} & (P_{\bm{\alpha}}^{(N)})^{\dagger}A_2 Y_U \\
				Y_U^{\dagger} A_3 P_{\bm{\alpha}}^{(N)}  & Y_U^{\dagger} A_{4}Y_U  \\
		\end{array} } \right)$ with $W_U= \left( {\begin{array}{cc}
				P_{\bm{\alpha}}^{(N)} & 0 \\
				0  & Y_U  \\
		\end{array} } \right)$, and $\rho_{jk}=\ket{\bm{e}_j}\bra{\bm{e}_k}=\left( {\begin{array}{cc}
				\rho_{jk}^{(1)} & \rho_{jk}^{(2)} \\
				\rho_{jk}^{(3)} & \rho_{jk}^{(4)} \\
		\end{array} } \right)$. We have
		\begin{align}
			t_{jk} 
			= & \mathbb{E}_{U}\mathbb{E}_{W_V}\Tr\bigg[\ket{\bm{e}_j}\bra{\bm{e}_k}U^{\dagger}OU  \bigg]  \Tr\bigg[W_V^{\dagger}W_U^{\dagger}U^{\dagger}OU W_U W_V \ket{\bm{e}_k}\bra{\bm{e}_j}\bigg]
			\nonumber \\ 
			= & \mathbb{E}_{U} \Tr\left(\ket{\bm{e}_j}\bra{\bm{e}_k}U^{\dagger}OU \right) 	\mathbb{E}_{Y_V}\Tr\left(B_1\rho_{jk}^{(1)} +B_2Y_V\rho_{jk}^{(3)} + Y_V^{\dagger}B_3\rho_{jk}^{(2)}+Y_V^{\dagger}B_4Y_V \rho_{jk}^{(4)} \right)
			\nonumber \\ 
			= & \mathbb{E}_{U} \Tr\left(\ket{\bm{e}_j}\bra{\bm{e}_k}U^{\dagger}OU \right) 	\left(\Tr\left(	(P_{\bm{\alpha}}^{(N)})^{\dagger}A_1P_{\bm{\alpha}}^{(N)}\rho_{jk}^{(1)}\right) +\frac{\Tr(Y_U^{\dagger} A_{4}Y_U)\Tr(\rho_{jk}^{(4)})}{d-N} \right)
			\nonumber \\ 
			= & \mathbb{E}_{U} \Tr\left(\ket{\bm{e}_j}\bra{\bm{e}_k}U^{\dagger}OU \right) 	\left(\Tr\left(	A_1P_{\bm{\alpha}}^{(N)}\rho_{jk}^{(1)}(P_{\bm{\alpha}}^{(N)})^{\dagger} \right) +\frac{\Tr( A_{4})\Tr(\rho_{jk}^{(4)})}{d-N} \right), 
			\label{eq:T_2-1}
		\end{align}
		where $W_V= \left( {\begin{array}{cc}
				\mathbb{I}_N & 0 \\
				0  & Y_V  \\
		\end{array} } \right)$ is defined in Assumption~\ref{assum:W_decomp} with $\mathbb{I}_N$ being the $N$-dimensional identity matrix, the second equality follows algebraic operation with the block matrix representation of $B=W_U^{\dagger}U^{\dagger}OUW_U$ and $\rho_{jk}$, the third equality employs Property~\ref{prop:Tr(WW)} for the Haar integration on $Y_V$, the fourth equality follows the property of trace operation $\Tr(CD)=\Tr(DC)$. Note that $\rho_{jk}^{(1)}$ and $\rho_{jk}^{(4)}$ are zero matrix if $k\in[N], j \in [d] \backslash [N]$  or $j\in[N], k \in [d] \backslash [N]$ as $\rho_{jk}=\ket{\bm{e}_j}\bra{\bm{e}_k}$ refers to the matrix of zeros except the $jk$-th entry being $1$. In this regard, we have 
        \begin{equation}
            T_2=T_3=\sum_{j=1}^N \sum_{k=N+1}^{d} t_{jk}=\sum_{j=N+1}^d \sum_{k=1}^{N} t_{jk}=0.
            \label{eq:T_2-T_3}
        \end{equation}
        Moreover, for $j,k \in [d]\backslash [N]$, Eqn.~(\ref{eq:T_2-1}) can be rewritten in the form of $U^{\dagger}OU$ as following,
		\begin{align}
			T_{4} 
			= & \sum_{j,k=N+1}^d\mathbb{E}_{U} \Tr\left(\ket{\bm{e}_j}\bra{\bm{e}_k}U^{\dagger}OU \right) 	\left(\Tr\left(	A_1P_{\bm{\alpha}}^{(N)}\rho_{jk}^{(1)}(P_{\bm{\alpha}}^{(N)})^{\dagger} \right) +\frac{\Tr( A_{4})\Tr(\rho_{jk}^{(4)})}{d-N} \right) 
			\nonumber \\ 
			= & \sum_{j,k=N+1}^d\mathbb{E}_{U}\Tr\left(\ket{\bm{e}_j}\bra{\bm{e}_k}U^{\dagger}OU \right) 	\Tr\left( \left({\begin{array}{cc} A_1 & A_2 \\ A_3  & A_4  \\ \end{array} } \right) 	\left({\begin{array}{cc} 0 & 0 \\ 0  & \frac{1}{d-N} \mathbb{I}_{d-N}  \\ \end{array} } \right)\right)   
			\nonumber \\
            = & \sum_{j,k=N+1}^d\mathbb{E}_{U}\Tr\left(\ket{\bm{e}_j}\bra{\bm{e}_k}U^{\dagger}OU \right) 	\Tr\left( U^{\dagger}OU	\left({\begin{array}{cc} 0 & 0 \\ 0  & \frac{1}{d-N} \mathbb{I}_{d-N}  \\ \end{array} } \right)\right)   
			\nonumber \\
			= &  	\sum_{j,k=N+1}^d\frac{\Tr(\ket{\bm{e}_j}\bra{\bm{e}_k})(d\Tr(O)^2-\Tr(O^2))+\Tr(\ket{\bm{e}_j}\bra{\bm{e}_k}) (d-N)^{-1}(d\Tr(O^2)-\Tr(O)^2)}{d(d^2-1)} 
			\nonumber \\
			=& \frac{(d-N)(d\Tr(O)^2-\Tr(O^2)) + (d\Tr(O^2)-\Tr(O)^2)}{d(d^2-1)},
			\label{eq:T_4-2}
		\end{align}
        where the second equality follows that $\rho_{jk}^{(1)}$ is $N$-dimensional zero matrix and $\Tr(\rho_{jk}^{(4)})=1$ for $j,k\in[d]\backslash [N]$, the third equality utilizes the block matrix representation of $A=U^{\dagger}OU$, the fourth equality employs Property~\ref{prop:Tr(WW)Tr(WW)}.

        \smallskip
  
		\textbf{The lower bound of the average risk function}. In conjunction with Eqn.~(\ref{eq:two_stage_training}), Eqn.~(\ref{eq:T_1}) (term $T_1$), Eqn.~(\ref{eq:T_2-T_3}) (terms $T_2, T_3$), and Eqn.~(\ref{eq:T_4-2}) (term $T_4$) yields
		\begin{align}
			& \mathbb{E}_{U}\mathbb{E}_{\mathcal{\mathcal{D}_{\RQ}}} 	\Tr\left(U^{\dagger}OUV_{\RQ}^{\dagger}OV_{\RQ}\right) 
			\nonumber \\
			= & \frac{N(\Tr(O^2)+\Tr(O)^2)}{d(d+1)} + \frac{(N^2-N)(d\Tr(O^2)-\Tr(O)^2)}{d(d^2-1)}   \cdot  \mathbbm{1} (\bm{\alpha}\varpropto \mathbf{1}_N)  \nonumber \\
            & +  \frac{(d-N)(d\Tr(O)^2-\Tr(O^2)) + (d\Tr(O^2)-\Tr(O)^2)}{d(d^2-1)}
			\nonumber \\
			= & \frac{dN^2\Tr(O^2)+(d^2-N^2-1)\Tr(O)^2} {d(d^2-1)} \cdot  \mathbbm{1}(\bm{\alpha}\varpropto \mathbf{1}_N) + \frac{dN\Tr(O^2)+(d^2-N-1)\Tr(O)^2} {d(d^2-1)} \cdot  \mathbbm{1}(\bm{\alpha}\not\varpropto \mathbf{1}_N).
		\end{align} 
		Hence, utilizing Eqn.~(\ref{eq:app_risk_trans}), we have 
		\begin{align}
			& \mathbb{E}_U\mathbb{E}_{\mathcal{\mathcal{D}_{\RQ}}} R_{f_U}(V_{\RQ}) 
            \nonumber \\
            = & \frac{2}{d(d+1)}\left[\Tr(O^2)-\mathbb{E}_U\mathbb{E}_{\mathcal{\mathcal{D}_{\RQ}}}\Tr(U^{\dagger}OUV_{\RQ}^{\dagger}OV_{\RQ}) \right]
			\nonumber \\
			= & \frac{2(d^2-N^2-1)(d\Tr(O^2)-\Tr(O)^2)}{d^4(d+1)}\cdot  \mathbbm{1}(\bm{\alpha}\varpropto \mathbf{1}_N) + \frac{2(d^2-N-1)(d\Tr(O^2)-\Tr(O)^2)}{d^4(d+1)}\cdot  \mathbbm{1}(\bm{\alpha}\not\varpropto \mathbf{1}_N).
            \label{eq:RQ_bound_Orth_indicator}
		\end{align}
        where the first equality follows Lemma~\ref{lem:upper_to_lower} directly. Notably, the term $d\Tr(O^2)-\Tr(O)^2)$ is positive for all observables $O$. Specifically, denoting $\lambda_1, \cdots, \lambda_d$ as the eigenvalues of the observable $O$, we have $\Tr(O)^2=\sum_j \lambda_j^2 + \sum_{i,j} 2\lambda_i\lambda_j \le d\sum_j \lambda_j^2 = d\Tr(O^2)$ with the inequality following Cauchy-Schwartz inequality.

        \smallskip
        
        \textbf{The condition enabling $\bm{\alpha} \varpropto \mathbf{1}_N$}.
        We now analyze when the condition $\bm{\alpha}_k=\bm{\alpha}_j$ holds. Specifically, for any training data $(\ket{\bm{\psi}_i}, U\ket{\bm{\psi}_i})$ and $(\ket{\bm{\psi}_j}, V_{\RQ}\ket{\bm{\psi}_j})$ with $i\ne j$, we have
		\begin{equation}\label{eq:phase_equality}
			\braket{\bm{\psi}_k | \bm{\psi}_j} = 	\bra{\bm{\psi}_k}U^{\dagger}U\ket{\bm{\psi}_j} = e^{i(\bm{\alpha}_j-\bm{\alpha}_k)}  \bra{\bm{\psi}_k}V_{\RQ}^{\dagger} V_{\RQ}\ket{\bm{\psi}_j} = e^{i(\bm{\alpha}_j-\bm{\alpha}_k)} \braket{\bm{\psi}_k | \bm{\psi}_j}
		\end{equation}
		where the second equality exploits the assumpution of perfect training $U\ket{\bm{\psi}_j}= e^{-i\bm{\alpha}_j}V_{\RQ} \ket{\bm{\psi}_j}$. Eqn.~(\ref{eq:phase_equality}) implies that $\bm{\alpha}_{k}=\bm{\alpha}_j$ for all $k,j\in [N]$ if $\braket{\bm{\psi}_k|\bm{\psi}_j}\ne 0$ or equivalently the input states are not orthogonal. In this regard, the lower bound of the average risk function for orthogonal states yields
        \begin{align}
			\mathbb{E}_U\mathbb{E}_{\mathcal{\mathcal{D}_{\RQ}}} R_{f_U}(V_{\RQ}) 
			=   \frac{2(d^2-N-1)(d\Tr(O^2)-\Tr(O)^2)}{d^4(d+1)}.
            \label{eq:RQ_bound_Orth}
		\end{align}

        \noindent \underline{\textit{The input states are non-orthogonal but linearly independent.}} We analyze the cases of non-orthogonal yet linearly independent states by reducing to the case of orthogonal states. In particular, employing the QR 
        decomposition to the matrix $U_{\Psi}=[\ket{\bm{\psi}_1}, \cdots, \ket{\bm{\psi}_N}]$ yields $U_{\Psi}=QR$ with $Q=[\ket{\bm{Q}_1}, \cdots, \ket{\bm{Q}_N}]$ being a $d\times N$ orthogonal matrix and $R$ being an $N$-dimensional upper triangular matrix. As the inverse of $R$ is still an upper triangular matrix, we have $Q=U_{\Psi}R^{-1}$, or equivalently ,
        \begin{equation}
            \ket{\bm{Q}_j} = \sum_{k=1}^j R_{jk}\ket{\bm{\psi}_k}~\mbox{for any}~j\in [N],
        \end{equation}
        where $R_{jk}$ refers to the $jk$-entry of the upper triangular matrix $R^{-1}$.
        Then the assumption of perfect training for the orthogonal states refers to
        \begin{equation}
            U\ket{\bm{Q}_j} =\sum_{k=1}^j R_{jk}U\ket{\bm{\psi}_k} = \sum_{k=1}^j R_{jk}V_{\C}\ket{\bm{\psi}_k} e^{-i\bm{\alpha}_j} = e^{-i\bar{\bm{\alpha}}} \sum_{k=1}^j R_{jk}V_{\C}\ket{\bm{\psi}_k} = e^{-i\bar{\bm{\alpha}}} V_{\RQ}\ket{\bm{Q}_j},
        \end{equation}
        where $\bar{\bm{\alpha}}=\sum_{j=1}^N \bm{\alpha}_j/N$, the third equality follows that when the input states are linearly independent, the inter-state relative phase $\bm{\alpha}_j=\bm{\alpha}_k$ for any $j,k\in [N]$ as demonstrated in Eqn.~(\ref{eq:phase_equality}). This indicates that the assumption of perfect training also holds for the orthogonal states $\{Q_j\}_{j=1}^N$ with satisfying the condition of phase alignment $\bm{\alpha} \varpropto \mathbf{1}_N$. In this regard, employing the achieved results of Eqn.~(\ref{eq:RQ_bound_Orth_indicator}) with the condition $\bm{\alpha} \varpropto \mathbf{1}_N$ yields
        \begin{align}
			\mathbb{E}_U\mathbb{E}_{\mathcal{\mathcal{D}_{\RQ}}} R_{f_U}(V_{\RQ}) 
			=   \frac{2(d^2-N^2-1)(d\Tr(O^2)-\Tr(O)^2)}{d^4(d+1)}.
            \label{eq:RQ_bound_Non_Orth}.
		\end{align}
        
        \smallskip

        \noindent \underline{\textit{The input states are linearly dependent.}}
        We note that the cases of $\mathcal{D}_{\RQ}$ consisting of linearly dependent input states could be reduced to the case of $\widetilde{N}$ linearly independent states with $\widetilde{N}< d$. In this regard, the achieved bound in Eqn.~(\ref{eq:RQ_bound_Non_Orth}) also holds by replacing $N$ with $\widetilde{N}$. We now focus on the case of $N>\widetilde{N}= d$ linearly dependent input states. In particular, there could exist multiple sets of linearly independent states $\{\ket{\bm{\psi}_{(1)}}, \cdots, \ket{\bm{\psi}_{(d)}}\} \subset \{\ket{\bm{\psi}_1}, \cdots, \ket{\bm{\psi}_N}\}$. We denote the corresponding set of indices as $\mathcal{I}=\{(1),\cdots,(d)\} \subset [N]$ and the phase-related matrix as $P_{\bm{\alpha}}^{\mathcal{I}}$ consisting of  $e^{i\bm{\alpha}_{(1)}}, \cdots, e^{i\bm{\alpha}_{(d)}}$, leading to $V_{\RQ}=U P_{\bm{\alpha}}^{\mathcal{I}}$ as shown in Lemma~\ref{lem:U_dagger_V}. Hence, we have
        \begin{align}
			& 	\mathbb{E}_{U,\mathcal{D}_{\RQ}}\Tr\left(U^{\dagger}OUV_{\RQ}^{\dagger}OV_{\RQ}\right)
            \nonumber \\
			= & \mathbb{E}_{U, \mathcal{D}_{\RQ}} \Tr\left(U^{\dagger}OU(P_{\bm{\alpha}}^{\mathcal{I}})^{\dagger}U^{\dagger}OU P_{\bm{\alpha}}^{\mathcal{I}}\right)
			\nonumber \\
			= & \mathbb{E}_{\mathcal{D}_{\RQ}}\left(\frac{1}{d^2-1}\left(\Tr(O)^2\Tr(\mathbb{I}_d) + \Tr(O^2)\Tr((P_{\bm{\alpha}}^{\mathcal{I}})^{\dagger})\Tr(P_{\bm{\alpha}}^{\mathcal{I}}) \right) - \frac{1}{d(d^2-1)}\left(\Tr(O^2)\Tr(\mathbb{I}_d) + \Tr(O)^2\Tr((P_{\bm{\alpha}}^{\mathcal{I}})^{\dagger})\Tr(P_{\bm{\alpha}}^{\mathcal{I}})\right) \right)
			\nonumber \\
            \le  & \mathbb{E}_{\bm{\alpha}} \max_{\mathcal{I}}\frac{1}{d(d^2-1)}\left(d^2\Tr(O)^2-d\Tr(O^2) + (d\Tr(O^2)-\Tr(O)^2)\Tr((P_{\bm{\alpha}}^{\mathcal{I}})^{\dagger})\Tr(P_{\bm{\alpha}}^{\mathcal{I}})\right) 
            \nonumber \\
             = &  \frac{1}{d(d^2-1)}\left(d^2\Tr(O)^2-d\Tr(O^2) + (d\Tr(O^2)-\Tr(O)^2) \cdot \mathbb{E}_{\bm{\alpha}} \max_{\mathcal{I}}\sum_{k\in \mathcal{I}}\sum_{j\in \mathcal{I}} e^{i(\bm{\alpha}_{k}-\bm{\alpha}_{j})}\right) 
            \nonumber \\
            \le & \frac{1}{d(d^2-1)}\left(d^2\Tr(O)^2-d\Tr(O^2) + (d\Tr(O^2)-\Tr(O)^2) \cdot \min \left\{\mathbb{E}_{\bm{\alpha}} \sum_{k=1}^N\sum_{j=1}^N e^{i(\bm{\alpha}_{k}-\bm{\alpha}_{j})}, d^2 \right\}\right) 
            \nonumber \\
            = &  \frac{(d^2-\min\{N,d^2\})\Tr(O)^2+\min\{N,d^2\}\cdot (d-1)\Tr(O^2) }{d(d^2-1)}, 
            \label{eq:RQ_N_Ge_d-1}
        \end{align}
        where the second equality employs Property~\ref{prop:Tr(WWWW)}, the first inequality follows that there exists various matrix $P_{\bm{\alpha}}^{\mathcal{I}}$ corresponding to different sets of linearly independent states $\{\ket{\bm{\psi}_j}\}_{j\in\mathcal{I}}$ and $d\Tr(O^2)-\Tr(O)^2\ge 0$ for any observable $O$, the second inequality follows that $\mathbb{E}_{\bm{\alpha}} \max_{\mathcal{I}}\sum_{k\in \mathcal{I}}\sum_{j\in \mathcal{I}} e^{i(\bm{\alpha}_{k}-\bm{\alpha}_{j})}\le \mathbb{E}_{\bm{\alpha}} \sum_{k=1}^N\sum_{j=1}^N e^{i(\bm{\alpha}_{k}-\bm{\alpha}_{j})}$ as $\mathbb{E}_{\bm{\alpha}}e^{i(\bm{\alpha}_{k}-\bm{\alpha}_{j})}=\delta_{jk}\ge 0$ for any $j,k\in [N]$, and $\sum_{j\in \mathcal{I}} \sum_{k\in \mathcal{I}} e^{i(\bm{\alpha}_j-\bm{\alpha}_k)} =\sum_{j\in \mathcal{I}} \sum_{k\in \mathcal{I}} \cos(\bm{\alpha}_j-\bm{\alpha}_k)\le d^2$ as $\cos(\bm{\alpha}_j-\bm{\alpha}_k)\le 1$ always  holds, the last equality follows that the inter-state relative phases $\bm{\alpha_1}, \cdots, \bm{\alpha}_N$ are independently distributed over each period of the function $e^{i\bm{\alpha}_k}$ as stated in Assumption~\ref{assum:W_decomp} and hence $\mathbb{E}_{\bm{\alpha}_k,\bm{\alpha}_j}e^{i(\bm{\alpha}_k-\bm{\alpha}_j)}=0$ for any $k\ne j$.

        In conjunction with Eqn.~(\ref{eq:app_risk_simplify}) and Eqn.~(\ref{eq:RQ_N_Ge_d-1}), we have
        \begin{align}
			\mathbb{E}_U\mathbb{E}_{\mathcal{\mathcal{D}_{\RQ}}} R_{U}(h_{\mathcal{D}_{\RQ}}) = & 	\frac{2}{d(d+1)}\left[\Tr(O^2)-\mathbb{E}_U\mathbb{E}_{\mathcal{\mathcal{D}_{\RQ}}}\Tr(U^{\dagger}OUV_{\RQ}^{\dagger}OV_{\RQ}) \right]
			\nonumber \\
			\ge & \frac{(d^2-N)(d\Tr(O^2)-\Tr(O)^2)}{d^2(d+1)(d^2-1)} \cdot \mathbbm{1}(N\le d^2).
            \label{eq:RQ_final_LD_bound}
		\end{align}
        
		This completes the proof.
	\end{proof}

	\section{Proof of Theorem 3 (NFL theorem for quantum learning protocols)}\label{app_sec:NFL_SQ_dagger}
	As introduced in the main text, the training dataset for quantum learning protocols refers to $\{\ket{\bm{\psi}_j}, U^{\dagger}\ket{\bm{\psi}_j}\}_{j=1}^N$ and the assumption of perfect training is given by
	$|\bra{\bm{\psi}_j}UV_{\Q}^{\dagger}\ket{\bm{\psi}_j} |=1$ (i.e. $U^{\dagger}\ket{\bm{\psi}_j} = e^{i\bm{\beta}_j}V_{\Q}^{\dagger}\ket{\bm{\psi}_j}$) where $\bm{\beta}_j$ represents the inter-state relative phase between the true response state $U^{\dagger}\ket{\bm{\psi}_j}$ and the estimated response state $V_{\Q}^{\dagger}\ket{\bm{\psi}_j}$. 
	The NFL theorem for quantum learning protocols is encapsulated in the following theorem.
	\begin{theorem-non}[Formal statement of Theorem~\ref{thm:NFL_U_dagger_maintext}]
        Let $f_U(\psi)=\Tr(U^{\dagger}OU\ket{\bm{\psi}}\bra{\bm{\psi}})$ be the target concept, the observable $O$ be any Hermitian matrix and the training dataset $\mathcal{D}_{\Q}=\{\ket{\bm{\psi}_i}, \ket{\bm{\phi}_i} \}_{i=1}^N$ with $\ket{\bm{\phi}_i}=U^{\dagger}\ket{\bm{\psi}_i}$. Let $h_{V_{\Q}}$ be the learned hypothesis with the corresponding unitary operator $V_{\Q}$ satisfying the assumption of perfect training $U^{\dagger}\ket{\bm{\psi}_j} = e^{i\bm{\beta}_j}V_{\Q}^{\dagger}\ket{\bm{\psi}_j}$ where $\bm{\beta}_j$ is the inter-state relative phase. The averaged risk function over all unitaries $U$ and training datasets $\mathcal{D}_{\Q}$ yields
		\begin{equation}
			\mathbb{E}_{U,\mathcal{D}_{\Q}}R_U(V_{\Q})\ge \frac{1}{d^3(d+1)}\left((d-N)\cdot(d\Tr(O^2)-\Tr(O)^2)+(N^2-N) \sum_{k\ne j}^d O_{kj}^2 \cdot \mathbbm{1}(\bm{\beta} \not\varpropto \mathbf{1}_N)\right),
		\end{equation} 
		where the equality holds when the input states are non-orthogonal and linearly independent, $\bm{\beta}=(\bm{\beta}_1, \cdots, \bm{\beta}_N)$ and $\mathbf{1}_N$ refers to the all-ones vector of dimension $N$, $\mathbbm{1}(\bm{\beta}\varpropto \mathbf{1}_N)$ means that there exists a constant $c$ such that the phase vector $\bm{\beta}=c\mathbf{1}_N$, i.e., $\bm{\beta}_k=\bm{\beta}_j$ for any $j,k\in [N]$.
	\end{theorem-non}
	
	\begin{proof}[Proof of Theorem~\ref{thm:NFL_U_dagger_maintext}]
		Recall that the problem of deriving the lower bound of the average function could be reduced to obtaining the upper bound of the average term $\Tr(U^{\dagger}OUV_Q^{\dagger}OV_Q)$ as suggested by Lemma~\ref{lem:upper_to_lower}. In the following, we analyze the bound of such a term by separately considering the case of linearly independent and dependent states. 

        \smallskip
    
        \noindent \underline{\textit{The input states are linearly independent.}} When the training dataset refers to $\{\ket{\bm{\psi}_j}, U^{\dagger}\ket{\bm{\psi}_j}\}_{j=1}^N$ with $N\le d$ and $\{\ket{\bm{\psi}_j}\}_{j=1}^N$ are linearly independent, we could obtain the reduced representation about $UV_Q^{\dagger}$, In particular, utilizing Lemma~\ref{lem:U_dagger_V} yields
		\begin{equation}\label{eq:Q_perf_tr_mat}
			W:=UV_{\Q}^{\dagger}=\left(\begin{array}{ccc|c}e^{i \bm{\beta}_{1}} & \cdots & 0 & \\ \vdots & \ddots & & 0 \\ 0 & & e^{i \bm{\beta}_{N}} & \\ \hline & 0 & & \mathrm{Y}\end{array}\right) = \left( {\begin{array}{cc} P_{\bm{\beta}}^{(N)} & 0 \\ 0 & Y \\ \end{array} } \right) \mbox{ with denoting }   P_{\bm{\beta}}^{(N)} = \left( {\begin{array}{ccc}
					e^{i \bm{\beta}_{1}} & \cdots & 0 \\
					\vdots & \ddots &  \vdots \\
					0 & \cdots & e^{i \bm{\beta}_{N}} \\
			\end{array} } \right).
		\end{equation}
		where $Y\in \mathbb{C}^{d-N}$ is assumed to be a random Haar unitary matrix. Notably, the phase-related term $e^{i\bm{\beta}_j}$ does not need to be located together in a block, while this term $e^{i\bm{\beta}_j}$ should be uniformly distributed across the diagonal elements of $W$ in the context of NFL theorem considering the average prediction error over all possible training datasets $\mathcal{D}_{Q}$. We will keep this in mind in our subsequent discussions.
        
        We now evaluate an average of the risk function $R_{U}(V_Q)$ in Eqn.~(\ref{eq:app_risk_simplify}) over both the training dataset $\mathcal{D}_{\Q}$ and the target unitary $U$.
		Denote $O=\left( {\begin{array}{cc}
				O_{\mathcal{D}\mathcal{D}} & O_{\mathcal{D}\mathcal{D}_c} \\
				O_{\mathcal{D}_c\mathcal{D}} & O_{\mathcal{D}_c\mathcal{D}_c} \\
		\end{array} } \right)$ where the subscripts $\mathcal{D}$ and $\mathcal{D}_c$ referring to the location of submatrix $O_{[\cdot]}$ which is determined by the location of the phase matrix $P_{\bm{\beta}}^{(N)}$ in the matrix $W$, we have 
		\begin{align}
			& \mathbb{E}_{\mathcal{D}}\int_{\haar}\mathrm{d}W \Tr\left(OUV_{\Q}^{\dagger}OV_{\Q}U^{\dagger}\right)    \nonumber \\
			= & \mathbb{E}_{\mathcal{D}} \int_{\haar}\mathrm{d}Y \Tr\Big( O_{\mathcal{D}\mathcal{D}}P_{\bm{\beta}}^{(N)}O_{\mathcal{D}\mathcal{D}}(P_{\bm{\beta}}^{(N)})^{\dagger} + 	O_{\mathcal{D}\mathcal{D}_c}YO_{\mathcal{D}_c\mathcal{D}}(P_{\bm{\beta}}^{(N)})^{\dagger} + O_{\mathcal{D}_c\mathcal{D}}P_{\bm{\beta}}^{(N)}O_{\mathcal{D}\mathcal{D}_c}Y^{\dagger} + O_{\mathcal{D}_c\mathcal{D}_c}YO_{\mathcal{D}_c\mathcal{D}_c}Y^{\dagger}\Big) \nonumber \\
			= & \mathbb{E}_{\mathcal{D}} \Tr\left( 	O_{\mathcal{D}\mathcal{D}}P_{\bm{\beta}}^{(N)}O_{\mathcal{D}\mathcal{D}}(P_{\bm{\beta}}^{(N)})^{\dagger}\right)+\int_{\haar}\mathrm{d}Y  \Tr\left(O_{\mathcal{D}_c\mathcal{D}_c}YO_{\mathcal{D}_c\mathcal{D}_c}Y^{\dagger}\right) \cdot \mathbbm{1}(N<d)  \nonumber \\
			= & \mathbb{E}_{\mathcal{D}}\sum_{j=1}^N (O_{\mathcal{D}\mathcal{D}}^{(j,j)})^2 +\mathbb{E}_{\mathcal{D}}\mathbb{E}_{\bm{\beta}} \sum_{k\ne j}^N e^{i(\bm{\beta}_k-\bm{\beta}_j)} (O_{\mathcal{D}\mathcal{D}}^{(k,j)})^2 + \mathbb{E}_{\mathcal{D}}\frac{\Tr(O_{\mathcal{D}_c\mathcal{D}_c})^2}{d-N} \cdot 	\mathbbm{1}(N<d)   
            \nonumber \\
            = & \mathbb{E}_{\mathcal{D}} \left(\sum_{j=1}^N (O_{\mathcal{D}\mathcal{D}}^{(j,j)})^2 + \sum_{k\ne j }^N  (O_{\mathcal{D}\mathcal{D}}^{(k,j)})^2 \cdot \mathbbm{1}(\bm{\beta} \varpropto \mathbf{1}_N) + \frac{\Tr(O_{\mathcal{D}_c\mathcal{D}_c})^2}{d-N} \cdot 	\mathbbm{1}(N<d)  \right) 
            \nonumber \\
            = & \mathbb{E}_{\mathcal{D}} \left(\Tr\left((O_{\mathcal{D}\mathcal{D}})^2\right) \cdot \mathbbm{1}(\bm{\beta} \varpropto \mathbf{1}_N) + \sum_{j=1}^N (O_{\mathcal{D}\mathcal{D}}^{(j,j)})^2  \cdot \mathbbm{1}(\bm{\beta} \not\varpropto \mathbf{1}_N)+ \frac{\Tr(O_{\mathcal{D}_c\mathcal{D}_c})^2}{d-N} \cdot 	\mathbbm{1}(N<d) \right)
            \nonumber \\
            = & \frac{N\Tr(O^2)}{d}  - \frac{N^2-N}{d^2-d} \sum_{k\ne j}^d O_{kj}^2 \cdot \mathbbm{1}(\bm{\beta} \not\varpropto \mathbf{1}_N)+ \frac{(d-N)\Tr(O)^2}{d^2}
            \label{eq:two_stage_Tr_bound}
		\end{align}
		where the expectation of training dataset $\mathcal{D}$ is over all ingredients related to $\mathcal{D}$ including the relative phase $e^{i\bm{\beta}_j}$ and its location in matrix $P_{\bm{\beta}}^{(N)}$, the first equality uses the representation of $UV_{\Q}^{\dagger}$ given in Eqn.~(\ref{eq:Q_perf_tr_mat}), the second equality follows the fact that the Haar measure is left- and right-invariant under the action of the unitary group (in particular, under $-\mathbb{I}$), i.e., $\int\Tr(YX)\mathrm{d} Y =\int\Tr(YX)\mathrm{d} (-\mathbb{I}_NY)$ for any fixed unitary matrix $X$, indicating that $\int\Tr(YX)\mathrm{d} Y=0$ and hence $\int\Tr(O_{\mathcal{D}\mathcal{D}_c}YO_{\mathcal{D}_c\mathcal{D}}(P_{\bm{\beta}}^{(N)})^{\dagger})\mathrm{d} Y =0$ with taking $X=O_{\mathcal{D}_c\mathcal{D}}(P_{\bm{\beta}}^{(N)})^{\dagger})O_{\mathcal{D}\mathcal{D}_c}$, the same result holds for $\int\Tr(O_{\mathcal{D}_c\mathcal{D}}P_{\bm{\beta}}^{(N)}O_{\mathcal{D}\mathcal{D}_c}Y^{\dagger}\mathrm{d} Y)$, the term $\mathbbm{1}(N<d)$ follows that when $N=d$, $O_{\mathcal{D}\mathcal{D}}=O$ and the terms involving $O_{\mathcal{D}_c\mathcal{D}_c}$ disappears, the third equality employs Property~\ref{prop:Tr(WW)}  on the $(d-N)$-dimensional Haar unitary $Y$, the fourth equality follows the same derivation of Eqn.~(\ref{eq:T_1}) by considering whether the distribution of input states enable $\bm{\beta}_k=\bm{\beta}_j$ for any $k,j\in [N]$, the lase equality follows that the average on the location of the sub-matrix $O_{[\cdot]}$ related to $\mathcal{D}$ and $\mathcal{D}_c$ yields that $\mathbb{E}_{\mathcal{D}}\Tr(O_{\mathcal{D}\mathcal{D}}^2)=N\Tr(O^2)/d$, and $\mathbb{E}_{\mathcal{D}}\Tr(O_{\mathcal{D}_c\mathcal{D}_c})=(d-N)\Tr(O)/d$. In particular, considering the average of all possible locations of $O_{\mathcal{D}\mathcal{D}}$ is because the relative phase $e^{i\bm{\beta}_j}$ in the matrix $P_{\bm{\beta}}^{(N)}$ is uniformly located across the diagonal parts of the matrix $W$ defined in Eqn.~(\ref{eq:Q_perf_tr_mat}) in the context of NFL. 
  
        In conjunction with Eqn.~(\ref{eq:app_risk_simplify}) in Lemma~\ref{lem:upper_to_lower} and Eqn.~(\ref{eq:two_stage_Tr_bound}), we have
		\begin{align}
			\mathbb{E}_{U,\mathcal{D}_{\Q}}R_U(V_{\Q})= &	\frac{1}{d(d+1)}\left(\Tr(O^2)- \frac{N\Tr(O^2)}{d}  + \frac{N^2-N}{d^2-d} \sum_{k\ne j}^d O_{kj}^2 \cdot \mathbbm{1}(\bm{\beta} \not\varpropto \mathbf{1}_N) -  \frac{(d-N)\Tr(O)^2}{d^2}\right)
            \nonumber \\
            = & \frac{1}{d^3(d+1)}\left((d-N)\cdot(d\Tr(O^2)-\Tr(O)^2)+(N^2-N) \sum_{k\ne j}^d O_{kj}^2 \cdot \mathbbm{1}(\bm{\beta} \not\varpropto \mathbf{1}_N)\right)
            \label{eq:Q_bound_LD}
		\end{align}

        \smallskip

        \textbf{The condition enabling $\mathbbm{1}(\bm{\beta}  \varpropto \mathbf{1}_N)$.}
		We now analyze when the condition $\bm{\beta}_k=\bm{\beta}_j$ holds. Specifically, for any training data $(\ket{\bm{\psi}_i}, U^{\dagger}\ket{\bm{\psi}_i})$ and $(\ket{\bm{\psi}_j}, U^{\dagger}\ket{\bm{\psi}_j})$ with $i\ne j$, we have
		\begin{equation}\label{eq:Q_phase_equality}
			\braket{\bm{\psi}_k | \bm{\psi}_j} = 	\bra{\bm{\psi}_k}UU^{\dagger}\ket{\bm{\psi}_j} = e^{i(\bm{\beta}_j-\bm{\beta}_k)}  \bra{\bm{\psi}_k}V_{\Q} V_{\Q}^{\dagger}\ket{\bm{\psi}_j} = e^{i(\bm{\beta}_j-\bm{\beta}_k)} \braket{\bm{\psi}_k | \bm{\psi}_j}
		\end{equation}
		where the second equality exploits the assumption of perfect training $U^{\dagger}\ket{\bm{\psi}_j}= e^{i\bm{\beta}_j}V_{\Q}^{\dagger} \ket{\bm{\psi}_j}$. Eqn.~(\ref{eq:Q_phase_equality}) implies that $\bm{\beta}_{k}=\bm{\beta}_j$ for all $k,j\in [N]$ if $\braket{\bm{\psi}_k|\bm{\psi}_j}\ne 0$ or equivalently the input states are not orthogonal.

        \smallskip
        
        \noindent \underline{\textit{The input states are linearly dependent.}} 
        We note that the cases of $\mathcal{D}_{\Q}$ consisting of linearly dependent input states could be reduced to the case of $\widetilde{N}$ linearly independent states with $\widetilde{N}<N$. In this regard, the achieved bound in Eqn.~(\ref{eq:Q_bound_LD}) also holds by replacing $N$ with $\widetilde{N}$.
	\end{proof}

\end{document}